\documentclass[runningheads]{llncs}
\usepackage{graphicx}
\usepackage{hyperref}

\usepackage[dvipsnames]{xcolor}
\usepackage{multicol}
\usepackage[frozencache,cachedir=.]{minted}
\usepackage{courier}
\usepackage{latexsym}
\usepackage{amsmath,ascmac,amssymb,thmtools} 
\usepackage{stmaryrd}
\usepackage{graphicx}
\usepackage{mydef}
\usepackage[normalem]{ulem}
\usepackage{minitoc}
\usepackage{fontawesome5}

\usepackage{bcprules} % 
\usepackage{bussproofs} % 

\usepackage{tikz}
\usetikzlibrary{positioning, arrows, arrows.meta, decorations.markings, shapes.geometric, shapes.misc, fit, calc}

\definecolor{LightGray}{gray}{0.95}
\newcommand{\vlkey}[1]{\textcolor{OliveGreen}{\textbf{#1}}}

\newif\ifdraftComments
\draftCommentstrue
\def\mkDraftFn#1#2{%
  \expandafter\def\csname #1\endcsname##1{\ifdraftComments\textcolor{#2}{[#1: ##1]}\marginpar[$\longrightarrow$]{$\longleftarrow$}\fi}%
}
\mkDraftFn{YT}{red}

\setminted[haskell]{frame=lines,
baselinestretch=1.1,
bgcolor=LightGray,
fontsize=\footnotesize,
escapeinside=@@,
}

\begin{document}
\title{Compilation Semantics for a Programming Language with Versions}

% \titlerunning{Compilation with Versions}

% \author{Anonymous}
\author{Yudai Tanabe\inst{1}\orcidID{0000-0002-7990-0989} \faIcon{envelope}\and\\
Luthfan Anshar Lubis\inst{2}\orcidID{0000-0002-1498-7788} \and\\
Tomoyuki Aotani\inst{3}\orcidID{0000-0003-4538-0230} \and\\
Hidehiko Masuhara\inst{2}\orcidID{0000-0002-8837-5303}
}

\authorrunning{Y. Tanabe et al.}

% \institute{Anonymous}
\institute{
Kyoto University, Kyoto, Japan
\faIcon{envelope} \footnote{\faIcon{envelope} Corresponding author}\email{yudaitnb@fos.kuis.kyoto-u.ac.jp}\\\and
Tokyo Institute of Technology, Tokyo, Japan
\email{\{luthfanlubis@prg.is.titech.ac.jp | masuhara@acm.org\}}\\\and
Sanyo-Onoda City University, Yamaguchi, Japan
\email{aotani@rs.socu.ac.jp}
}
\maketitle            % typeset the header of the contribution
\begin{abstract}
\emph{Programming with versions} is a paradigm that allows a program to use multiple versions of a module so that the programmer can selectively use functions from both older and newer versions of a single module.
Previous work formalized \corelang{}, a core calculus for programming with versions, but it has not been integrated into practical programming languages.
In this paper, we propose \mylang{}, a Haskell-subset surface language for \corelang{} along with its compilation method.
We formally describe the core part of the \mylang{} compiler, which translates from the surface language to the core language by leveraging Girard's translation, soundly infers the consistent version of expressions along with their types, and generates a multi-version interface by bundling specific-version interfaces. 
We conduct a case study to show how \mylang{} supports practical software evolution scenarios and discuss the method's scalability.
\end{abstract}
\section{Introduction}
\label{introduction}
Updating dependent software packages is one of the major issues in software development. Even though a newer version of a package brings improvements, it also brings the risk of breaking changes, which can make the entire software defective.

We argue that this issue originates from the principle of most programming languages that only allow the use of one version of a package at a time. Due to this principle, developers are faced with the decision to either update to a new, improved version of a package that requires many changes or to remain with an older version. The problem gets worse when a package is indirectly used. This dilemma often results in delays in adopting upgrades, leading to stagnation in software development and maintenance~\cite{6676878,10.1007/s10664-014-9325-9}.

\emph{Programming with versions}~\cite{Tanabe:2018:CPA:3242921.3242923,Tanabe_2021,batakjava2022,bundling2023} is a recent proposal that allows programming languages to support multiple versions of programming elements at a time so that the developer can flexibly cope with incompatible changes. \corelang{} is the core calculus in which a \emph{versioned value} encapsulates multiple versions of a value (including a function value). The \corelang{} type system checks the consistency of each term so that a value produced in a version is always passed to functions in the same version. The calculus and the type system design are based on coeffect calculus~\cite{brunel_core_2014,Orchard:2019:Granule}.

While \corelang{} offers the essential language constructs to support multiple versions in a program, the language is far from practical. For example, with  multiple versions of a module, each version of the function must be manually represented inside a versioned value (i.e., a record-like expression). \corelang{} is as simple as lambda calculus, yet it has a verbose syntax due to the coeffect calculus. In short, there are aspects of versioning in \corelang{} that a surface language compiler can automate.

We propose the functional language \mylang{} as a surface language for \corelang{} along with its compilation method. In \mylang{}, a function name imported from an external module represents a multi-version term, where each occurrence of the function name can reference a different version of the function. The \mylang{} compiler translates a program into an intermediate language \vlmini{}, a version-label-free variant of \corelang{}, determines the version for each name occurrence based on a type and version inference algorithm, and translates it back into a version-specialized Haskell program.
\mylang{} also offers the constructs to explicitly control versions of expressions, which are useful to keep using an older version for some reason.  

This paper presents the following techniques in \mylang{}: (a) \emph{an application of Girard's translation} for translating \mylang{} into \vlmini{}, (b) \emph{the bundling} for making a top-level function act as a versioned value, and (c) \emph{a type and version inference algorithm} for identifying the version of each expression with respect to the \corelang{} type system. Finally, we prove the soundness of the inference system and implement a \mylang{} compiler. Code generation converts a \mylang{} program into a version-specialized Haskell program using the solution obtained from z3~\cite{10.1007/978-3-540-78800-3_24}.

\vspace{-0.5\baselineskip}
\paragraph{Paper Organization.}
Section \ref{overview} introduces incompatibility issues and fundamental concepts in programming with versions with \corelang{} and \mylang{}.
Section \ref{compilation} introduces bundling and Girard's transformation.
Section \ref{inference} presents an algorithmic version inference for \mylang{}.
Section \ref{implementation} features an implementation of \mylang{}, and Section \ref{casestudy} introduces a case study that simulates an incompatible update made in a Haskell library.
Finally, Section \ref{conclusion} discusses further language development and concludes the paper by presenting related work and a conclusion.
% Section Template
\section{Overview} % Main Section title
\label{overview}
\subsection{Motivating Example}
\label{sec:motivatingexample}
\begin{figure}[tb]
\centering
% \begin{minipage}[t]{0.49\textwidth}
\begin{tikzpicture}[thick, node distance=0.8cm and 1cm]
\tikzset{codeblock/.style={
    % inner xsep=20pt,
    outer xsep=-7pt,
    outer ysep=-7pt,
    rectangle,
    % fill=gray!10,
    % text width=3cm,
    text centered,
    rounded corners,
    minimum width=5.8cm
}};

\node[codeblock] (app1) {
\begin{minipage}{.44\linewidth}
% \begin{lstlisting}[style=haskell]
\begin{minted}[linenos,numbersep=2pt,highlightcolor=white,highlightlines={6-7}]{haskell}
module App where
import Dir
import Hash
main () =
  let s = getArg ()
      digest = mkHash s in
  if exists digest
    then print "Found"
    else error "Not found"
\end{minted}
% \end{lstlisting}
\end{minipage}
% \hspace*{-15pt}
};

\node[codeblock, right=of app1] (dir1) {
\begin{minipage}{.42\linewidth}
% \begin{lstlisting}[style=haskell]
\begin{minted}{haskell}
module Dir where
import Hash
-- version 1.0.0
exists hash = 
  let fs = getFiles () in
  foldLeft
    (\(acc, f) ->
      acc || match f hash)
    false fs
\end{minted}
% \end{lstlisting}
\end{minipage}
% \hspace*{-15pt}
};

\node[codeblock, below right=of app1] (hash1) {
\begin{minipage}{.42\linewidth}
% \begin{lstlisting}[style=haskell]
\begin{minted}{haskell}
module Hash where
-- version 1.0.0
mkHash s = {- MD5 hash -}
match s hash =
    mkHash s == hash
\end{minted}
% \end{lstlisting}
\end{minipage}
% \hspace*{10pt}
};

\node[codeblock, below=of app1] (hash2) {
\begin{minipage}{.44\linewidth}
% \begin{lstlisting}[style=haskell]
\begin{minted}{haskell}
module Hash where
-- version 2.0.0
mkHash s = {- SHA-3 hash -}
match s hash =
    mkHash s == hash
\end{minted}
% \end{lstlisting}
\end{minipage}
% \hspace*{10pt}
};

\coordinate [below = 0.4 of dir1] (B) {};

\draw[->] (app1.east)  -- node[midway,above] {} (dir1.west);
% \draw[->] ([yshift=-0.5cm]app1.east)  -- node[midway,below] {version 2} ([yshift=-0.5cm]dir1.west);
\draw[->,red] (app1.south) -- node[midway,left] {\textcolor{red}{version 2.0.0}}(hash2.north);
% node[midway, draw=red, shape=cross out,inner sep=7pt, outer sep=0pt] {}
\draw[dashed,->] ([xshift=0.5cm]app1.south) |- ([xshift=-0.5cm]B) -- ([xshift=-0.5cm]hash1.north);
% \draw[->] (app1.south) -- node[midway,above] {version 2} (hash2.north);
\draw[->] (dir1.south) -- node[midway,right] {version 1.0.0} (hash1.north);
% \draw[line width=1pt, double distance=3pt,
%              arrows = {-Latex[length=0pt 3 0]}] (hash1.west) -- node[midway,above] {Update} (hash2.east);

\end{tikzpicture}

\parbox[t]{.9\linewidth}{\small The \fn{exists} provided from \mn{Dir} (which depends on version $1$ of \mn{Hash}) expects an MD5 hash as an argument. However, after the dependency update of \mn{App} on \mn{Hash}, the value assigned to \texttt{digest} is a SHA-3 hash.}

\caption{Minimal module configuration before and after the dependency update  causing an error due to inconsistency expected to the dependent package.}
\label{fig2-1}
% \end{minipage}
\end{figure}
First, we will explain a small example to clarify incompatibility issues. Consider a scenario where an incompatible change is made to a dependent package. Figure \ref{fig2-1} shows the package dependencies in a file explorer \mn{App} based on a hash-based file search. This function is developed using the system library \mn{Dir} and the cryptography library \mn{Hash}. For simplicity, we equate packages and modules here (each package consists of a single module), and we only focus on the version of \mn{Hash}. The pseudocode is written in a Haskell-like language.

Before its update, \mn{App} depends on version 1.0.0 of \mn{Hash} (denoted by $\dashrightarrow$). The \mn{App}'s \fn{main} function implements file search by a string from standard input using \fn{mkHash} and \fn{exists}.
The function \fn{mkHash} is in version 1.0.0 of \mn{Hash}, and it generates a hash value using the MD5 algorithm from a given string. \mn{Hash} also provides a function \fn{match} that determines if the argument string and hash value match under \fn{mkHash}.
The function \fn{exists} is in version 1.0.0 of \mn{Dir}, which is also dependent on version 1.0.0 of \mn{Hash}, and it determines if a file with a name corresponding to a given hash exists. 

Due to security issues, the developer of \mn{App} updated \mn{Hash} to version 2.0.0 (denoted by \textcolor{red}{$\longrightarrow$}). In version 2.0.0 of \mn{Hash}, SHA-3 is adopted as the new hash algorithm. Since \mn{Dir} continues to use version 1.0.0 of \mn{Hash}, \mn{App} needs two different versions of \mn{Hash}. 
Various circumstances can lead to this situation: \mn{Dir} may have already discontinued maintenance, or functions in \mn{Dir}, other than \fn{exists}, might still require the features provided by version 1.0.0 of \mn{Hash}.

Although the update does not modify \mn{App}, it causes errors within \mn{App}. Even if a file with an input filename exists, the program returns \texttt{Not Found} error contrary to the expected behavior. The cause of the unexpected output lies in the differences between the two versions required for \fn{main}. In line 6 of \mn{App}, an SHA-3 hash value is generated by \fn{mkHash} and assigned to \texttt{digest}. Since \fn{exists} evaluates hash equivalence using MD5, \texttt{exists digest} compares hashes generated by different algorithms, evaluating to \texttt{false}.

This example highlights the importance of version compatibility when dealing with functions provided by external packages. Using different versions of \mn{Hash} in separate program parts is fine, but comparing results may be semantically incorrect. Even more subtle changes than those shown in Figure \ref{fig2-1} can lead to significant errors, especially when introducing side effects or algorithm modifications that break the application's implicit assumptions. Manually managing version compatibility for all external functions is unfeasible.

In practical programming languages, dependency analysis is performed before the build process to prevent such errors, and package configurations requiring multiple versions of the same package are rejected. However, this approach tends towards conservative error reporting.
In cases where a core package, which many other libraries depend on, receives an incompatible change, no matter how minuscule, it requires coordinated updates of diverse packages across the entire package ecosystem~\cite{10.1007/s10664-014-9325-9,Tanabe_2021,semvertrick}.

\subsection{\corelang{}}
\label{core}
\corelang{}~\cite{Tanabe:2018:CPA:3242921.3242923,Tanabe_2021} is a core calculus designed to follow the principles: (1) enabling simultaneous usage of multiple versions of a package, (2) ensuring version consistency within a program. \corelang{} works by encapsulating relevant terms across multiple versions into a record-like term, tagged with a label indicating the specific module version. Record-like terms accessible to any of its several versions are referred to as \emph{versioned values}, and the associated labels are called \emph{version labels}.

% \vspace{-\baselineskip}
\subsubsection{Version Labels}
Figure \ref{syntax:lambdavl} shows the syntax of \corelang{}.
Given modules and their versions, the corresponding set of version labels characterizes the variation of programs of a versioned value. 
In \corelang{}, version labels are implicitly generated for all external module-version combinations, in which $M_i$ is unique, with the universal set of these labels denoted by $\mathcal{L}$. 
Specifically, in the example illustared in Figure \ref{fig2-1}, $\mathcal{L} = \{l_1,l_2\}$ and $l_1 = \{\mn{Hash} = 1.0.0,\,\mn{Dir} = 1.0.0\}, l_2 = \{\mn{Hash} = 2.0.0,\,\mn{Dir} = 1.0.0\}$.
The size of $\mathcal{L}$ is proportional to $V^M$ where $M$ is the number of modules and $V$ is the maximum number of versions.

\begin{figure}[tb]
\centering
\hrule
\medskip
\begin{minipage}{.9\linewidth}
\textbf{\corelang{} syntax}
\end{minipage}
\begin{minipage}{\textwidth}
\vspace{-.7\baselineskip}
\begin{align*}
t &::= n \mid x \mid \app{t_1}{t_2} \mid \lam{x}{t} \mid \clet{x}{t_1}{t_2} \mid  u.l \mid \ivval{\overline{l_i=t_i}}{l_k} \mid u \hspace{2.5em}\tag{terms}\\
\end{align*}
\end{minipage}
\begin{minipage}{\textwidth}
    \vspace{-2.2\baselineskip}
    \begin{minipage}{.5\textwidth}
    \begin{align*}
    u &::= \pr{t} \mid \nvval{\overline{l_i=t_i}} \tag{versioned values}\\ % \vval{\overline{l_i=t_i}}{l_j}
    A, B &::= \inttype \mid \ftype{A}{B} \mid \vertype{r}{A} \tag{types}
    \end{align*}
    \end{minipage}
    \begin{minipage}{.5\textwidth}
    \begin{align*}
    r    &::= \bot \mid \{\,\overline{l_i}\,\} \tag{resources}\\
    \mathcal{L}\ \ni\ l&::= \{\overline{M_i = V_i}\} \tag{version labels}
    \end{align*}
    \end{minipage}
\end{minipage}
\parbox[t]{\linewidth}{\smallskip
\small $M_i$ and $V_i$ are metavariables over module names and versions of $M_i$, respectively.}
\smallskip
\hrule
\caption{The syntax of \corelang{}.}
\label{syntax:lambdavl}
\end{figure}
% \vspace{-\baselineskip}
\subsubsection{Syntax of \corelang{}}
\corelang{} extends \lrpcf{}~\cite{brunel_core_2014} and GrMini~\cite{Orchard:2019:Granule} with additional terms that facilitate introducing and eliminating versioned values.
Versioned values can be introduced through versioned records $\nvval{\overline{l_i=t_i}}$ and promotions $\pr{t}$.
A versioned record encapsulates related definitions $t_1,\ldots,t_n$ across multiple versions and their version labels $l_1,\ldots,l_n$.
For instance, the two versions of \fn{mkHash} in Figure \ref{fig2-1} can be bundled as the following version record.
\begin{align*}
\mathit{mkHash}\quad:=\quad
&
\begin{aligned}
    \{&l_1=\lam{s}{\textnormal{\com{make MD5 hash}}},\\
    &l_2=\lam{s}{\textnormal{\com{make SHA-3 hash}}}\}
\end{aligned}
\end{align*}

In \corelang{}, programs are constructed via function application of versioned values. A function application of \fn{mkHash} to the string \texttt{s} can be written as follows.
\begin{align*}
\mathit{app}\quad:=\quad
\begin{aligned}
&\clet{\mathit{mkHash'}}{\mathit{mkHash}}{
    \\&\clet{s}{
        \pr{``\mathtt{compiler.vl}"}
    }{
        \pr{\app{\mathit{mkHash'}}{s}}
    }
}
\end{aligned}
\end{align*}

This program ($app$ hereafter) makes a hash for the string ``\texttt{compiler.vl}" and is available for both $l_1$ and $l_2$.
The contextual let-binding $\clet{x}{t_1}{t_2}$ provides the elimination of version values by binding a versioned value for $t_1$ to $x$, thus making it accessible in $t_2$.
Promotion $\pr{x}$ offers an alternative way to introduce versioned values, making any term $t$ act as a versioned value.

The evaluation of terms $t_i$ stored in a versioned value $\nvval{\overline{l_i=t_i}}$ and $\pr{t}$ is postponed until a specific version label is later specified.
To proceed with a postponed evaluation of a versioned value, we use extraction $u.l_k$. Extraction specifies one versioned label $l_k$ for the versioned value $u$ and recursively extracts the inner term $t_k$ corresponding to $l_k$ from $\nvval{l_i=t_i}$, and $t$ from $\pr{t}$ as follows.
\begin{align*}
\mathit{app}\#l_1\quad:=\quad
&
\begin{aligned}
&\clet{\mathit{mkHash'}}{\mathit{mkHash}}{
    \\&\clet{s}{
        \pr{``\mathtt{compiler.vl}"}
    }{
        \pr{\app{\mathit{mkHash'}}{s}}
    }
}.l_1
\end{aligned}\\
\longrightarrow^* \quad&\app{(\lam{s}{\textnormal{\com{make MD5 hash}}})}{``\mathtt{compiler.vl}"}\\
\longrightarrow^{\phantom{*}}\quad &\texttt{4dcb6ebe3c6520d1f57c906541cf3823}
\end{align*}

Consequently, $\mathit{app}\#l_1$ evaluates into an MD5 hash corresponding to $l_1$.
\vspace{-\baselineskip}
\subsubsection{Type of Versioned Values}
The type of a versioned value is expressed as \vertype{r}{A}, assigning a set of version labels $r$, called \textit{version resources}, to a type $A$. Intuitively, the type of a versioned value represents the versions available to that versioned value. For example, $\mathit{mkHash}$ and $\mathit{app}$ are typed as follows.
\begin{align*}
\mathit{mkHash}\,:\,\vertype{\{l_1,l_2\}}{\left(\ftype{\strtype{}}{\strtype{}}\right)}\quad
\mathit{app}\,:\,\vertype{\{l_1,l_2\}}{\left(\ftype{\strtype{}}{\strtype{}}\right)}
\end{align*}

The types have $\{l_1,l_2\}$ as their version resource, illustrating that the versioned values have definitions of $l_1$ and $l_2$.
For function application, the type system computes the intersection of the version resource of subterms.
Since the promoted term is considered to be available in all versions, the version resource of the entire function application indicates $\{l_1,l_2\} = \{l_1,l_2\} \cap \mathcal{L}$.

For extractions, the type system verifies if the version resource contains the specified version as follows.
\begin{align*}
\mathit{app}\#l_1\,:\,\ftype{\strtype{}}{\strtype{}} \quad
\mathit{app}\#{\textcolor{red}{l_3}}\,:\,\textcolor{red}{(rejected)}
\end{align*}

Assuming $\mathcal{L} = \{l_1,l_2,l_3\}$, $\mathit{app}\#{\textcolor{red}{l_3}}$ is rejected by type checking because the version resource of $\mathit{app}$ does not contain $l_3$. Conversely, $\mathit{app}\#l_1$ is well-typed, but note that the resultant type lost its version resource. It is attributed to the design principle that it could be used in other versions upon extraction.

The \corelang{} type system incorporates the notion of version consistency in addition to the standard notions of preservation and progress. Proofs of these theorems can be found in Appendix \ref{appendix:lambdavl_safety}.

\subsection{Programming with Versions in \mylang{}}
\begin{figure}[tb]
\centering
% \begin{minipage}[t]{0.49\textwidth}
\begin{tikzpicture}[thick, node distance=0.5cm and 1cm]
\tikzset{codeblock/.style={
    % inner xsep=20pt,
    outer xsep=-5pt,
    outer ysep=-7pt,
    rectangle,
    % fill=gray!10,
    % text width=3cm,
    text centered,
    rounded corners,
    minimum width=5.7cm
}};

\node[codeblock] (app1) {
\begin{minipage}{.50\linewidth}
% \begin{lstlisting}[style=haskell]
% \begin{minted}[linenos,numbersep=2pt,highlightcolor=white,highlightlines={7,9}]{haskell}
% module App where
% import Dir
% import Hash
% main () =
%   let s = getArg ()
%       digest = mkHash s
%       d = @\vlkey{ver}@ [@\textcolor{blue}{\textbf{Hash}}@=1.0.0] of
%            digest in
%   if @\vlkey{unversion}@ (exists d)
%     then print "Found"
%     else error "Not found"
% \end{minted}
\begin{minted}[linenos,numbersep=2pt,highlightcolor=white,highlightlines={7}]{haskell}
module App where
import Dir
import Hash
main () =
  let s = getArg ()
      digest = mkHash s in
  if exists digest
    then print "Found"
    else error "Not found"
\end{minted}
% \end{lstlisting}
% \hspace*{5pt}
\end{minipage}
};

\node[codeblock, right=of app1] (dir1) {
\begin{minipage}{.42\linewidth}
% \begin{lstlisting}[style=haskell]
\begin{minted}{haskell}
module Dir where
-- version 1.0.0
import Hash
exists hash = 
  let fs = getFiles () in
  foldLeft
    (\(acc, f) ->
      acc || match f hash)
    false fs
\end{minted}
% \end{lstlisting}
\end{minipage}
% \hspace*{-15pt}
};

\node[codeblock, below right=of app1, yshift=-5] (hash1) {
\begin{minipage}{.42\linewidth}
% \begin{lstlisting}[style=haskell]
\begin{minted}{haskell}
module Hash where
-- version 1.0.0
mkHash s = {- MD5 hash -}
match s hash =
    mkHash s == hash
\end{minted}
% \end{lstlisting}
\end{minipage}
% \hspace*{10pt}
};

\node[codeblock, below=of app1, yshift=-5] (hash2) {
\begin{minipage}{.44\linewidth}
% \begin{lstlisting}[style=haskell]
\begin{minted}{haskell}
module Hash where
-- version 2.0.0
mkHash s = {- SHA-3 hash -}
match s hash =
    mkHash s == hash
\end{minted}
% \end{lstlisting}
\end{minipage}
% \hspace*{10pt}
};

\node[draw, fit=(dir1), label={above:Bundled \mn{Dir}}] (Dir) {};
\node[draw, fit=(hash1) (hash2), label={above:Bundled \mn{Hash}}] (Hash) {};

\draw[->] ([xshift=2]app1.east)  -- node[midway,below] {} (Dir.west);
\draw[->] ([yshift=5]app1.south) -- node[midway,above] {} ([yshift=5]hash2.north);
\draw[->] ([yshift=5]dir1.south) -- node[midway,right] {} ([yshift=5]hash1.north);

\end{tikzpicture}
\parbox[t]{.9\linewidth}{\small The versions of each external module are bundled. Programs using a bundled module can refer to the definitions of all versions of the bundled module.}

\caption{The programs in Figure \ref{fig2-1} in \mylang{}.}
\label{fig2-3}
% \end{minipage}
\end{figure}
Our contributions enjoy the benefits of programming with versions on a $\lambda$-calculus-based functional language \mylang{}. To achieve this, we develop a compilation method between lambda calculus and \vlmini{}, a version-label free variant of \corelang{}, and a version inference algorithm to infer the appropriate version of expressions.

In \mylang{}, (1) all versions are available for every module, and (2) the version of each expression is determined by expression-level dependency analysis. This approach differs from existing languages that determine one version for each dependent package.
Figure \ref{fig2-3} shows how the programs in Figure \ref{fig2-1} are interpreted in \mylang{}.
The \mylang{} compiler bundles the interfaces of multiple versions and generates a cross-version interface to make external functions available in multiple versions.
The \mylang{} type system enforces version consistency in \fn{main} and selects a newer version if multiple versions are available. Thus it gives the version label $\{\mn{Hash} = 2.0.0,\,\mn{Dir} = 1.0.0\}$ to dependent expressions of \fn{main}. As a result, since \mn{Hash} version referenced from \mn{Dir} is no longer limited to 1.0.0, \texttt{exists digest} is evaluated using SHA-3 under the context of \mn{Hash} version 2.0.0.

Furthermore, \mylang{} provides \emph{version control terms} to convey the programmer's intentions of versions to the compiler. For example, to enforce the evaluation in Figure \ref{fig2-3} to MD5, a programmer can rewrite line 7 of \mn{App} as follows.

\begin{minted}[linenos,numbersep=2pt,highlightcolor=white,firstnumber=7,highlightlines={7}]{haskell}
  if @\vlkey{ver}@ [@\textcolor{blue}{\textbf{Hash}}@=1.0.0] of (exists digest)
\end{minted}
The program dictates that \texttt{exists digest} is evaluated within the context of the \mn{Hash} version 1.0.0. Consequently, both \fn{mkHash} and \fn{match}, which depend on \texttt{exists digest}, are chosen to align with version 1.0.0 of \mn{Hash}.
Moreover, \mylang{} provides \texttt{\vlkey{unversion} t}. It eliminates the dependencies associated with term \texttt{t}, facilitating its collaboration with other versions under the programmer's responsibility, all while maintaining version consistency within its subterm.
Thus, \mylang{} not only ensures version consistency but also offers the flexibility to control the version of a particular part of the program.

% Section Template

\section{Compilation} % Main Section title
\label{compilation} % Change X to a consecutive number; for referencing this Section elsewhere, use \ref{SectionX}

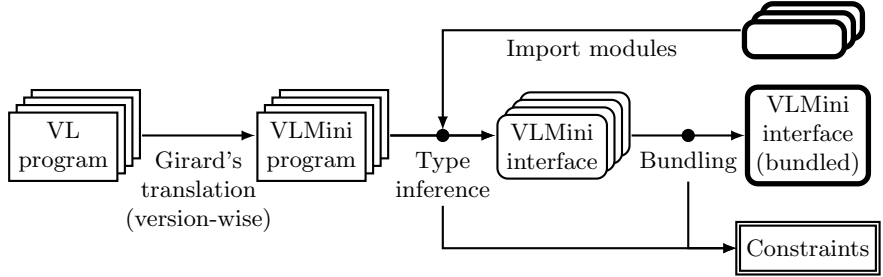
\begin{figure}[tb]
\centering
\begin{tikzpicture}[thick]
    \tikzset{vecArrow/.style={thick, decoration={markings,mark=at position
    1 with {\arrow[semithick]{open triangle 60}}},
    double distance=1.4pt, shorten >= 5.5pt,
    preaction = {decorate},
    postaction = {draw,line width=1.4pt, white,shorten >= 4.5pt}}};
    \tikzset{innerWhite/.style={semithick, white,line width=1.4pt, shorten >= 4.5pt}};
    \tikzset{Package/.style={rectangle, fill=cyan!10, text centered, rounded corners, minimum width=1.9cm, minimum height=0.65cm}};
    \tikzset{App/.style={rectangle, text centered, minimum width=1.2cm, minimum height=0.65cm}};

    \newlength\mywidth
    \newlength\myheight
    \newlength\tempdimen
    
    \newcommand\getdimensions[1]{
        \pgfextractx{\mywidth}{\pgfpointanchor{#1}{east}}
        \pgfextractx{\tempdimen}{\pgfpointanchor{#1}{west}}
        \addtolength{\mywidth}{-\tempdimen}
        \pgfextracty{\myheight}{\pgfpointanchor{#1}{north}}
        \pgfextracty{\tempdimen}{\pgfpointanchor{#1}{south}}
        \addtolength{\myheight}{-\tempdimen}
    }

    \tikzset{
      pics/stacked rectangles/.style n args={4}{
        code={
          \def\rectangleWidth{2cm}  % 長方形の幅を固定
          \def\rectangleHeight{1cm} % 長方形の高さを固定
          \pgfmathsetmacro\offsetY{(#1+1)*#2/2} % Yのオフセットを計算
          
          % 一時的なrect1を描画してサイズを取得
          \node[draw=none, align=center] (tempRect1) at (1*#2,1*#2-\offsetY) { #3 };
          
          % ここでtempRect1の寸法を取得
          \getdimensions{tempRect1}
          
          % 長方形を描画
          \foreach \i in {#1,...,2} {
            \node[draw, fill=white, minimum width=\mywidth, minimum height=\myheight] (rect\i) at (\i*#2,\i*#2-\offsetY) {};
          }

          % 一番前の長方形を描画
          \node[draw, fill=white, minimum width=\mywidth, minimum height=\myheight, align=center] (rect1) at (1*#2,1*#2-\offsetY) { #3 };
    
          % すべての長方形を囲む仮想的なノード
          \node[draw=none, inner sep=0pt, fit={(rect1.south west) (rect#1.north east)}, name=#4] {};
        }
      }
    }

    \tikzset{
      pics/stackedty rectangles/.style n args={4}{
        code={
          \def\rectangleWidth{2cm}  % 長方形の幅を固定
          \def\rectangleHeight{1cm} % 長方形の高さを固定
          \pgfmathsetmacro\offsetY{(#1+1)*#2/2} % Yのオフセットを計算
          
          % 一時的なrect1を描画してサイズを取得
          \node[draw=none, align=center] (tempRect1) at (1*#2,1*#2-\offsetY) { #3 };
          
          % ここでtempRect1の寸法を取得
          \getdimensions{tempRect1}
          
          % 長方形を描画
          \foreach \i in {#1,...,2} {
            \node[draw, fill=white, minimum width=\mywidth, minimum height=\myheight, rounded corners] (rect\i) at (\i*#2,\i*#2-\offsetY) {};
          }

          % 一番前の長方形を描画
          \node[draw, fill=white, minimum width=\mywidth, minimum height=\myheight, align=center, rounded corners] (rect1) at (1*#2,1*#2-\offsetY) { #3 };
    
          % すべての長方形を囲む仮想的なノード
          \node[draw=none, inner sep=0pt, fit={(rect1.south west) (rect#1.north east)}, name=#4] {};
        }
      }
    }

    \tikzset{
      pics/stackedtybundled rectangles/.style n args={4}{
        code={
          \def\rectangleWidth{2cm}  % 長方形の幅を固定
          \def\rectangleHeight{1cm} % 長方形の高さを固定
          \pgfmathsetmacro\offsetY{(#1+1)*#2/2} % Yのオフセットを計算
          
          % 一時的なrect1を描画してサイズを取得
          \node[draw=none, align=center] (tempRect1) at (1*#2,1*#2-\offsetY) { #3 };
          
          % ここでtempRect1の寸法を取得
          \getdimensions{tempRect1}
          
          % 長方形を描画
          \foreach \i in {#1,...,2} {
            \node[draw, fill=white, minimum width=\mywidth, minimum height=\myheight, rounded corners, line width=2pt] (rect\i) at (\i*#2,\i*#2-\offsetY) {};
          }

          % 一番前の長方形を描画
          \node[draw, fill=white, minimum width=\mywidth, minimum height=\myheight, align=center, rounded corners, line width=2pt] (rect1) at (1*#2,1*#2-\offsetY) { #3 };
    
          % すべての長方形を囲む仮想的なノード
          \node[draw=none, inner sep=0pt, fit={(rect1.south west) (rect#1.north east)}, name=#4] {};
        }
      }
    }

    \pic at (0,0) {stacked rectangles={4}{0.1cm}{\mylang{}\\program}{vl1}};
    % \node[draw,App, align=center] at (3.0,0) (vlmini1) {\vlmini{}\\($V_{M_i}$)};
    \pic at (3.3,0) {stacked rectangles={4}{0.1cm}{\vlmini{}\\program}{vlmini1}};
    % \node[draw,App, align=center] at (6.0,0) (vlmini2) {\vlmini{}\\interface};
    \pic at (6.5,0) {stackedty rectangles={4}{0.1cm}{\vlmini{}\\interface}{vlmini2}};
    \pic at (9.7,1.4) {stackedtybundled rectangles={3}{0.12cm}{\vphantom{b}\hspace{3em}}{vlmini4}};
    \node[draw,App, align=center, rounded corners, line width=2pt] at (10,0) (vlmini3) {\vlmini{}\\interface\\(bundled)};
    % \node[draw,App, align=center, rounded corners, line width=2pt, minimum width=1.62cm, minimum height=0.5cm] at (10,1.4) (vlmini4) {};
    \coordinate (topofinference) at ($(vlmini1.east)!0.5!(vlmini2.west)$);
    \coordinate (botofinference) at ($(vlmini1.east)!0.5!(vlmini2.west) + (0, -1.75)$);
    \coordinate (topofbundling) at ($(vlmini2.east)!0.5!(vlmini3.west)$);
    \coordinate (botofbundling) at ($(vlmini2.east)!0.5!(vlmini3.west) + (0, -1.75)$);
    % \node[draw,App,double,align=center] at (botofinference) (vdep) {\footnotesize{Variable}\\\footnotesize{dependency}};
    % \node[draw,App,double,align=center] at (botofbundling) (ldep) {\footnotesize{Label}\\\footnotesize{dependency}};
    \node[draw,App,double] at (10,-1.5) (constraints) {Constraints};
    \fill (topofinference) circle (2.5pt);
    \fill (topofbundling) circle (2.5pt);
    
    \draw[-latex] (vl1.east) to [yshift=-5]node[midway,below,align=center] {\footnotesize{Girard's}\\\footnotesize{translation}\\\footnotesize{(version-wise)}} (vlmini1.west);
    \draw[-latex] (vlmini1.east) to node[midway,below] {} (vlmini2.west);
    \draw[-latex] (vlmini1.east) to [yshift=-5]node[midway,below,align=center] (inference){\footnotesize{Type}\\\footnotesize{inference}} (vlmini2.west);
    \draw[-latex] (vlmini2.east) to [yshift=-5]node[midway,below,align=center] (bundling) {\footnotesize{Bundling}} (vlmini3.west);
    \draw[-latex] (vlmini4.west) -| node[near start,below,align=center] {\footnotesize{Import modules}} ([yshift=2]topofinference);
    \draw[-latex] (inference.south) |- (constraints.west);
    \draw[-latex] (bundling.south) |- (constraints.west);
    % \draw[-latex] (inference.south) -- (vdep.north);
    % \draw[-latex] (bundling.south) -- (ldep.north);
\end{tikzpicture}
\caption{The translation phases for a single module with multiple versions.}
\label{fig:translationoverview}
\end{figure}
The entire translation consists of three parts: (1) \emph{Girard's translation}, (2) an \emph{algorithmic type inference}, and (3) \emph{bundling}.
Figure \ref{fig:translationoverview} shows the translation process of a single module. First, through Girard's translation, each version of the \mylang{} program undergoes a version-wise translation into the \vlmini{} program. 
Second, the type inference synthesizes types and constraints for top-level symbols. Variables imported from external modules reference the bundled interface generated in the subsequent step.
Finally, to make the external variables act as multi-version expressions, bundling consolidates each version's interface into one \vlmini{} interface.
These translations are carried out in order from downstream of the dependency tree.
By resolving all constraints up to the main module, the appropriate version for every external variable is determined.

It is essential to note that the translations focus on generating constraints for dispatching external variables into version-specific code. While implementing versioned records in \corelang{} presents challenges, such as handling many version labels and their code clones, our method is a constraint-based approach in \vlmini{} that enables static inference of version labels without their explicit declaration.

In the context of coeffect languages, constraint generation in \mylang{} can be seen as the automatic generation of type declarations paired with resource constraints.
Granule\cite{Orchard:2019:Granule} can handle various resources as coeffects, but it requires type declarations to indicate resource constraints. \mylang{} restricts its resources solely to the version label set. This specialization enables the automatic collection of version information from external sources outside the codebase.

\subsection{An Intermediate Language, \vlmini{}}
% \vspace{-1\baselineskip}
\subsubsection{Syntax of \vlmini{}}
\begin{figure}[tb]
    \hrule
    \medskip
    \begin{minipage}{.9\textwidth}
        \textbf{\vlmini{} syntax (w/o data constructors and version control terms)}
    \end{minipage}
    \begin{minipage}{\textwidth}
        \vspace{-.8\baselineskip}
        \begin{minipage}{.475\textwidth}
          \begin{align*}
            t & ::= n \mid x \mid \app{t_1}{t_2} \mid \lam{p}{t} \mid \pr{t} \tag{terms}\\
            p & ::= x \mid [x] \tag{patterns}\\
            A, B & ::= \inttype{} \mid \alpha \mid \ftype{A}{B} \mid \vertype{r}{A} \tag{types}\\
            \kappa & ::= \typekind{} \mid \labelskind{} \tag{kinds}
          \end{align*}
        \end{minipage}
        \begin{minipage}{.475\textwidth}
          \begin{align*}
            \Gamma & ::= \emptyset \mid \Gamma, x:A \mid \Gamma, x:\verctype{A}{r} \tag{contexts}\\
            \Sigma & ::= \emptyset \mid \Sigma, \alpha:\kappa \tag{type variable kinds}\\
            R      & ::= - \mid r \tag{resource contexts}\\
          \end{align*}
        \end{minipage}
    \end{minipage}
    \begin{minipage}{.9\textwidth}
        \medskip\textbf{Extended with data constructors}
    \end{minipage}
    \begin{minipage}{\textwidth}
        \vspace{-.5\baselineskip}
        \begin{minipage}{.52\textwidth}
          \begin{align*}
            t & ::= \ldots \mid C\,\overline{t_i} \mid \caseof{t}{\overline{p_i \mapsto t_i}}\tag{terms}\\
            p & ::= \ldots \mid C\,\overline{p_i} \tag{patterns}\\
            C & ::= (,) \mid [,] \tag{constructors}
          \end{align*}
        \end{minipage}
        \begin{minipage}{.44\textwidth}
          \begin{align*}
            A, B & ::= ... \mid K \overline{A_i} \tag{types}\\
            K    & ::= (,) \mid [,] \tag{type constructors}\\
          \end{align*}
        \end{minipage}
    \end{minipage}
    \begin{minipage}{.9\textwidth}
        \medskip\textbf{Extended with version control terms}
    \end{minipage}
    \begin{minipage}{\textwidth}
        \vspace{-.5\baselineskip}
        \begin{minipage}{\textwidth}
          \begin{align*}
            t & ::= \ldots \mid \verof{\{\overline{M_i=V_i}\}}{t} \mid \unver{t}\hspace{6em}\tag{terms}
          \end{align*}
        \end{minipage}
    \end{minipage}
    \begin{minipage}{.9\textwidth}
        \medskip\textbf{\vlmini{} constraints}
    \end{minipage}
    \begin{minipage}{\textwidth}
      \vspace{-.5\baselineskip}
      \begin{align*}
        \mathcal{C} & ::= \underbrace{\top \mid \mathcal{C}_1 \land \mathcal{C}_2 \mid \mathcal{C}_1 \lor \mathcal{C}_2}_{\substack{\text{propositional}\\\text{formulae}}} \mid \underbrace{\vphantom{\mid}\alpha \preceq \alpha'}_{\substack{\text{variable}\\\text{dependencies}}} \mid \underbrace{\vphantom{\mid}\alpha \preceq \mathcal{D}}_{\substack{\text{label}\\\text{dependencies}}}\tag{dependency 
 constraints}\\
        \mathcal{D} & ::= \cs{\overline{M_i = V_i}} \tag{dependent labels}\\
        \Theta      & ::= \top \mid \Theta_1 \land \Theta_2 \mid \{A \sim B\} \tag{type constraints}
      \end{align*}
    \end{minipage}
  % \begin{align*}
  %   t & ::= n \mid x \mid \app{t_1}{t_2} \mid \lam{p}{t} \mid \pr{t} \mid C\,\overline{t_i} \mid \caseof{t}{\overline{p_i \mapsto t_i}}\tag{terms}\\
  %   p & ::= n \mid x \mid \_ \mid [p] \mid C\,\overline{p_i} \tag{patterns}\\
  %   C & ::= (,) \mid [,] \tag{constructors}\\
  %   A, B & ::= \inttype{} \mid K \overline{A_i} \mid \bgr{\alpha} \mid \ftype{A}{B} \mid \vertype{r}{A} \tag{types}\\
  %   K    & ::= (,) \mid [,] \tag{type constructors}\\
  %   \kappa & ::= \typekind{} \mid \labelskind{} \tag{kinds}\\
  %   \Gamma & ::= \emptyset \mid \Gamma, x:A \mid \Gamma, x:\verctype{A}{r} \tag{contexts}\\
  %   \Sigma & ::= \emptyset \mid \Sigma, \alpha:\kappa \tag{type variable kinds}\\
  %   R      & ::= - \mid r \tag{resource contexts}\\
  %   \mathcal{C} & ::= \top \mid \mathcal{C}_1 \land \mathcal{C}_2 \mid \mathcal{C}_1 \lor \mathcal{C}_2 \mid \underbrace{\vphantom{\mid}\alpha \preceq \alpha'}_{\textnormal{variable dependencies}} \mid \underbrace{\vphantom{\mid}\alpha \preceq L}_{\textnormal{label dependencies}}\tag{constraints}\\
  %   L & ::= \cs{\overline{M_i \mapsto V_i}} \tag{dependent labels}
  % \end{align*}
  \medskip
  \hrule
  \caption{The syntax of \vlmini{}.}
  \label{syntax:vlmini}
\end{figure}
Figure \ref{syntax:vlmini} shows the syntax of \vlmini{}.
\vlmini{} encompasses all the terms in \corelang{} except for versioned records $\nvval{l_i=t_i}$, intermediate term $\ivval{\overline{l_i=t_i}}{l_k}$, and extractions $t.l_k$. As a result, its terms are analogous to those in \lrpcf{}\cite{brunel_core_2014} and GrMini\cite{Orchard:2019:Granule}. However, \vlmini{} is specialized to treat version resources as coeffects.
We also introduce data constructors by introduction $C\,t_1,...,t_n$ and elimination $\caseof{t}{\overline{p_i \mapsto t_i}}$ for lists and pairs, and version control terms \unver{t} and \verof{\{\overline{M_i=V_i}\}}{t}. 
Here, contextual-let in \corelang{} is a syntax sugar of lambda abstraction applied to a promoted pattern.
\begin{align*}
\clet{x}{t_1}{t_2} \triangleq \app{(\lam{\pr{x}}{t_2})}{t_1}
\end{align*}

Types, version labels, and version resources are almost the same as \corelang{}.
Type constructors are also added to the type in response to the \vlmini{} term having a data constructor.
The remaining difference from \corelang{} is type variables $\alpha$. Since \vlmini{} is a monomorphic language, type variables act as unification variables; type variables are introduced during the type inference and are expected to be either concrete types or a set of version labels as a result of constraint resolution.
To distinguish those two kinds of type variables, we introduce kinds $\kappa$.
The kind \labelskind{} is given to type variables that can take a set of labels $\{\overline{l_i}\}$ and is used to distinguish them from those of kind \typekind{} during algorithmic type inference.

% \vspace{-1\baselineskip}
\subsubsection{Constraints}
The lower part of Figure \ref{syntax:vlmini} shows constraints generated through bundling and type inference.
Dependency constraints comprise \emph{variable dependencies} and \emph{label dependencies} in addition to propositional formulae.
Variable dependencies $\alpha \sqsubseteq \alpha'$ require that if a version label for $\alpha'$ expects a specific version for a module, then $\alpha$ also expects the same version.
Similarly, label dependencies $\alpha \preceq \cs{\overline{M_i = V_i}}$ require that a version label expected for $\alpha$ must be $V_i$ for $M_i$. For example, assuming that versions $1.0.0$ and $2.0.0$ exist for both modules \mn{A} and \mn{B}, the minimal upper bound set of version labels satisfying $\alpha \preceq \cs{\mn{A}\mapsto 1.0.0}$ is $\alpha = \{\{\mn{A}=1.0.0,\mn{B}=1.0.0\},\{\mn{A}=1.0.0,\mn{B}=2.0.0\}\}$. If the constraint resolution is successful, $\alpha$ will be specialized with either of two labels.
$\Theta$ is a set of type equations resolved by the type unification.

\subsection{Girard's Translation for \vlmini{}}
\label{sec:VLMini}
We extend Girard's translation between \mylang{} (lambda calculus) to \vlmini{} following Orchard's approach~\cite{Orchard:2019:Granule}.
\begin{align*}
\llbracket n \rrbracket \equiv n \qquad
\llbracket x \rrbracket \equiv x \qquad
\llbracket \lam{x}{t} \rrbracket \equiv \lam{[x]}{\llbracket t \rrbracket} \qquad
\llbracket t\ s\rrbracket \equiv \llbracket t\rrbracket\ [ \llbracket s \rrbracket ]
\end{align*}

The translation replaces lambda abstractions and function applications of \mylang{} by lambda abstraction with promoted pattern and promotion of \vlmini{}, respectively. From the aspect of types, this translation replaces all occurrences of $\ftype{A}{B}$ with $\ftype{\vertype{r}{A}}{B}$ with a version resource $r$.
This translation inserts a syntactic annotation $[*]$ at each location where a version resource needs to be addressed. Subsequent type inference will compute the resource at the specified location and produce constraints to ensure version consistency at that point.

The original Girard's translation~\cite{girardlinear1987} is well-known as a translation between the simply-typed $\lambda $-calculus and an intuitionistic linear calculus. The approach involves replacing every intuitionistic arrow $A \rightarrow B$ with $!A \multimap B$, and subsequently unboxing via let-in abstraction and promoting during application \cite{Orchard:2019:Granule}.

\subsection{Bundling}
\label{sec:bundling}
Bundling produces an interface encompassing types and versions from every module version, allowing top-level symbols to act as multi-version expressions. During this process, bundling reviews interfaces from across module versions, identifies symbols with the same names and types after removing $\square_r$ using Girard's transformation, and treats them as multiple versions of a singular symbol (also discussed in Section \ref{sec:typeincompatibilities}). In a constraint-based approach, bundling integrates label dependencies derived from module versions, ensuring they align with the version information in the typing rule for versioned records of \corelang{}.

For example, assuming that the $\mathit{id}$ that takes an \inttype{} value as an argument is available in version 1.0.0 and 2.0.0 of \mn{M} as follows:
\begin{align*}
\mathit{id} &: \vertype{\alpha_1}{(\ftype{\vertype{\alpha_2}{\inttype}}{\inttype})}\ |\ \mathcal{C}_1 \tag{\textnormal{version 1.0.0}}\\
\mathit{id} &: \vertype{\beta_1}{(\ftype{\vertype{\beta_2}{\inttype}}{\inttype})}\ |\ \mathcal{C}_2 \tag{\textnormal{version 2.0.0}}
\end{align*}
where $\alpha_1$ and $\alpha_2$ are version resource variables given from type inference. They capture the version resources of $\mathit{id}$ and its argument value in version 1.0.0. $\mathcal{C}_1$ is the constraints that resource variables of version 1.0.0 will satisfy. Likewise for $\beta_1$, $\beta_2$, and $\mathcal{C}_2$.
Since the types of $\mathit{id}$ in both versions become $\ftype{\inttype}{\inttype}$ via Girard's translation, they can be bundled as follows:
\begin{align*}
\mathit{id} : \vertype{\gamma_1}{(\ftype{\vertype{\gamma_2}{\textsf{Int}}}{\textsf{Int}})}\ |\
\mathcal{C}_1 \land \mathcal{C}_2 \land \Bigl(\ 
     &(\gamma_1 \preceq \cs{\mn{M} = 1.0.0} \land \gamma_1 \preceq \alpha_1 \land \gamma_2 \preceq \alpha_2)\\
\lor\ &(\gamma_1 \preceq \cs{\mn{M} = 2.0.0} \land \gamma_1 \preceq \beta_1 \land \gamma_2 \preceq \beta_2)\ \Bigr)
\end{align*}
where $\gamma_1$ and $\gamma_2$ are introduced by this conversion for the bundled $id$ interface, with label and variable dependencies that they will satisfy.
$\gamma_1$ captures the version resource of the bundled $\mathit{id}$. The generated label dependencies $\gamma_1 \preceq \cs{\mn{M} = 1.0.0}$ and $\gamma_1 \preceq \cs{\mn{M} = 2.0.0}$ indicate that $\mathit{id}$ is available in either version 1.0.0 or 2.0.0 of \mn{M}.
These label dependencies are exclusively\footnote{In the type checking rules for $\verof{l}{t}$, type inference exceptionally generates label dependencies. Please see Appendix \ref{appendix:vlmini_version_control_terms}} generated during bundling.
The other new variable dependencies indicate that the $\mathit{id}$ bundled interface depends on one of the two version interfaces. The dependency is made apparent by pairing the new resource variables with their respective version resource variable for each version.
These constraints are retained globally, and the type definition of the bundled interface is used for type-checking modules importing $\mathit{id}$.

\section{Algorithmic Type Inference}
\label{inference}
\begin{figure*}[t]
  \centering
  \hrule
  \smallskip
  \begin{tabular}{c}
    \begin{minipage}{.9\linewidth{}}
        \textbf{\vlmini{} pattern type synthesis \ \ \ 
        \fbox{\ensuremath{\Sigma, R\,\vdash\, p : A \rhd \Gamma; \Sigma'; \Theta; \mathcal{C}}}}
    \end{minipage}
    \smallskip\\
    % \begin{minipage}{.55\linewidth}
    %   \infrule[\mbox{[}p\_\mbox{]}]{
    %     \Sigma \,\vdash\, r : \textsf{Labels}
    %     \andalso
    %     \Sigma \,\vdash\, A : \textsf{Type}\\
    %   }{
    %     \Sigma; r \,\vdash\, \_ : A \rhd \emptyset; \Sigma; \emptyset
    %   }
    % \end{minipage}
    % \\\\
    % \begin{minipage}{.45\linewidth}
    %   \infrule[pInt]{
    %     \\
    %   }{
    %     \Sigma; - \,\vdash\, n : \inttype \rhd \emptyset; \Sigma; \emptyset
    %   }
    % \end{minipage}
    % \begin{minipage}{.50\linewidth}
    %   \infrule[\mbox{[}pInt\mbox{]}]{
    %     \Sigma \,\vdash\, r : \textsf{Labels}
    %   }{
    %     \Sigma; r \,\vdash\, n : \inttype \rhd \emptyset; \Sigma; \emptyset
    %   }
    % \end{minipage}
    % \\\\
    \begin{minipage}{.475\linewidth}
      \infrule[pVar]{
        % \vdash \Sigma
        % \andalso
        % \Sigma \vdash A : \textsf{Type}
        \\
      }{
        \Sigma; - \vdash x : A \rhd x:A; \Sigma; \top; \top
      }
    \end{minipage}
    \begin{minipage}{.475\linewidth}
      \infrule[\mbox{[}pVar\mbox{]}]{
        % \vdash \Sigma
        % \andalso
        % \Sigma \vdash A : \textsf{Type}
        % \andalso
        % \Sigma \vdash r : \textsf{Labels}
        \\
      }{
        \Sigma; r \vdash x : A \rhd x:\verctype{A}{r}; \Sigma; \top; \top
      }
    \end{minipage}
    \smallskip\\
    \begin{minipage}{.85\linewidth}
      \infrule[p$\square$]{
        \Sigma, \alpha:\textsf{Labels}, \beta:\textsf{Type}; \alpha \vdash x : \beta \rhd \Delta; \Sigma'; \Theta; \mathcal{C}
      }{
        \Sigma; - \vdash [x] : A \rhd \Delta; \Sigma'; \Theta \land \{A \sim \vertype{\alpha}{\beta}\}; \mathcal{C}
      }
    \end{minipage}
    % \smallskip\\
    % \begin{minipage}{.85\linewidth}
    %   \infrule[\mbox{[}p$\square$\mbox{]}]{
    %     \Sigma, \alpha:\textsf{Labels}, \beta:\textsf{Type}; r \otimes \alpha \vdash p : \beta \rhd \Delta; \Sigma'; \Theta; \mathcal{C}
    %     % \textsf{flatten}(r,R,r',R')=(s,S)
    %   }{
    %     \Sigma; r \vdash [p] : A \rhd \Delta; \Sigma'; \Theta\land \{A \sim \vertype{\alpha}{\beta}\}; \mathcal{C}
    %   }
    % \end{minipage}
  \end{tabular}
  \smallskip
  \hrule
  \smallskip
  \begin{tabular}{c}
    \begin{minipage}{.9\linewidth{}}
        \textbf{\vlmini{} type synthesis (excerpt) \ \ \ \fbox{\ensuremath{\Sigma;\Gamma \vdash t \Rightarrow A;\Sigma';\Theta; \mathcal{C}}}}
    \end{minipage}
    \smallskip\\
    % \begin{minipage}{.95\linewidth}
    %   \infrule[$\Rightarrow_{\textsc{int}}$]{
    %     \\% \\
    %   }{
    %     \Sigma; \Gamma \vdash n \Rightarrow \textsf{Int}; \Sigma; \emptyset; \top
    %   }
    % \end{minipage}
    % \\\\
    \begin{minipage}{.45\linewidth}
      \infrule[$\Rightarrow_{\textsc{lin}}$]{
        % \vdash \Sigma
        % \andalso
        % \Sigma \vdash \Gamma
        % \andalso
        x:A\in\Gamma
      }{
        \Sigma; \Gamma \vdash x \Rightarrow A; \Sigma; x:A; \top; \top
      }
    \end{minipage}
    \begin{minipage}{.45\linewidth}
      \infrule[$\Rightarrow_{\textsc{gr}}$]{
        % \vdash \Sigma
        % \andalso
        % \Sigma \vdash \Gamma
        % \andalso
        x:\verctype{A}{r}\in\Gamma
      }{
        \Sigma; \Gamma \vdash x \Rightarrow A; \Sigma; x:\verctype{A}{1}; \top; \top
      }
    \end{minipage}
    \medskip\\
    \begin{minipage}{.95\linewidth}
      \infrule[$\Rightarrow_{\textsc{abs}}$]{
        \Sigma_1, \alpha:\textsf{Type};- \vdash p:\alpha \rhd \Gamma'; \Sigma_2; \Theta_1
        \andalso
        \Sigma_2;\Gamma,\Gamma' \vdash t \Rightarrow B;\Sigma_3;\Delta; \Theta_2; \mathcal{C}
      }{
        \Sigma_1;\Gamma \vdash \lam{p}{t} \Rightarrow \ftype{\alpha}{B};\Sigma_3;\Delta\backslash\Gamma' ; \Theta_1\land\Theta_2; \mathcal{C}
      }
    \end{minipage}
    \medskip\\
    % \begin{minipage}{.95\linewidth}
    %   \infrule[$\Rightarrow_{\textsc{app}}$]{
    %     \Sigma_1; \Gamma \vdash t_1 \Rightarrow A ; \Sigma_2; \theta_1; \mathcal{C}_1
    %     \andalso
    %     \Sigma_2; \Gamma \vdash t_2 \Rightarrow A'; \Sigma_3; \theta_2; \mathcal{C}_2
    %     \andalso\\
    %     \Sigma_4 = \Sigma_3, \beta:_{\exists}\typekind
    %     \andalso
    %     \Sigma_4 \vdash A \sim \ftype{A'}{\beta} \rhd \theta_3
    %     \andalso
    %     \theta_4 = \theta_1 \uplus \theta_2 \uplus \theta_3
    %   }{
    %     \Sigma_1;\Gamma \vdash \app{t_1}{t_2} \Rightarrow \theta_4 \beta; \Sigma_4; \theta_4; \mathcal{C}_1 \land \mathcal{C}_2
    %   }
    % \end{minipage}
    % \\\\
    \begin{minipage}{.95\linewidth}
      \infrule[$\Rightarrow_{\textsc{pr}}$]{
        \Sigma_1 \vdash [\Gamma\cap\textsf{FV}(t)]_{\textsf{Labels}} \rhd \Gamma'
        \andalso
        \Sigma_1; \Gamma' \vdash t \Rightarrow A; \Sigma_2; \Delta; \Theta; \mathcal{C}_1
        \andalso\\
        \Sigma_3 = \Sigma_2, \alpha:\textsf{Labels}
        \andalso
        \Sigma_3 \vdash \alpha \sqsubseteq_c \Gamma' \rhd \mathcal{C}_2
      }{
        \Sigma_1;\Gamma \vdash [t] \Rightarrow \vertype{\alpha}{A}; \Sigma_3; \alpha \cdot \Delta ;\Theta; \mathcal{C}_1 \land \mathcal{C}_2
      }
    \end{minipage}
  \end{tabular}
  \smallskip
  \hrule
  \smallskip
  \begin{tabular}{c}
    \begin{minipage}{.9\linewidth{}}
        % \medskip
        \textbf{\vlmini{} constraints generation \ \ \ 
        \fbox{\ensuremath{\Sigma \,\vdash\, \alpha \sqsubseteq_{c} \Gamma \rhd \mathcal{C}}}}
    \end{minipage}
    \smallskip\\
    \begin{minipage}{.35\linewidth}
      \infrule[$\emptyset$]{
        % \Sigma \,\vdash\, \alpha : \textsf{Labels}
        \\
      }{
        \Sigma \,\vdash\, \alpha \sqsubseteq_{c} \emptyset \rhd \top
      }
    \end{minipage}
    \begin{minipage}{.60\linewidth}
      \infrule[$\alpha$]{
        % \Sigma \,\vdash\, \alpha : \textsf{Labels}
        % \andalso
        \Sigma \,\vdash\, \alpha \sqsubseteq_{c} \Gamma \rhd \mathcal{C}
      }{
        \Sigma \,\vdash\, \alpha \sqsubseteq_{c} (x:[A]_r, \Gamma) \rhd (\alpha \preceq r \land \mathcal{C})
      }
    \end{minipage}
  \end{tabular}
  \smallskip
  \hrule
  \caption{\vlmini{} algorithmic typing.}
  %\ecaption{Typing of \corelang}
  % \Description{\corelang algorithmic typing}
  \label{fig:rule_algorithmic_typing}
\end{figure*}
We develop the algorithmic type inference for \vlmini{} derived from the declarative type system of \corelang{}~\cite{Tanabe:2018:CPA:3242921.3242923,Tanabe_2021}.
The type inference consists of two judgments: \emph{type synthesis} and \emph{pattern type synthesis}.
The judgment forms are similar to Gr~\cite{Orchard:2019:Granule}, which is similarly based on coeffect calculus.
While Gr provides type-checking rules in a bidirectional approach~\cite{10.1145/2544174.2500582,10.1145/3290322} to describe resource constraint annotations and performs unifications inside the type inference, \vlmini{} only provides synthesis rules and unification performs after the type inference.
In addition, Gr supports user-defined data types and multiple computational resources, while \vlmini{} supports only built-in data structures and specializes in version resources.
The inference system is developed to be sound for declarative typing in \corelang{}, with the proof detailed in Appendix \ref{appendix:vlmini_safety}.

Type synthesis takes type variable kinds $\Sigma$, a typing context $\Gamma$ of term variables, and a term $t$ as inputs. Type variable kinds $\Sigma$ are added to account for distinct unification variables for types and version resources.
The synthesis produces as outputs a type $A$, type variable kinds $\Sigma'$, type constraints $\Theta$, and dependency constraints $\mathcal{C}$.
The type variable kinds $\Sigma$ and $\Sigma'$ always satisfy $\Sigma \subseteq \Sigma'$ due to the additional type variables added in this phase.

Pattern type synthesis takes a pattern $p$, type variable kinds $\Sigma$, and resource environment $R$ as inputs. It synthesizes outputs, including typing context $\Gamma$, type variable kinds $\Sigma'$, and type and dependency constraints $\Theta$ and $\mathcal{C}$.
Pattern type synthesis appears in the inference rules for $\lambda$-abstractions and case expressions. It generates a typing context from the input pattern $p$ for typing $\lambda$-bodies and branch expressions in case statements.
When checking a nested promoted pattern, the resource context $R$ captures version resources inside a pattern.

\subsection{Pattern Type Synthesis}
Pattern type synthesis conveys the version resources captured by promoted patterns to the output typing context. The rules are classified into two categories, whether or not it has resources in the input resource context $R$. The base rules are \textsc{pVar}, \textsc{p}$\Box$, while the other rules are resource-aware versions of the corresponding rules. The resource-aware rules assume they are triggered within the promoted pattern and collect version resource $r$ in the resource context.

The rules for variables \textsc{pVar} and \textsc{[pVar]} differ in whether the variable pattern occurs within a promoted pattern. \textsc{pVar} has no resources in the resource context because the original pattern is not inside a promoted pattern. Therefore, this pattern produces typing context $x:A$. \textsc{[pVar]} is for a variable pattern within the promoted pattern, and a resource $r$ is recorded in the resource context. The rule assigns the collected resource $r$ to the type $A$ and outputs it as a versioned assumption $x:\verctype{A}{r}$.

The rules for promoted patterns \textsc{p}$\square$
propagate version resources to the subpattern synthesis. The input type $A$ is expected to be a versioned type, so the rule generates the fresh type variables $\alpha$ and $\beta$, then performs the subpattern synthesis considering $A$ as $\vertype{\alpha}{\beta}$. Here, the resource $\alpha$ captured by the promoted pattern is recorded in the resource context. Finally, the rule unifies $A$ and $\vertype{\alpha}{\beta}$ and produces the type constraints $\Theta'$ for type refinement.

\subsection{Type Synthesis}
The algorithmic typing rules for \vlmini{}, derived from declarative typing rules for \corelang{}, are listed in Figure \ref{fig:rule_algorithmic_typing}. We explain a few important rules in excerpts.

The rule $\Rightarrow_{\textsc{abs}}$ generates a type variable $\alpha$, along with the binding pattern $p$ of the $\lambda$-abstraction generating the typing context $\Gamma'$. Then the rule synthesizes a type $B$ for the $\lambda$-body under $\Gamma'$, and the resulting type of the $\lambda$-abstraction is $\alpha \rightarrow B$ with the tentatively generated $\alpha$.
With the syntax sugar, the type rules of the contextual-let are integrated into $\Rightarrow_{\textsc{abs}}$.
Instead, $\lambda$-abstraction does not just bind a single variable but is generalized to pattern matching, which leverages pattern typing, as extended by promoted patterns and data constructors. 

The rule $\Rightarrow_{\textsc{pr}}$ is the only rule that introduces constraints in the entire type inference algorithm.
This rule intuitively infers consistent version resources for the typing context $\Gamma$. Since we implicitly allow for weakening, we generate a constraint from $\Gamma'$ that contains only the free variables in $t$, produced by \emph{context grading} denoted as $[\Gamma]_\textsf{Labels}$.
Context grading converts all assumptions in the input environment into versioned assumptions by assigning the empty set
for the assumption with no version resource.

Finally, the rule generates constraints from $\Gamma'$ and a fresh type variable $\alpha$ by constraints generation defined in the lower part of Figure \ref{fig:rule_algorithmic_typing}.
The rules assert that the input type variable $\alpha$ is a subset of all the resources of the versioned assumptions in the input environment $\Gamma$. The following judgment is the simplest example triggered by the type synthesis of $\pr{\app{f}{x}}$.
\begin{align*}
r:\labelskind,s:\labelskind \,\vdash\, \alpha \sqsubseteq_{c} f:\verctype{\ftype{\inttype}{\inttype}}{r}, x:\verctype{\inttype}{s} \rhd \alpha \preceq r \land \alpha \preceq s
\end{align*}
The inputs are type variable $\alpha$ and the type environment ($f:\verctype{\ftype{\inttype}{\inttype}}{r}, x:\verctype{\inttype}{s}$). In this case, the rules generate variable dependencies for $r$ and $s$, each resource of the assumptions, and return a constraint combined with $\land$.

\subsection{Extensions}
\subsubsection{Version Control Terms}
The rule for $\verof{l}{t}$ uses the same trick as ($\Rightarrow_\textsc{pr}$), and generates label dependencies from the input environment $\Gamma$ to $\cs{l}$. Since $\verof{l}{t}$ only instructs the type inference system, the resulting type is the same as $t$.
$\unver{t}$ removes the version resource from the type of $t$, which is assumed to be a versioned value. We extend Girard's translation so that $t$ is always a versioned value.
Since a new resource variable is given to the term by the promotion outside of \textbf{unversion}, the inference system guarantees the version consistency inside and outside the boundary of \textbf{unversion}.
The list of the rules is provided in Appendix \ref{appendix:vlmini_version_control_terms}.

% \vspace{-\baselineskip}
\subsubsection{Data Structures}
To support data structures, Hughes et al. suggest that coeffectful data types are required to consider the interaction between the resources inside and outside the constructor~\cite{EPTCS353.6}. They introduce the derivation algorithm for \emph{push} and \emph{pull} for an arbitrary type constructor $K$ to address this.

\begin{center}
\begin{minted}{haskell}
push : @$\forall$@{a b: Type, r: Labels}. (a,b)[r] -> (a[r],b[r])
push [(x, y)] = ([x], [y])
pull : @$\forall$@{a b: Type, m n: Labels}. (a[n],b[m]) -> (a,b)[n@$\sqcap$@m]
pull ([x], [y]) = [(x, y)]
\end{minted}
\end{center}

Following their approach, we developed inference rules for pairs and lists.
When a data structure value $p$ is applied to a function $f$, the function application $\app{f}{p}$ is implicitly interpreted as $\app{f}{(\app{pull}{p})}$. As a dual, a pattern match for a data structure value $\caseof{p}{\overline{p_i \mapsto t_i}}$ is interpreted as $\caseof{(\app{push}{p})}{\overline{p_i \mapsto t_i}}$.
Appendix \ref{appendix:vlmini_data_structures} provides the complete set of extended rules.
 
% Section Template

\section{Implementation}
\label{implementation} % Change X to a consecutive number; for referencing this Section elsewhere, use \ref{SectionX}

%----------------------------------------------------------------------------------------
%	section 1
%----------------------------------------------------------------------------------------

% \input{figs/fig6-1.tex}
% \subsection{Overview}
We implement the \mylang{} compiler\footnote{\url{https://github.com/yudaitnb/vl}} on GHC (v9.2.4) with haskell-src-exts\footnote{\url{https://hackage.haskell.org/package/haskell-src-exts}} as its parser with an extension of versioned control terms, and z3~\cite{10.1007/978-3-540-78800-3_24} as its constraint solver.
The \mylang{} compiler performs the code generation by compiling \vlmini{} programs back into $\lambda$-calculus via Girard's translation and then translating them into Haskell ASTs using the version in the result version labels.

% \vspace{-1\baselineskip}
% \vspace{-.5\baselineskip}
\paragraph{\textbf{Ad-hoc Version Polymorphism via Duplication}}
\label{sec:adhocversionpolymorphism}
The \mylang{} compiler replicates external variables to assign individual versions to homonymous external variables.
Duplication is performed before type checking of individual versions and renames every external variable along with the type and constraint environments generated from the import declarations.
Such ad hoc conversions are necessary because \vlmini{} is monomorphic, and the type inference of \vlmini{} generates constraints by referring only to the variable's name in the type environment.
Therefore, assigning different versions to homonymous variables requires manual renaming in the preliminary step of the type inference.
A further discussion on version polymorphism can be found in Section \ref{sec:fullversionpolymorphism}.

% \vspace{-.5\baselineskip}
\paragraph{\textbf{Constraints Solving with z3}}
We use sbv\footnote{\url{https://hackage.haskell.org/package/sbv-9.0}} as the binding of z3.
The sbv library internally converts constraints into SMT-LIB2 scripts~\cite{barrett2010smt} and supplies it to z3.
Dependency constraints are represented as vectors of symbolic integers, where the length of the vector equals the number of external modules, and the elements are unique integers signifying each module's version number. Constraint resolution identifies the expected vectors for symbolic variables, corresponding to the label on which external identifiers in \mylang{} should depend. If more than one label satisfies the constraints, the default action is to select a newer one.
 
\section{Case Study and Evaluation}
\label{casestudy}
\begin{table*}[tbp]
  \centering
  \begin{tabular}{r|ccc}
    version & \texttt{join} & \texttt{vjoin} & \texttt{udot}, \texttt{sortVector}, \texttt{roundVector}\\ \hline
    $<$ 0.15     & available  & undefined & undefined\\
    $\geq$ 0.16  & deleted & available & available\\
  \end{tabular}
  \caption{Availability of functions in hmatrix before and after tha update.}
  \label{table:join}
\end{table*}
\subsection{Case Study}
We demonstrate that \mylang{} programming achieves the two benefits of programming with versions. 
The case study simulated the incompatibility of hmatrix,\footnote{\url{https://github.com/haskell-numerics/hmatrix/blob/master/packages/base/CHANGELOG}} a popular Haskell library for numeric linear algebra and matrix computations, in the VL module \mn{Matrix}.
This simulation involved updating the applications \mn{Main} depending on \mn{Matrix} to reflect incompatible changes.

Table \ref{table:join} shows the changes introduced in version 0.16 of hmatrix. Before version 0.15, hmatrix provided a \texttt{join} function for concatenating multiple vectors.
The update from version 0.15 to 0.16 replaced \texttt{join} with \texttt{vjoin}. Moreover, several new functions were introduced.
We implement two versions of \mn{Matrix} to simulate backward incompatible changes in \mylang{}.
Also, due to the absence of user-defined types in \mylang{}, we represent \texttt{Vector a} and \texttt{Matrix a} as \texttt{List Int} and \texttt{List (List Int)} respectively, using \mn{List}, a partial port of \texttt{Data.List} from the Haskell standard library.

\begin{figure}[t]

\begin{minipage}{.5\textwidth}
% \begin{lstlisting}[style=haskell]
\begin{minted}{haskell}
module Main where
import Matrix
import List
main = let
  vec = [2, 1]
  sorted = sortVector vec
  m22 = join -- [[1,2],[2,1]]
          (singleton sorted)
          (singleton vec)
  in determinant m22
-- error: version inconsistent
\end{minted}
% \end{lstlisting}
\end{minipage}
\begin{minipage}{.5\textwidth}
\begin{minted}{haskell}
module Main where
import Matrix
import List
main = let
  vec = [2, 1]
  sorted = @\vlkey{unversion}@
             (sortVector vec)
  m22 = join -- [[1,2],[2,1]]
          (singleton sorted)
          (singleton vec)
  in determinant m22 -- ->* -3
\end{minted}
\end{minipage}
\caption{Snippets of \texttt{Main} before (left) and after (right) rewriting.}
\label{fig6-5}

\end{figure}
We implement \mn{Main} working with two conflicting versions of \mn{Matrix}. The left side of Figure \ref{fig6-5} shows a snippet of \mn{Main} in the process of updating \mn{Matrix} from version 0.15.0 to 0.16.0. \fn{main} uses functions from both versions of \mn{Matrix} together: \fn{join} and \fn{sortVector} are available only in version 0.15.0 and 0.16.0 respectively, hence \mn{Main} has conflicting dependencies on both versions of \mn{Matrix}. Therefore, it will be impossible to successfully build this program in existing languages unless the developer gives up using either \fn{join} or \fn{sortVector}.

\begin{itemize}
\item \textbf{Detecting Inconsistent Version}:
\mylang{} can accept \mn{Main} in two stages. First, the compiler flags a version inconsistency error.
It is unclear which \mn{Matrix} version the \fn{main} function depends on as \fn{join} requires version 0.15.0 while \fn{sortVector} requires version 0.16.0.
The error prevents using such incompatible version combinations, which are not allowed in a single expression.

\item \textbf{Simultaneous Use of Multiple Versions}:
In this case, using \fn{join} and \fn{sortVector} simultaneously is acceptable, as their return values are vectors and matrices. Therefore, we apply \texttt{\vlkey{unversion} t} for $t$ to collaborate with other versions.
The right side of Figure \ref{fig6-5} shows a rewritten snippet of \mn{Main}, where \texttt{sortVector vec} is replaced by \texttt{\vlkey{unversion} (sortVector vec)}. Assuming we avoid using programs that depend on a specific version elsewhere in the program, we can successfully compile and execute \fn{main}.
\end{itemize}

\subsection{Scalability of Constraint Resolution}
\begin{figure}[tbp]
    \centering
    \includegraphics[height=6.5cm]{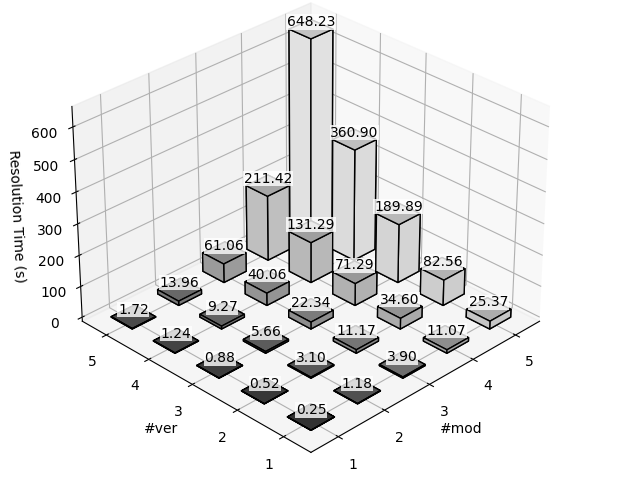}
    \caption{Constraint resolution time for the duplicated \mn{List} by \texttt{\#mod} $\times$ \texttt{\#ver}.}
    \label{fig:consres}
\end{figure}

We conducted experiments on the constraint resolution time of the \mylang{} compiler. In the experiment, we duplicated a \mylang{} module, renaming it to \texttt{\#mod} like \mn{List\_i}, and imported each module sequentially. Every module had the same number of versions, denoted as \texttt{\#ver}. Each module version was implemented identically to \mn{List}, with top-level symbols distinguished by the module name, such as \fn{concat\_List\_i}. The experiments were performed ten times on a Ryzen 9 7950X running Ubuntu 22.04, with \texttt{\#mod} and \texttt{\#ver} ranging from 1 to 5.

Figure \ref{fig:consres} shows the average constraint resolution time. 
The data suggests that the resolution time increases polynomially (at least square) for both \texttt{\#mod} and \texttt{\#ver}.
Several issues in the current implementation contribute to this inefficiency:
First, we employ sbv as a z3 interface, generating numerous redundant variables in the SMT-Lib2 script.
For instance, in a code comprising 2600 LOC (with $\texttt{\#mod} =5$ and $\texttt{\#ver} =5$), the \mylang{} compiler produces 6090 version resource variables and the sbv library creates SMT-Lib2 scripts with approximately 210,000 intermediate symbolic variables.
Second, z3 solves versions for all AST nodes, whereas the compiler's main focus should be on external variables and the subterms of \texttt{\vlkey{unversion}}.
Third, the current \mylang{} nests the constraint network, combined with $\lor$, \texttt{\#mod} times at each bundling. This approach results in an overly complex constraint network for standard programs.
Hence, to accelerate constraint solving, we can develop a more efficient constraint compiler for SMT-Lib2 scripts, implement preprocess to reduce constraints, and employ a greedy constraint resolution for each module.

\section{Related Work, Future Work, and Conclusion} % Main Section title
\label{conclusion} % Change X to a consecutive number; for referencing this Section elsewhere, use \ref{SectionX}

% \vspace{-.5\baselineskip}
\subsubsection{Managing Dependency Hell}
Mainstream techniques for addressing dependency hell stand in stark contrast to our approach, which seeks to manage dependencies at a finer granularity.
\emph{Container}~\cite{container:docker} encapsulates each application with all its dependencies in an isolated environment, a container, facilitating multiple library versions to coexist on one physical machine. However, it does not handle internal dependencies within the container.
\emph{Monorepository}~\cite{monorepo:google,monorepo:facebook} versions logically distinct libraries within a single repository, allowing updates across multiple libraries with one commit. It eases testing and bug finding but can lower the system modularity.

% \vspace{-.5\baselineskip}
\subsubsection{Toward a Language Considering Compatibility}
The next step in this research is to embed compatibility tracking within the language system. The current \mylang{} considers different version labels incompatible unless a programmer uses \texttt{\vlkey{unversion}}. Since many updates maintain backward compatibility and change only minor parts of the previous version, the existing type system is overly restrictive.

To illustrate, consider Figure \ref{fig2-3} again with more version history. The module \mn{Hash} uses the MD5 algorithm for \fn{mkHash} and \fn{match} in the 1.x.x series. However, it adopts the SHA-3 algorithm in version 2.0.0, leaving other functions the same. The hash by \fn{mkHash} version 1.0.1 (an MD5 hash) aligns with any MD5 hash from the 1.x.x series. Therefore, we know that comparing the hash using \fn{match} version 1.0.0 is appropriate. However, the current \mylang{} compiler lacks mechanisms to express such compatibility in constraint resolution. The workaround involves using \texttt{\vlkey{unversion}}, risking an MD5 hash's use with \fn{match} version 2.0.0.

One promising approach to convey compatibilities is integrating semantic versioning~\cite{preston2013semantic} into the type system.
If we introduce semantics into version labels, the hash generated in version 1.0.1 is backward compatible with version 1.0.0. Thus, by constructing a type system that respects explicitly defined version compatibilities, we can improve \mylang{} to accept a broader range of programs.

It is important to get reliable versions to achieve this goal.
Lam et al.~\cite{10.1145/3426428.3426922} emphasize the need for tool support to manage package modifications and the importance of analyzing compatibility through program analysis.
\emph{Delta-oriented programming}~\cite{10.1145/1868688.1868696,10.5555/1885639.1885647,10.1145/1960275.1960283} could complement this approach by facilitating the way modularizing addition, overriding, and removal of programming elements and include application conditions for those modifications.
This could result in a sophisticated package system that provides granular compatibility information.

Such a language could be an alternative to existing technologies for automatic update, collectively known as \emph{adoptation}.
These methods generate replacement rules based on structural similarities~\cite{Cossette2014,5970177} and extract API replacement patterns from migrated code bases \cite{10.1145/1368088.1368153}.
Some techniques involve library maintainers recording refactorings~\cite{10.1002/smr.328,1553570} and providing annotations~\cite{565039} to describe how to update client code. However, the reported success rate of these techniques is less than 20\% on average~\cite{10.1145/2393596.2393661}.

% \vspace{-.5\baselineskip}
\subsubsection{Supporting Type Incompatibility}
\label{sec:typeincompatibilities}
One of the apparent problems with the current \mylang{} does not support \emph{type incompatibilities}.
\mylang{} forces terms of different versions to have the same type, both on the theoretical (typing rules in \corelang{}) and implementation (bundling in \vlmini{}) aspects. Supporting type incompatibility is important because type incompatibility is one of the top reasons for error-causing incompatibilities \cite{RAEMAEKERS2017140}.
The current \mylang{} is designed in such a way because it retains the principle that equates the types of promotions and versioned records in \corelang{}, easing the formalization of the semantics.

A promising approach to address this could be to decouple version inference from type inference and develop a version inference system on the polymorphic record calculus \cite{10.1145/218570.218572}.
The idea stems from the fact that versioned types $\vertype{\{l_1,l_2\}}{A}$ are structurally similar to record types $\{ l_1 : A,\, l_2 : A\}$ of $\Lambda^{\forall,\bullet}$. 
Since $\Lambda^{\forall,\bullet}$ allows different record-element types for different labels and has concrete inference algorithms with polymorphism, implementing version inference on top of $\Lambda^{\forall,\bullet}$ would also make \mylang{} more expressive.

% \vspace{-\baselineskip}
\subsubsection{Adequate Version Polymorphism}
\label{sec:fullversionpolymorphism}
In the current \mylang{}, there is an issue that the version label of top-level symbols in imported modules must be specified one, whereas users can select specific versions of external variables using \texttt{\vlkey{unversion}} within the importing module.
Consider using a generic function like \fn{List.concat} in Figure \ref{fig6-5}. If it is used in one part of the program within the context of \mn{Matrix} version 1.0.0, the solution of the resource variable of \fn{List.concat} version 1.0.0 becomes confined to $\{\mn{Matrix}=1.0.0,\mn{List}=\ldots\}$. As a result, it is impossible to utilize \fn{List.concat} version 1.0.0 with \mn{Matrix} version 2.0.0 elsewhere in the program. This problem becomes apparent when we define a generic module like a standard library.

It is necessary to introduce full-version polymorphism in the core calculus instead of duplication to address this problem.
The idea is to generate a type scheme by solving constraints for each module during bundling and instantiate each type and resource variable at each occurrence of an external variable.
Such resource polymorphism is similar to that already implemented in Gr~\cite{Orchard:2019:Granule}. However, unlike Gr, \vlmini{} provides a type inference algorithm that collects constraints on a per-module basis, so we need the well-defined form of the principal type. This extension is future work.

% \vspace{-\baselineskip}
\subsubsection{Conclusion}
This paper proposes a method for dependency analysis and version control at the expression level by incorporating versions into language semantics, which were previously only identifiers of packages.
This enables the simultaneous use of multiple versions and identifies programs violating version consistency at the expression level, which is impossible with conventional languages.

Our next step is to extend the version label, which currently only identifies versions, to \emph{semantic versions} and to treat the notion of compatibility with language semantics.
Like automatic updates by modern build tools based on semantic versioning, it would be possible to achieve incremental updates, which would be done step-by-step at the expression level.
Working with existing package managers to collect compatibility information at the expression level would be more feasible to realize the goal.

\bibliographystyle{splncs04}
\bibliography{tanabe.bib}

%
% ---- Appendix ----
%
\appendix
\clearpage
\section{\corelang{} Definitions}

\subsection{\corelang{} Syntax}

\paragraph{\textnormal{\textbf{\corelang{} syntax}}}
\begin{align*}
t      \quad&::=\quad n \mid x \mid \app{t_1}{t_2} \mid \lam{p}{t} \mid \\
            &\hspace{2em}\quad \clet{x}{t_1}{t_2} \mid u.l \mid \ivval{\overline{l_i=t_i}}{l_j} \mid u\tag{terms}\\
p      \quad&::=\quad x \mid  [x] \tag{patterns}\\
u      \quad&::=\quad \pr{t} \mid \nvval{\overline{l_i=t_i}} \tag{versioned values}\\
v      \quad&::=\quad \lam{p}{t} \mid n \mid u \tag{values}\\
A, B   \quad&::=\quad \inttype \mid \ftype{A}{B} \mid \vertype{r}{A} \tag{types}\\
r      \quad&::=\quad \bot \mid \{\overline{l_i}\} \tag{version resources}\\
\mathcal{L}\ \ni\ l\quad&::=\quad \{\overline{M_i = V_i}\} \tag{version labels}
\end{align*}
where $M_i\in\mathcal{M}$ and $V_i\in\mathcal{V}_{M_i}$ are metavariables over module names and versions of $M_i$, respectively.

\paragraph{\textnormal{\textbf{\corelang{} contexts}}}
\begin{align*}
\Gamma,\Delta \quad&::=\quad \emptyset \mid \Gamma, x:A \mid \Gamma, x:\verctype{A}{r} \tag{contexts}\\
R      \quad&::=\quad - \mid r \tag{resource contexts}\\
E      \quad&::=\quad [\cdot]\ \mid\ E\ t \ \mid\ E.l\ \mid\ \clet{x}{E}{t} \tag{evaluation contexts}
\end{align*}

\subsection{\corelang{} Well-formedness}

\begin{rules}{Type well-formedness}{\vdash A}
    \begin{minipage}{.20\linewidth}
        \infrule[Tw$_\textsc{Int}$]{
             \\
        }{
            \vdash \inttype{}
        }
    \end{minipage}
    \begin{minipage}{.35\linewidth}
        \infrule[Tw$_\rightarrow$]{
             \vdash A
             \andalso
             \vdash B
        }{
            \vdash \ftype{A}{B}
        }
    \end{minipage}
    \begin{minipage}{.30\linewidth}
        \infrule[Tw$_\square$]{
             \vdash r
             \andalso
             \vdash A
        }{
            \vdash \vertype{r}{A}
        }
    \end{minipage}
\end{rules}

\begin{rules}{Resource well-formedness}{\vdash r}
    \begin{minipage}{.2\linewidth}
        \infrule[Rw$_\bot$]{
             \\
        }{
            \vdash \bot
        }
    \end{minipage}
    \begin{minipage}{.30\linewidth}
        \infrule[Rw$_\textsc{Label}$]{
            l_i \in \mathcal{L}
        }{
            \vdash \{\overline{l_i}\}
        }
    \end{minipage}
\end{rules}

\begin{rules}{Type environment well-formedness}{\vdash \Gamma}
    \begin{minipage}{.25\linewidth}
        \infrule[Tew$_\emptyset$]{
             \\
        }{
            \vdash \emptyset
        }
    \end{minipage}
    \begin{minipage}{.50\linewidth}
        \infrule[Tew$_\textsc{Lin}$]{
             \vdash \Gamma
             \andalso
             \vdash A
             \andalso
             x \notin \mathrm{dom}(\Gamma)
        }{
            \vdash \Gamma, x : A
        }
    \end{minipage}
    \begin{minipage}{.6\linewidth}
        \infrule[Tew$_\textsc{Gr}$]{
             \vdash \Gamma
             \andalso
             \vdash A
             \andalso
             \vdash r
             \andalso
             x \notin \mathrm{dom}(\Gamma)
        }{
            \vdash \Gamma, x : \verctype{A}{r}
        }
    \end{minipage}
\end{rules}

\begin{rules}{Resource environment well-formedness}{\vdash R}
    \begin{minipage}{.25\linewidth}
        \infrule[Rew$_{-}$]{
             \\
        }{
            \vdash -
        }
    \end{minipage}
    \begin{minipage}{.50\linewidth}
        \infrule[Rew$_r$]{
            \vdash_{\textsc{Rw}} r
        }{
            \vdash_{\textsc{Rew}} r
        }
    \end{minipage}
\end{rules}
\vspace{0.5\baselineskip}
where we use the notations $\vdash_{\textsc{Rw}}$ and $\vdash_{\textsc{Rew}}$ in $(\textsc{Rew}_r)$ to represent the  judgements of resource and resource environment well-formedness respectively, to avoid ambiguity between two syntactically indistinguishable judgements.

% var, pr, abs, 
\subsection{\corelang{} Type System (Declarative)}
\begin{rules}{\corelang{} typing}{\Gamma \vdash t:A}
    \begin{minipage}{.22\linewidth}
        \infrule[int]{
             \\% \\
        }{
            \emptyset \vdash n : \textsf{Int}
        }
    \end{minipage}
    \begin{minipage}{.25\linewidth}
        \infrule[var]{
            \vdash A\\
        }{
            x:A \vdash x:A
        }
    \end{minipage}
    \begin{minipage}{.44\linewidth}
        \infrule[abs]{
            - \vdash p : A \rhd \Delta
            \andalso
            \Gamma, \Delta \vdash t : B
        }{
            \Gamma \vdash \lam{p}{t} : A \rightarrow B
        }
    \end{minipage}
    \begin{minipage}{.55\linewidth}
        \infrule[let]{
            \Gamma_1 \,\vdash\, t_1 : \vertype{r}{A}
            \andalso
            \Gamma_2, x:\verctype{A}{r} \,\vdash\, t_2 : B
        }{
            \Gamma_1 + \Gamma_2 \,\vdash\, \clet{x}{t_1}{t_2} : B
        }
    \end{minipage}
    \begin{minipage}{.49\linewidth}
        \infrule[app]{
            \Gamma_1 \vdash t_1 : \ftype{A}{B}
            \andalso
            \Gamma_2 \vdash t_2 : A
        }{
            \Gamma_1 + \Gamma_2 \vdash \app{t_1}{t_2} : B
        }
    \end{minipage}
    % \\\\
    \begin{minipage}{.32\linewidth}
        \infrule[weak]{
            \Gamma \vdash t : A
            \andalso
            \vdash \Delta
        }{
            \Gamma, \verctype{\Delta}{0} \vdash t : A
        }
    \end{minipage}
    \begin{minipage}{.33\linewidth}
        \infrule[der]{
            \Gamma, x:A \vdash t : B
        }{
            \Gamma, x:\verctype{A}{1} \vdash t : B
        }
    \end{minipage}
    \begin{minipage}{.33\linewidth}
        \infrule[pr]{
            [\Gamma] \vdash t : A
            \andalso
            \vdash r
        }{
            r\cdot[\Gamma] \vdash [t] : \vertype{r}{A} 
        }
    \end{minipage}
    % \\\\
    \begin{minipage}{.55\linewidth}
        \infrule[sub]{
            \Gamma,x:\verctype{A}{r}, \Gamma' \vdash t : B
            \andalso
            r \sqsubseteq s
            \andalso
            \vdash s
        }{
            \Gamma,x:\verctype{A}{s}, \Gamma' \vdash t : B
        }
    \end{minipage}
    \begin{minipage}{.40\linewidth}
        \infrule[extr]{
            \Gamma \vdash u : \vertype{r}{A}
            \andalso
            l \in r
        }{
            \Gamma \vdash u.l : A
        }
    \end{minipage}
    % \\\\
    \begin{minipage}{.47\linewidth}
        \infrule[ver]{
            [\Gamma_i] \vdash t_i : A
            \andalso
            \vdash \{\overline{l_i}\}
        }{
            \bigcup\{l_i\}\cdot [\Gamma_i] \vdash \nvval{\overline{l_i=t_i}} : \vertype{\{\overline{l_i}\}}{A}
        }
    \end{minipage}
    % \\\\
    \begin{minipage}{.49\linewidth}
        \infrule[veri]{
            [\Gamma_i] \vdash t_i : A
            \andalso
            \vdash \{\overline{l_i}\}
            \andalso
            l_k \in \{\overline{l_i}\}
        }{
            \bigcup\{l_i\}\cdot [\Gamma_i] \vdash \ivval{\overline{l_i=t_i}}{l_k} : A
        }
    \end{minipage}
\end{rules}

\ \newline
\hbox{where $0=\bot$, $1 = \emptyset$ and $\sqsubseteq\,=\,\subseteq$.}

\begin{rules}{\corelang{} pattern typing}{R \vdash p : A \rhd \Delta}
    % \begin{minipage}{.29\linewidth}
    %     \infrule[pInt]{
    %         \vdash R\\
    %     }{
    %         R \vdash n:\inttype{} \rhd \emptyset
    %     }
    % \end{minipage}
    % \begin{minipage}{.29\linewidth}
    %     \infrule[pWild]{
    %         \vdash A
    %         \andalso
    %         \vdash R
    %     }{
    %         R \vdash \_ : A \rhd \emptyset
    %     }
    % \end{minipage}
    \begin{minipage}{.33\linewidth}
        \infrule[pVar]{
            \vdash A
        }{
            - \vdash x:A \rhd x:A
        }
    \end{minipage}
    \begin{minipage}{.38\linewidth}
        \infrule[\mbox{[}pVar\mbox{]}]{
            \vdash r
            \andalso
            \vdash A
        }{
            r \vdash x:A \rhd x:\verctype{A}{r}
        }
    \end{minipage}
    \begin{minipage}{.33\linewidth}
        \infrule[p$\square$]{
            r \vdash x : A \rhd \Delta
        }{
            - \vdash \pr{x} : \vertype{r}{A} \rhd \Delta
        }
    \end{minipage}
    % \begin{minipage}{.45\linewidth}
    %     \infrule[\mbox{[}p$\square$\mbox{]}]{
    %         r\otimes s \vdash p : A \rhd \Delta
    %     }{
    %         r \vdash \pr{p} : \vertype{s}{A} \rhd \Delta
    %     }
    % \end{minipage}
\end{rules}

\subsection{\corelang{} Dynamic Semantics}
\begin{rules}{Evaluation}{t \longrightarrow t'}
    \begin{minipage}{.95\linewidth}
        \infrule[]{
            t \leadsto t'
        }{
            E[t] \longrightarrow E[t']
        }
    \end{minipage}
\end{rules}

\begin{rules}{Reduction}{t \leadsto t'}
    \begin{minipage}{.475\linewidth}
        \infrule[E-abs1]{
            \\
        }{
            \app{(\lam{x}{t})}{t'} \leadsto (t' \rhd x)t
        }
    \end{minipage}
    \begin{minipage}{.49\linewidth}
        \infrule[E-abs2]{
            \\
        }{
            \app{(\lam{\pr{x}}{t})}{t'} \leadsto \clet{x}{t'}{t}
        }
    \end{minipage}
    % \begin{minipage}{.3\linewidth}
    %     \infrule[E-abs$_\_$]{
    %         \\
    %     }{
    %         \app{(\lam{\_}{t})}{t'} \leadsto t
    %     }
    % \end{minipage}
    % \begin{minipage}{.3\linewidth}
    %     \infrule[E-abs$_n$]{
    %         \\
    %     }{
    %         \app{(\lam{n}{t})}{t'} \leadsto t
    %     }
    % \end{minipage}
    \begin{minipage}{.25\linewidth}
        \infrule[E-ex1]{
            \\
        }{
            \pr{t}.l \leadsto t@l
        }
    \end{minipage}
    \begin{minipage}{.36\linewidth}
        \infrule[E-ex2]{
            \\
        }{
            \nvval{\overline{l_i=t_i}}.l_i \leadsto t_i@l_i
        }
    \end{minipage}
    \begin{minipage}{.50\linewidth}
        \infrule[E-clet]{
            \\
        }{
            \clet{x}{u}{t} \leadsto (u \rhd \pr{x})t
        }
    \end{minipage}
    \begin{minipage}{.40\linewidth}
        \infrule[E-veri]{
            \\
        }{
            \langle\overline{l_i=t_i}\,|\,l_k\rangle \leadsto t_k@l_k
        }
    \end{minipage}
\end{rules}

\paragraph{\textnormal{\textbf{Substitutions}}}
\begin{align*}
    (t' \rhd x) t \quad&=\quad [t'/x]t \\%\tag{$\rhd_{\mathrm{var}}$}
    ([t'] \rhd [x])t \quad&=\quad (t' \rhd x)t\\%\tag{$\rhd_{\square}$}\\
    (\nvval{\overline{l_i=t_i}} \rhd [x]) t \quad&=\quad [\ivval{\overline{l_i=t_i}}{l_k}/x] t \quad (l_k \in \{\overline{l_i}\}) %\tag{$\rhd_{\mathrm{ver}}$}
\end{align*}

\paragraph{\textnormal{\textbf{Default version overwriting}}}
\begin{align*}
n@l \quad&=\quad n\\
x@l \quad&=\quad x\\
(\lam{p}{t})@l \quad&=\quad \lam{p}{(t@l)}\\
(\app{t}{u})@l \quad&=\quad \app{(t@l)}{(u@l)}\\
(\clet{x}{t_1}{t_2})@l \quad&=\quad \clet{x}{(t_1@l)}{(t_2@l)}\\
\pr{t}@l \quad&=\quad \pr{t}\\
\nvval{\overline{l_i=t_i}}@l \quad&=\quad \nvval{\overline{l_i=t_i}}\\
(u.l')@l \quad&=\quad (u@l).l'\\
\ivval{\overline{l_i=t_i}}{l_i}@l' \quad&=\quad \left\{ 
\begin{aligned}
\ivval{\overline{l_i=t_i}}{l'} \quad (l' \in \{\overline{l_i}\})\\
\ivval{\overline{l_i=t_i}}{l_i} \quad (l' \notin \{\overline{l_i}\})
\end{aligned}
\right.
\end{align*}

\clearpage
\section{\vlmini{} Definitions}

\subsection{\vlmini{} Syntax (w/o version control terms/data constructors)}

\paragraph{\textnormal{\textbf{\vlmini{} syntax}}}
\begin{align*}
t \quad&::=\quad n \mid x \mid \app{t_1}{t_2} \mid \lam{p}{t} \mid \pr{t} \tag{terms}\\
p \quad&::=\quad x \mid [x] \tag{patterns}\\
A, B \quad&::=\quad \alpha \mid \inttype{} \mid \ftype{A}{B} \mid \vertype{r}{A} \tag{types}\\
\kappa \quad&::=\quad \typekind{} \mid \labelskind{} \tag{kinds}\\
r      \quad&::=\quad \alpha \mid  \bot \mid \{\overline{l_i}\} \tag{version resources}\\
\mathcal{L}\ \ni\ l\quad&::=\quad \{\overline{M_i = V_i}\} \tag{version labels}
\end{align*}
where $M_i\in\mathcal{M}$ and $V_i\in\mathcal{V}_{M_i}$ are metavariables over module names and versions of $M_i$, respectively.

\paragraph{\textnormal{\textbf{\vlmini{} contexts}}}
\begin{align*}
\Gamma,\Delta \quad&::=\quad \emptyset \mid \Gamma, x:A \mid \Gamma, x:\verctype{A}{r} \tag{contexts}\\
\Sigma \quad&::=\quad \emptyset \mid \Sigma, \alpha:\kappa \tag{type variable kinds}\\
R      \quad&::=\quad - \mid r \tag{resource contexts}
\end{align*}

\paragraph{\textnormal{\textbf{\vlmini{} constraints}}}
\begin{align*}
\mathcal{C} \quad&::=\quad \top \mid \mathcal{C}_1 \land \mathcal{C}_2 \mid \mathcal{C}_1 \lor \mathcal{C}_2 \mid \alpha \preceq \alpha' \mid \alpha \preceq \mathcal{D} \tag{dependency constraints}\\
\mathcal{D} \quad&::=\quad \cs{\overline{M_i = V_i}} \tag{dependent labels}\\
\Theta      \quad&::=\quad \top \mid \Theta_1 \land \Theta_2 \mid \{A \sim B\} \tag{type constraints}
\end{align*}

\paragraph{\textnormal{\textbf{\vlmini{} type substitutions}}}
\begin{align*}
\theta \quad&::=\quad \emptyset \mid \theta \circ [\alpha \mapsto A] \mid \theta \circ [\alpha \mapsto r] \tag{type substitutions}\\
\eta   \quad&::=\quad \emptyset \mid \eta\circ[\alpha \mapsto \{l\}] \tag{label substituions} 
\end{align*}

\subsection{\vlmini{} Well-formedness and Kinding}
\begin{rules}{Type variable kinds well-formedness}{\vdash \Sigma}
    \begin{minipage}{.25\linewidth}
        \infrule[Kw$_\emptyset$]{
            \\
        }{
            \vdash \emptyset
        }
    \end{minipage}
    \begin{minipage}{.70\linewidth}
        \infrule[Kw$_\alpha$]{
            \vdash \Sigma
            \andalso
            \kappa \in \{\typekind,\,\labelskind\}
            \andalso
            \alpha \notin \mathrm{dom}(\Sigma)
        }{
            \vdash \Sigma, \alpha : \kappa
        }
    \end{minipage}
\end{rules}

\begin{rules}{\vlmini{} kinding}{\Sigma \vdash A:\kappa}
    \begin{minipage}{.30\linewidth}
        \infrule[$\kappa_\textsc{Int}$]{
            \vdash \Sigma
        }{
            \Sigma \vdash \inttype : \typekind
        }
    \end{minipage}
    \begin{minipage}{.45\linewidth}
        \infrule[$\kappa_\rightarrow$]{
             \Sigma \vdash A : \typekind
             \andalso
             \Sigma \vdash B : \typekind
        }{
            \Sigma \vdash \ftype{A}{B} : \typekind
        }
    \end{minipage}
    \begin{minipage}{.50\linewidth}
        \infrule[$\kappa_\square$]{
            \Sigma \vdash r : \labelskind
            \andalso
            \Sigma \vdash A : \typekind
        }{
            \Sigma \vdash \vertype{r}{A} : \typekind
        }
    \end{minipage}
    \begin{minipage}{.35\linewidth}
        \infrule[$\kappa_\alpha$]{
            \vdash \Sigma
            \andalso
            \Sigma(\alpha) = \kappa
        }{
            \Sigma \vdash \alpha : \kappa
        }
    \end{minipage}
    \begin{minipage}{.30\linewidth}
        \infrule[$\kappa_\bot$]{
            \vdash \Sigma
        }{
            \Sigma \vdash \bot : \labelskind
        }
    \end{minipage}
    \begin{minipage}{.35\linewidth}
        \infrule[$\kappa_\textsc{Label}$]{
            \vdash \Sigma
            \andalso
            l_i \in \mathcal{L} \\
        }{
            \Sigma \vdash \{\overline{l_i}\} : \labelskind
        }
    \end{minipage}
\end{rules}

\begin{rules}{Type environment well-formedness}{\Sigma \vdash A}
    \begin{minipage}{.19\linewidth}
        \infrule[Tew$_\emptyset$]{
            \vdash \Sigma
        }{
            \Sigma \vdash \emptyset
        }
    \end{minipage}
    \begin{minipage}{.65\linewidth}
        \infrule[Tew$_\textsc{Lin}$]{
             \Sigma \vdash \Gamma
             \andalso
             \Sigma \vdash A : \typekind
             \andalso
             x \notin \mathrm{dom}(\Gamma)
        }{
            \Sigma \vdash \Gamma, x : A
        }
    \end{minipage}
    \begin{minipage}{.80\linewidth}
        \infrule[Tew$_\textsc{Gr}$]{
             \Sigma \vdash \Gamma
             \andalso
             \Sigma \vdash A : \typekind
             \andalso
             \Sigma \vdash r : \labelskind
             \andalso
             x \notin \mathrm{dom}(\Gamma)
        }{
            \Sigma \vdash \Gamma, x : \verctype{A}{r}
        }
    \end{minipage}
\end{rules}

\begin{rules}{Resource environment well-formedness}{\Sigma \vdash R}
    \begin{minipage}{.25\linewidth}
        \infrule[Rew$_{-}$]{
            \vdash \Sigma
        }{
            \Sigma \vdash -
        }
    \end{minipage}
    \begin{minipage}{.50\linewidth}
        \infrule[Rew$_r$]{
            \vdash \Sigma
            \andalso
            \Sigma \vdash r : \labelskind
        }{
            \Sigma \vdash r
        }
    \end{minipage}
\end{rules}

\begin{rules}{Type substitutions well-formedness}{\Sigma \vdash \theta}
    \begin{minipage}{.25\linewidth}
        \infrule[Sw$_{\emptyset}$]{
            \vdash \Sigma
        }{
            \Sigma \vdash \emptyset
        }
    \end{minipage}
    \begin{minipage}{.65\linewidth}
        \infrule[Sw$_\textsc{Ty}$]{
            \Sigma \vdash \theta
            \andalso
            \Sigma \vdash \alpha : \typekind
            \andalso
            \Sigma \vdash A : \typekind
        }{
            \Sigma \vdash \theta \circ [\alpha \mapsto A]
        }
    \end{minipage}
    \begin{minipage}{.65\linewidth}
        \infrule[Sw$_\textsc{Res}$]{
            \Sigma \vdash \theta
            \andalso
            \Sigma \vdash \alpha : \labelskind
            \andalso
            \Sigma \vdash r : \labelskind
        }{
            \Sigma \vdash \theta \circ [\alpha \mapsto r]
        }
    \end{minipage}
\end{rules}

\subsection{\vlmini{} Algorithmic Type Inference System}
\label{appendix:typing_vlmini}

\begin{rules}{\vlmini{} type synthesis}{\Sigma;\Gamma \vdash t \Rightarrow A;\Sigma'; \Delta; \Theta; \mathcal{C}}
    \begin{minipage}{.45\linewidth}
      \infrule[$\Rightarrow_{\textsc{int}}$]{
        \vdash \Sigma
        \andalso
        \Sigma \vdash \Gamma
      }{
        \Sigma; \Gamma \vdash n \Rightarrow \inttype; \Sigma; \emptyset; \top; \top
      }
    \end{minipage}
    \begin{minipage}{.5\linewidth}
      \infrule[$\Rightarrow_{\textsc{lin}}$]{
        \vdash \Sigma
        \andalso
        \Sigma \vdash \Gamma
        \andalso
        x:A\in\Gamma
      }{
        \Sigma; \Gamma \vdash x \Rightarrow A; \Sigma; x:A; \top; \top
      }
    \end{minipage}
    \begin{minipage}{.5\linewidth}
      \infrule[$\Rightarrow_{\textsc{gr}}$]{
        \vdash \Sigma
        \andalso
        \Sigma \vdash \Gamma
        \andalso
        x:\verctype{A}{r}\in\Gamma
      }{
        \Sigma; \Gamma \vdash x \Rightarrow A; \Sigma; x:\verctype{A}{1}; \top; \top
      }
    \end{minipage}
    \begin{minipage}{.95\linewidth}
      \infrule[$\Rightarrow_{\textsc{abs}}$]{
        \Sigma_1, \alpha:\textsf{Type};- \vdash p:\alpha \rhd \Gamma'; \Sigma_2; \Theta_1
        \andalso
        \Sigma_2;\Gamma,\Gamma' \vdash t \Rightarrow B;\Sigma_3;\Delta; \Theta_2; \mathcal{C}
      }{
        \Sigma_1;\Gamma \vdash \lam{p}{t} \Rightarrow \ftype{\alpha}{B};\Sigma_3;\Delta\backslash\Gamma' ; \Theta_1\land\Theta_2; \mathcal{C}
      }
    \end{minipage}
    \begin{minipage}{.95\linewidth}
      \vspace{0.5\baselineskip}
      \infrule[$\Rightarrow_{\textsc{app}}$]{
        \Sigma_1; \Gamma \vdash t_1 \Rightarrow A_1 ; \Sigma_2; \Delta_1; \Theta_1; \mathcal{C}_1
        \andalso
        \Sigma_2; \Gamma \vdash t_2 \Rightarrow A_2; \Sigma_3; \Delta_2; \Theta_2; \mathcal{C}_2
      }{
        \Sigma_1;\Gamma \vdash \app{t_1}{t_2} \Rightarrow \beta; \Sigma_3, \beta:\typekind; \Delta_1+\Delta_2; \\\hspace{11em}\Theta_1\land\Theta_2\land\{A_1\sim \ftype{A_2}{\beta}\}; \mathcal{C}_1 \land \mathcal{C}_2
      }
    \end{minipage}
    \begin{minipage}{.95\linewidth}
      \vspace{0.5\baselineskip}
      \infrule[$\Rightarrow_{\textsc{pr}}$]{
        \Sigma_1 \vdash [\Gamma\cap\textsf{FV}(t)]_{\textsf{Labels}} \rhd \Gamma'
        \andalso
        \Sigma_1; \Gamma' \vdash t \Rightarrow A; \Sigma_2; \Delta; \Theta; \mathcal{C}_1
        \andalso\\
        \Sigma_3 = \Sigma_2, \alpha:\textsf{Labels}
        \andalso
        \Sigma_3 \vdash \alpha \sqsubseteq_c \Gamma' \rhd \mathcal{C}_2
      }{
        \Sigma_1;\Gamma \vdash [t] \Rightarrow \vertype{\alpha}{A}; \Sigma_3; \alpha \cdot \Delta ;\Theta; \mathcal{C}_1 \land \mathcal{C}_2
      }
    \end{minipage}
\end{rules}

\begin{rules}{\vlmini{} pattern type synthesis}{\Sigma; R\vdash p : A \rhd \Gamma; \Sigma'; \Theta; \mathcal{C}}
    % \begin{minipage}{.55\linewidth}
    %   \infrule[pInt]{
    %     \vdash \Sigma
    %     \andalso
    %     \Sigma \vdash R
    %     \andalso
    %     \Sigma \vdash A:\typekind
    %   }{
    %     \Sigma; R \vdash n : A \rhd \emptyset; \Sigma; \{A\sim\inttype\}; \top
    %   }
    % \end{minipage}
    % \begin{minipage}{.45\linewidth}
    %   \infrule[p\_]{
    %     \vdash \Sigma
    %     \andalso
    %     \Sigma \vdash R
    %     \andalso
    %     \Sigma \vdash A : \textsf{Type}\\
    %   }{
    %     \Sigma; R \vdash \_ : A \rhd \emptyset; \Sigma; \top; \top
    %   }
    % \end{minipage}
    \begin{minipage}{.50\linewidth}
      \infrule[pVar]{
        \vdash \Sigma
        \andalso
        \Sigma \vdash A : \textsf{Type}
      }{
        \Sigma; - \vdash x : A \rhd x:A; \Sigma; \top; \top
      }
    \end{minipage}
    \begin{minipage}{.6\linewidth}
      \infrule[\mbox{[}pVar\mbox{]}]{
        \vdash \Sigma
        \andalso
        \Sigma \vdash A : \textsf{Type}
        \andalso
        \Sigma \vdash r : \textsf{Labels}
      }{
        \Sigma; r \vdash x : A \rhd x:\verctype{A}{r}; \Sigma; \top; \top
      }
    \end{minipage}
    \begin{minipage}{.7\linewidth}
      \vspace{0.75\baselineskip}
      \infrule[p$\square$]{
        \Sigma, \alpha:\textsf{Labels}, \beta:\textsf{Type}; \alpha \vdash x : \beta \rhd \Delta; \Sigma'; \Theta; \mathcal{C}
      }{
        \Sigma; - \vdash [x] : A \rhd \Delta; \Sigma'; \Theta \land \{A \sim \vertype{\alpha}{\beta}\}; \mathcal{C}
      }
    \end{minipage}
    % %
    % \begin{minipage}{.85\linewidth}
    %   \vspace{0.75\baselineskip}
    %   \infrule[\mbox{[}p$\square$\mbox{]}]{
    %     \Sigma, \alpha:\textsf{Labels}, \beta:\textsf{Type}; r \otimes \alpha \vdash p : \beta \rhd \Delta; \Sigma'; \Theta; \mathcal{C}
    %     % \textsf{flatten}(r,R,r',R')=(s,S)
    %   }{
    %     \Sigma; r \vdash [p] : A \rhd \Delta; \Sigma'; \Theta\land \{A \sim \vertype{\alpha}{\beta}\}; \mathcal{C}
    %   }
    % \end{minipage}
\end{rules}

\begin{rules}{\vlmini{} context grading}{\Sigma \vdash [\Gamma]_{\textsf{Labels}} \rhd \Gamma'}
        \begin{minipage}{.30\linewidth}
            \infrule[$\emptyset$]{
                \\
            }{
                \Sigma \vdash [\emptyset]_{\textsf{Labels}} \rhd \emptyset
            }
        \end{minipage}
        \begin{minipage}{.5\linewidth}
            \infrule[{[}lin{]}]{
                \Sigma \vdash [\Gamma]_{\textsf{Labels}} \rhd \Gamma'
            }{
                \Sigma \vdash [\Gamma,x:A]_{\textsf{Labels}} \rhd \Gamma', x:[A]_1
            }
        \end{minipage}
        \begin{minipage}{.55\linewidth}
            \infrule[{[}gr{]}]{
                \Sigma \vdash [\Gamma]_{\textsf{Labels}} \rhd \Gamma'
            }{
                \Sigma \vdash [\Gamma,x:[A]_r]_{\textsf{Labels}} \rhd \Gamma', x:[A]_{r}
            }
        \end{minipage}
\end{rules}

\ \newline
\hbox{where $1 = \emptyset$.}

\begin{rules}{\vlmini{} constraints generation}{\Sigma \vdash \alpha \sqsubseteq_{c} [\Gamma] \rhd \mathcal{C}}
    \begin{minipage}{.35\linewidth}
      \infrule[$\sqsubseteq^{\textsc{VD}}_\emptyset$]{
        % \Sigma \vdash \alpha : \textsf{Labels}
        \\
      }{
        \Sigma \vdash \alpha \sqsubseteq_{c} \emptyset \rhd \top
      }
    \end{minipage}
    \begin{minipage}{.60\linewidth}
      \infrule[$\sqsubseteq^{\textsc{VD}}_\Gamma$]{
        % \Sigma \vdash \alpha : \textsf{Labels}
        % \andalso
        \Sigma \vdash \alpha \sqsubseteq_{c} [\Gamma] \rhd \mathcal{C}
      }{
        \Sigma \vdash \alpha \sqsubseteq_{c} ([\Gamma], x:[A]_r) \rhd (\mathcal{C} \land (\alpha \preceq r))
      }
    \end{minipage}
  % \caption{\vlmini{} constraints generation}
  % %\ecaption{Typing of \corelang}
  % % \Description{\corelang algorithmic typing}
  % \label{fig:rule_constraints_generation}
\end{rules}

\begin{rules}{\vlmini{} type unification}{\Sigma \vdash A \sim B \rhd \theta}
    \begin{minipage}{.45\linewidth}
      \infrule[$U_{\alpha}$]{
        \Sigma \vdash \alpha:\typekind
        \andalso
        \Sigma \vdash A:\typekind
      }{
        \Sigma \vdash \alpha \sim A \rhd \alpha \mapsto A
      }
    \end{minipage}
    % \begin{minipage}{.35\linewidth}
    %   \infrule[$U_{\alpha=}$]{
    %     \Sigma \vdash \alpha:\typekind\\
    %   }{
    %     \Sigma \vdash \alpha \sim \alpha \rhd \emptyset
    %   }
    % \end{minipage}
    \begin{minipage}{.35\linewidth}
      \infrule[$U_{=}$]{
        \Sigma \vdash A:\typekind
      }{
        \Sigma \vdash A \sim A \rhd \emptyset      }
    \end{minipage}
    \begin{minipage}{.62\linewidth}
      \infrule[$U_\rightarrow$]{
        \Sigma \vdash A' \sim A \rhd \theta_1
        \andalso
        \Sigma \vdash \theta_1B\sim \theta_1B' \rhd \theta_2
      }{
        \Sigma \vdash A \rightarrow B \sim A' \rightarrow B' \rhd \theta_1 \uplus \theta_2
      }
    \end{minipage}
    %  &
    % \begin{minipage}{.20\linewidth}
    %   \infrule[$U_{\textsc{Con}}$]{
    %     \\
    %   }{
    %     \Sigma \vdash K \sim K \rhd \emptyset
    %   }
    % \end{minipage}
    \begin{minipage}{.6\linewidth}
      \infrule[$U_{\Box}$]{
        \Sigma \vdash A \sim A' \rhd \theta_1
        \andalso
        \Sigma \vdash \theta_1r \sim \theta_1r' \rhd \theta_2
      }{
        \Sigma \vdash \vertype{r}{A} \sim \vertype{r'}{A'} \rhd \theta_1 \uplus \theta_2
        % \Sigma \vdash \vertype{r}{A} \sim \vertype{r'}{A'} \rhd \theta_1
      }
    \end{minipage}
\end{rules}
\vspace{\baselineskip}
otherwise fail.

\begin{rules}{\vlmini{} unification}{\Sigma \vdash \Theta \rhd \theta}
    \begin{minipage}{.25\linewidth}
      \infrule[$U_\emptyset$]{
        \\
      }{
        \Sigma \vdash \top \rhd \emptyset
      }
    \end{minipage}
    \begin{minipage}{.55\linewidth}
      \infrule[$U_\Theta$]{
        \Sigma \vdash \Theta \rhd \theta_1
        \andalso
        \Sigma \vdash \theta_1A \sim \theta_1B \rhd \theta_2
      }{
        \Sigma \vdash \Theta \land \{A \sim B\} \rhd  \theta_1 \uplus \theta_2
      }
    \end{minipage}
\end{rules}

\paragraph{\textnormal{\textbf{Type substitutions}}}
\begin{align*}
\theta \inttype{} \quad&= \quad \inttype{}\\
\theta\alpha \quad&= \quad
    \left\{\begin{aligned}
    A \hspace{1.5em} (\theta(\alpha)=A)\\
    \alpha \hspace{1.5em} (\mathrm{otherwise})
    \end{aligned}\right.\\
\theta(\ftype{A}{B}) \quad&= \quad \ftype{\theta A}{\theta B}\\
% \theta(\app{A}{B}) \quad&= \quad \app{(\theta A)}{(\theta B)}\\
% \theta(A \op B) \quad&= \quad (\theta A) \op (\theta B)\\
\theta(\vertype{r}{A}) \quad&= \quad \vertype{(\theta r)}{(\theta A)}\\
\theta \bot \quad&= \quad \bot \\
\theta \{\overline{l_i}\} \quad&= \quad \{\overline{l_i}\}
% \theta(r_1\otimes r_2) \quad&= \quad (\theta r_1)\otimes (\theta r_2)\\
% \theta(r_1\oplus r_2) \quad&= \quad (\theta r_1)\oplus (\theta r_2)
\end{align*}

\paragraph{\textnormal{\textbf{Substitution compositions}}}
\begin{align*}
\emptyset \uplus \theta_2 &= \theta_2\\
(\theta_1, \alpha \mapsto A) \uplus \theta_2 &= \left\{\begin{aligned}
    &(\theta_1 \uplus (\theta_2 \backslash \alpha) \uplus \theta), \alpha \mapsto \theta A &\theta_2(\alpha) = B \land \Sigma \vdash A \sim B \rhd \theta\\
    &(\theta_1 \uplus \theta_2),\alpha \mapsto A &\alpha\notin  \mathrm{dom}(\theta_2)
\end{aligned} \right.
\end{align*}

\subsection{Extensions for Version Control Terms}
\label{appendix:vlmini_version_control_terms}
\paragraph{\textnormal{\textbf{\vlmini{} syntax}}}
\begin{align*}
t \quad&::=\quad \ldots \mid \verof{l}{t} \mid \unver{t}\tag{terms}
\end{align*}

\begin{rules}{\vlmini{} algorithmic type synthesis}{\Sigma;\Gamma \vdash t \Rightarrow A;\Sigma'; \Delta; \Theta; \mathcal{C}}
    \begin{minipage}{.95\linewidth}
      \vspace{0.5\baselineskip}
      \infrule[$\Rightarrow_{\textsc{Ver}}$]{
        l \in \mathcal{L}
        \andalso
        \Sigma_1 \vdash [\Gamma\cap\textsf{FV}(t)]_{\textsf{Labels}} \rhd \Gamma'
        \andalso\\
        \Sigma_1 \vdash \Gamma' \sqsubseteq_c \cs{l} \rhd \mathcal{C}_2
        \andalso
        \Sigma_1; \Gamma' \vdash t \Rightarrow A; \Sigma_2; \Delta_1; \Theta_1; \mathcal{C}_1
      }{
        \Sigma_1;\Gamma \vdash \verof{l}{t} \Rightarrow A; \Sigma_2; \Delta_1 ; \Theta_1 ; \mathcal{C}_1 \land \mathcal{C}_2
      }
    \end{minipage}
    \begin{minipage}{.95\linewidth}
      \vspace{0.5\baselineskip}
      \infrule[$\Rightarrow_{\textsc{Unver}}$]{
        \Sigma_1; \Gamma \vdash t \Rightarrow A; \Sigma_2; \Delta; \Theta; \mathcal{C}
        \andalso\\
        A = \vertype{r}{A'}
        \andalso
        \Sigma_3 = \Sigma_2,\alpha:\labelskind
        \andalso
        \vdash \Sigma_3
      }{
        \Sigma_1;\Gamma \vdash \unver{t} \Rightarrow \vertype{\alpha}{A'}; \Sigma_3; \Delta; \Theta; \mathcal{C}
      }
    \end{minipage}
\end{rules}

\begin{rules}{\vlmini{} constraints generation}{\Sigma \vdash [\Gamma] \sqsubseteq_{c} \mathcal{D} \rhd \mathcal{C}}
    \begin{minipage}{.35\linewidth}
      \infrule[$\sqsubseteq^{\textsc{LD}}_{\emptyset}$]{
        % \Sigma \vdash \alpha : \textsf{Labels}
        \\
      }{
        \Sigma \vdash \emptyset \sqsubseteq_{c} \mathcal{D} \rhd \top
      }
    \end{minipage}
    \begin{minipage}{.60\linewidth}
      \infrule[$\sqsubseteq^{\textsc{LD}}_{\Gamma}$]{
        % \Sigma \vdash \alpha : \textsf{Labels}
        % \andalso
        \Sigma \vdash [\Gamma] \sqsubseteq_{c} \mathcal{D} \rhd \mathcal{C}
      }{
        \Sigma \vdash ([\Gamma], x:[A]_r) \sqsubseteq_{c} \mathcal{D} \rhd \left(\mathcal{C}\land (r \preceq \mathcal{D})\right)
      }
    \end{minipage}
  % \caption{\vlmini{} constraints generation}
  % %\ecaption{Typing of \corelang}
  % % \Description{\corelang algorithmic typing}
  % \label{fig:rule_constraints_generation}
\end{rules}
\vspace{\baselineskip}
Note that type environment resources in \vlmini{} always contain only type variables, so $r = \alpha~(\exists\alpha)$.

\subsection{Extensions for Data Constructors}
\label{appendix:vlmini_data_structures}
\paragraph{\textnormal{\textbf{\vlmini{} syntax}}}
\begin{align*}
t \quad&::=\quad \ldots \mid C\,\overline{t_i} \mid \caseof{t}{\overline{p_i \mapsto t_i}}\tag{terms}\\
p \quad&::=\quad \ldots \mid C\,\overline{p_i} \tag{patterns}\\
C \quad&::=\quad (,) \mid [,] \tag{constructors}\\
A, B \quad&::=\quad \ldots \mid K\,\overline{A_i} \tag{types}\\
K    \quad&::=\quad (,) \mid [,] \tag{type constructors}
\end{align*}

\begin{rules}{\vlmini{} algorithmic type synthesis}{\Sigma;\Gamma \vdash t \Rightarrow A;\Sigma'; \Delta; \Theta; \mathcal{C}}
    \begin{minipage}{.95\linewidth}
      \infrule[$\Rightarrow_{\textsc{(,)}}$]{
        \Sigma_{i-1};\Gamma \vdash t_i \Rightarrow A_i; \Sigma_i; \Delta_i; \Theta_i; \mathcal{C}_i
      }{
        \Sigma_0;\Gamma \vdash (t_1,\,..\,,t_n) \Rightarrow (A_1,\,..\,,A_n); \Sigma_{n}; \Delta_1,..\Delta_n; \bigwedge\Theta_{i}; \bigwedge \mathcal{C}_i
      }
    \end{minipage}
    \begin{minipage}{.95\linewidth}
      \infrule[$\Rightarrow_{\textsc{[,]}}$]{
        \Sigma_{i-1};\Gamma \vdash t_i \Rightarrow A; \Sigma_i; \Delta_i; \Theta_i; \mathcal{C}_i
      }{
        \Sigma_0;\Gamma \vdash [t_1,\,..\,,t_n] \Rightarrow [A]; \Sigma_{n}; \Delta_1,..\Delta_n; \bigwedge\Theta_{i}; \bigwedge\mathcal{C}_i
      }
    \end{minipage}
    \begin{minipage}{.95\linewidth}
      \vspace{.5\baselineskip}
      \infrule[$\Rightarrow_{\textsc{case}}$]{
        \Sigma_{0};\Gamma \vdash t \Rightarrow A; \Sigma_1; \Gamma'_1; \Theta_0; \mathcal{C}_0
        \andalso\\
        \Sigma_{i-1};- \vdash p_{i}:A \rhd \Delta_{i}; \Sigma'_{i}; \Theta'_{i}
        \andalso
        \Sigma'_{i};\Gamma,\Delta_{i} \vdash t_{i} \Rightarrow B;\Sigma_{i}; \Delta'_{i}; \Theta_{i}; \mathcal{C}_{i}
      }{
        \Sigma_0;\Gamma \vdash \caseof{t}{\overline{p_i\mapsto t_i}} \Rightarrow B;\Sigma_{n}; \Gamma_1'+\bigcup(\Delta'_i\backslash\Delta_i); \bigwedge\Theta_{i}; \bigwedge \mathcal{C}_i
      }
    \end{minipage}
\end{rules}

\begin{rules}{\vlmini{} pattern type synthesis}{\Sigma; R\vdash p : A \rhd \Gamma; \Sigma'; \Theta; \mathcal{C}}
    \begin{minipage}{.95\linewidth}
      \vspace{.5\baselineskip}
      \infrule[pCon]{
        \Sigma_{i-1}' = \Sigma_{i-1}, \alpha_i:\textsf{Type}
        \andalso
        \Sigma_{i-1}'; - \vdash p_i : \alpha_i \rhd \Gamma_i; \Sigma_i; \Theta_i; \mathcal{C}
      }{
        \Sigma_0; - \vdash C\,p_1\,..\,p_n : A \rhd \Gamma_i, .. ,\Gamma_n; \Sigma_n; \{K\,\overline{\alpha_i} \sim A\} \land \bigwedge\Theta_i; \mathcal{C}
      }
    \end{minipage}
    \begin{minipage}{.95\linewidth}
      \vspace{.5\baselineskip}
      \infrule[\mbox{[}pCon\mbox{]}]{
        \Sigma_{i-1}' = \Sigma_{i-1}, \alpha_i:\textsf{Type}
        \andalso\\
        \Sigma_{i-1}'; r \vdash p_i : \alpha_i \rhd \Gamma_i; \Sigma_i; \Theta_i; \mathcal{C}
        \andalso
        \Sigma_i' \vdash r : \textsf{Labels}
      }{
        \Sigma_0; r \vdash C\,p_1\,..\,p_n : A \rhd \Gamma_i, .. ,\Gamma_n; \Sigma_n; \{K\,\overline{\alpha_i} \sim A\}\land\bigwedge\Theta_{i}; \mathcal{C}
      }
    \end{minipage}
\end{rules}

\begin{rules}{\vlmini{} type unification}{\Sigma \vdash A \sim B \rhd \theta}
    \begin{minipage}{.95\linewidth}
      \infrule[$U_\textsc{Con}$]{
        \Sigma \vdash A_1 \sim B_1 \rhd \theta_1
        \andalso
        \Sigma \vdash \theta_{i-1} A_i \sim \theta_{i-1} B_i \rhd \theta_i~(i\geq 2)
      }{
        \Sigma \vdash K\,\overline{A_i} \sim K\,\overline{B_i} \rhd \biguplus\theta_{i}}
    \end{minipage}
\end{rules}

\paragraph{\textnormal{\textbf{Type substitutions}}}
\begin{align*}
\theta (K\,\overline{A_i}) \quad&= \quad K\,\overline{\theta A_i}
\end{align*}
\clearpage
% center 環境元に戻す
\def\center{\trivlist \centering\item\relax}
\def\endcenter{\endtrivlist}

\section{\corelang{} Type Safety}
\label{appendix:lambdavl_safety}
\subsection{Resource Properties}
\begin{definition}[Version resource semiring]
The version resource semiring is given by the structural semiring (semiring with preorder) $(\mathcal{R},\oplus,0,\otimes,1,\sqsubseteq)$, defined as follows.
\begin{gather*}
    0 = \bot
    \quad
    1 = \emptyset
    \quad
    \begin{minipage}{0.10\hsize}
        \infax{\bot \sqsubseteq \texttt{r}}
    \end{minipage}
    \quad
    \begin{minipage}{0.10\hsize}
        \infrule{
            \texttt{r}_1 \subseteq \texttt{r}_2
        }{
            \texttt{r}_1 \sqsubseteq \texttt{r}_2
        }
    \end{minipage}
    \\
    \begin{aligned}
        r_1 \oplus r_2 &=
        \left\{
        \begin{aligned}
            &r_1 & &\mbox{$r_2 = \bot$}\\
            &r_2 & &\mbox{$r_1 = \bot$}\\
            &r_1 \cup r_2 & &\mbox{otherwise}
        \end{aligned}
        \right.
        \quad
        r_1 \otimes r_2 &=
        \left\{
        \begin{aligned}
            &\bot & &\mbox{$r_1 = \bot$}\\
            &\bot & &\mbox{$r_2 = \bot$}\\
            &r_1 \cup r_2 & &\mbox{otherwise}
        \end{aligned}
        \right.
    \end{aligned}
\end{gather*}
where $\bot$ is the smallest element of $\mathcal{R}$, and $r_1 \subseteq r_2$ is the standard subset relation over sets defined only when both $r_1$ and $r_2$ are not $\bot$.\\
\end{definition}

\begin{lemma}[Version resource semiring is a structural semiring]
\label{proof:semiring}
\end{lemma}
\begin{proof}
Version resource semiring $(\mathcal{R},\oplus,\bot,\otimes,\emptyset,\sqsubseteq)$ induces a semilattice with $\oplus$ (join).
\begin{itemize}
\item $(\mathcal{R},\oplus,\bot,\otimes,\emptyset)$ is a semiring, that is:
    \begin{itemize}
        \item $(\mathcal{R},\oplus,\bot)$ is a commutative monoid, i.e., for all $p,q,r\in\mathcal{R}$
        \begin{itemize}
            \item (Associativity) $(p\oplus q) \oplus r = p\oplus (q \oplus r)$ holds since $\oplus$ is defined in associative manner with $\bot$.
            \item (Commutativity) $p\oplus q = q\oplus p$  holds since $\oplus$ is defined in commutative manner with $\bot$.
            \item (Identity element) $\bot \oplus p = p \oplus \bot = p$
        \end{itemize}
        \item $(\mathcal{R},\otimes,\emptyset)$ is a monoid, i.e., for all $p,q,r\in\mathcal{R}$
        \begin{itemize}
            \item (Associativity) $(p\otimes q) \otimes r = p\otimes (q \otimes r)$ holds since $\oplus$ is defined in associative manner with $\bot$.
            \item (Identity element) $\emptyset \otimes p = p \otimes \emptyset = p$
            \begin{itemize}
                \item if $p = \bot$ then $\emptyset \otimes \bot = \bot \otimes \emptyset = \bot$
                \item otherwise if $p \neq \bot$ then $\emptyset \otimes p = \emptyset \cup p = p$ and $p \otimes \emptyset = p \cup \emptyset = p$
            \end{itemize}
        \end{itemize}
        \item multiplication $\otimes$ distributes over addition $\oplus$, i.e., for all $p, q, r \in \mathcal{R},r\otimes(p\oplus q) = (r\otimes p) \oplus (r \otimes q)$ and $(p\oplus q)\otimes r = (p\otimes r) \oplus (q \otimes r)$
        \begin{itemize}
            \item if $r = \bot$ then $r\otimes(p\oplus q) = \bot$ and $(r\otimes p) \oplus (r \otimes q) = \bot \oplus \bot = \bot$.
            \item otherwise if $r\neq \bot$ and $p = \bot$ and $q \neq \bot$ then $r\otimes(p\oplus q) = r\otimes q = r\cup q = (r\cup r) \cup q = r\cup (r \cup q) = (r \oplus p) \cup (r \cup q) = (r\otimes p) \oplus (r \otimes q)$
            \item otherwise if $r\neq \bot$ and $p = \bot$ and $q = \bot$ then $r\otimes(p\oplus q) = r\otimes \bot = \bot$ and $(r\otimes p) \oplus (r \otimes q) = \bot \oplus \bot = \bot$.
            \item otherwise if $r\neq \bot$ and $p \neq \bot$ and $q \neq \bot$ then $r\otimes(p\oplus q) = r\cup(p\cup q) = (r\cup p)\cup (r\cup q) = (r\otimes p) \oplus (r \otimes q)$
        \end{itemize}
        The other cases are symmetrical cases.
        \item $\bot$ is absorbing for multiplication: $p\otimes \bot= \bot \otimes p = \bot$ for all $p\in\mathcal{R}$
    \end{itemize}
\item $(\mathcal{R},\sqsubseteq)$ is a bounded semilattice, that is
    \begin{itemize}
        \item $\sqsubseteq$ is a partial order on $\mathcal{R}$ such that the least upper bound of every two elements $p,q \in \mathcal{R}$ exists and is denoted by $p\oplus q$.
        \item there is a least element; for all $r\in\mathcal{R}$, $\bot \sqsubseteq r$.
    \end{itemize}
\item (Motonicity of $\oplus$) $p\sqsubseteq q$ implies $p\oplus r\sqsubseteq q\oplus r$ for all $p, q, r \in \mathcal{R}$
   \begin{itemize}
        \item if $r = \bot$ then $p\oplus r \sqsubseteq q\oplus r \Leftrightarrow p\subseteq q$, so this case is trivial.
        \item otherwise if $r \neq \bot, p = q = \bot$ then $p\oplus r \sqsubseteq q \oplus r \Leftrightarrow r\subseteq r$, so this case is trivial.
        \item otherwise if $r \neq \bot, p = \bot, q\neq \bot$ then $p\oplus r \sqsubseteq q \oplus r \Leftrightarrow r\subseteq q\cup r$, and $r\subseteq q\cup r$ holds in standard subset relation.
        \item otherwise if $r \neq \bot, p \neq \bot, q \neq \bot$ then $p\oplus r \sqsubseteq q \oplus r \Leftrightarrow p \cup r \subseteq q\cup r$, and  $p \subseteq q$ implies $p \cup r \subseteq q\cup r$.
    \end{itemize}
\item (Motonicity of $\otimes$) $p\sqsubseteq q$ implies $p\otimes r\sqsubseteq q\otimes r$ for all $p, q, r \in \mathcal{R}$
    \begin{itemize}
        \item if $r = \bot$ then $p\otimes r \sqsubseteq q\otimes r \Leftrightarrow \bot\subseteq \bot$, so this case is trivial.
        \item otherwise if $r \neq \bot, p = q = \bot$ then $p\otimes r \sqsubseteq q \otimes r \Leftrightarrow \bot\subseteq \bot$, so this case is trivial.
        \item otherwise if $r \neq \bot, p = \bot, q\neq \bot$ then $p\otimes r \sqsubseteq q \otimes r \Leftrightarrow \bot\subseteq q\cup r$, so this case is trivial.
        \item otherwise if $r \neq \bot, p \neq \bot, q \neq \bot$ then $p\otimes r \sqsubseteq q \otimes r \Leftrightarrow p \cup r \subseteq q\cup r$, and  $p \subseteq q$ implies $p \cup r \subseteq q\cup r$.
    \end{itemize}
\end{itemize}
\end{proof}

\begin{definition}[Version resource summation]
Using the addition $+$ of version resource semiring, summation of version resouce is defined as follows:
\begin{align*}
\sum_i r_i = r_1 \oplus \cdots \oplus r_n
\end{align*}
\end{definition}

% \begin{lemma}[Well-typed linear substitution]
% Let $\Delta \vdash t':A$ and $\Gamma,x:A,\Gamma' \vdash t:B$. Then, $\Gamma + \Delta + \Gamma' \vdash [t'/x]t:B$
% \end{lemma}
% \begin{proof}
% By induction on the derivation of $\Gamma, x : A, \Gamma' \vdash t : B$. 
% \end{proof}

% \begin{lemma}[Well-typed versioned substitution]
% Let $[\Delta] \vdash t':A$ and $\Gamma,x:\verctype{A}{r},\Gamma' \vdash t:B$.
% Then, $\Gamma + \bigcup_i(r_i\cdot[\Delta_i]) + \Gamma' \vdash [t'/x]t:B$ where $\Sigma_ir_i = r$ and  $\bigcup_i[\Delta_i] = \Delta$%
% \end{lemma}
% \begin{proof}
% By induction on the derivation of $\Gamma,x:\verctype{A}{r},\Gamma' \vdash t:B$.
% \end{proof}

% \begin{lemma}[Default version overwriting type safety]
% Let $[\Gamma] \vdash t':A$.
% Then, $\exists t'. t@l \equiv t' \land \{l\}\cdot[\Gamma] \vdash t':A$%
% \end{lemma}
% \begin{proof}
% By induction on the derivation of $[\Gamma] \vdash t':A$.
% \end{proof}

% \begin{theorem}[\mylang{} type safety]
% Let $\Gamma \vdash t:A$. Then, (i) $t$ is a value or (ii) $\exists t',\Gamma'.t \leadsto t' \land \Gamma' \vdash t':A' \land \Gamma' \sqsubseteq \Gamma$
% \end{theorem}
% \begin{proof}
% By induction on the derivation of $\Gamma \vdash t:A$.
% \end{proof}

\subsection{Context Properties}
\begin{definition}[Context concatenation]
\label{def:contextconcat}
Two typing contexts can be concatenated by "$,$" if they contain disjoint assumptions. 
Furthermore, the versioned assumptions appearing in both typing contexts can be combined using the context concatenation $+$ defined with the addition $\oplus$ in the version resource semiring as follows.
\begin{align*}
\emptyset + \Gamma &= \Gamma\\
(\Gamma,x:A)+\Gamma' &= (\Gamma + \Gamma'),x:A \hspace{1em} \text{iff}\hspace{0.5em} x \notin \mathrm{dom}(\Gamma')\\
\Gamma + \emptyset &= \Gamma\\
\Gamma+(\Gamma',x:A) &= (\Gamma + \Gamma'),x:A \hspace{1em} \text{iff}\hspace{0.5em} x \notin \mathrm{dom}(\Gamma)\\
(\Gamma,x:\verctype{A}{r})+(\Gamma',x:\verctype{A}{s}) &= (\Gamma + \Gamma'),x:\verctype{A}{r \,\oplus\, s}
\end{align*}
\end{definition}

\begin{definition}[Context multiplication by version resource]
\label{def:multiply}
Assuming that a context contains only version assumptions, denoted $[\Gamma]$ in typing rules, then $\Gamma$ can be multiplied by a version resource $r \in \mathcal{R}$ by using the product $\otimes$ in the version resource semiring, as follows.
\begin{align*}
r \cdot \emptyset\ =\ \emptyset \hspace{4em}
r \cdot (\Gamma,\, x:\verctype{A}{s})\ =\ (r \cdot \Gamma),\, x:\verctype{A}{r\,\otimes\, s}
\end{align*}    
\end{definition}

\begin{definition}[Context summation]
Using the context concatenation $+$, summation of typing contexts is defined as follows:
\begin{align*}
\displaystyle \bigcup_{i=1}^{n} \ \Gamma_i\ =\ \Gamma_1 + \cdots + \Gamma_n
\end{align*}
\end{definition}

\begin{definition}[Context partition]
\label{def:restriction}
For typing contexts $\Gamma_1$ and $\Gamma_2$, we define $\incl{\Gamma_1}{\Gamma_2}$ and $\excl{\Gamma_1}{\Gamma_2}$ as follows.
\begin{align*}
\incl{\Gamma_1}{\Gamma_2} &\triangleq \{ x:A\ |\ x \in \mathrm{dom}(\Gamma_1) \land x \in \mathrm{dom}(\Gamma_2)\}\\
\excl{\Gamma_1}{\Gamma_2} &\triangleq \{ x:A\ |\ x \in \mathrm{dom}(\Gamma_1) \land x \notin \mathrm{dom}(\Gamma_2)\}
\end{align*}
$\incl{\Gamma_1}{\Gamma_2}$ is a subsequence of $\Gamma_1$ that contains all the term variables that are \emph{included} in $\Gamma_2$, and
$\excl{\Gamma_1}{\Gamma_2}$ is a subsequence of $\Gamma_1$ that contains all the term variables that are \emph{not included} in $\Gamma_2$.
\end{definition}

Using $\incl{\Gamma_1}{\Gamma_2}$ and $\excl{\Gamma_1}{\Gamma_2}$, we state some corollaries about typing contexts.
These theorems follow straightforwardly from the definitions of \ref{def:restriction}.

\begin{lemma}[Context collapse]
\label{lemma:restriction}
For typing contexts $\Gamma_1$ and $\Gamma_2$,
\begin{align*}
    \incl{\Gamma_1}{\Gamma_2} + \excl{\Gamma_1}{\Gamma_2} = \Gamma_1
\end{align*}
%これは$\incl{\Gamma_1}{\Gamma_2}$と$\excl{\Gamma_1}{\Gamma_2}$ definitionより明らかである。
\end{lemma}

\begin{lemma}[Context shuffle]
\label{lemma:shuffle}
For typing contexts $\Gamma_1$, $\Gamma_2$, $\Gamma_3$ and $\Gamma_4$, and variable $x$ and type $A$:
\begin{align*}
(\Gamma_1,x:A,\Gamma'_1)+\Gamma_2 &= (\Gamma_1 + \incl{\Gamma_2}{\Gamma_1}),x:A,(\Gamma_1' + \excl{\Gamma_2}{\Gamma_1}) \tag{1}\\
\Gamma_1 + (\Gamma_2,x:A,\Gamma'_2) &= (\excl{\Gamma_1}{\Gamma'_2} + \Gamma_2),x:A,(\incl{\Gamma_1}{\Gamma'_2} + \Gamma'_2) \tag{2}\\
(\Gamma_1,\Gamma_2)+(\Gamma_3,\Gamma_4) &= \left((\Gamma_1+\incl{\Gamma_3}{\Gamma_1}+\incl{\Gamma_4}{\Gamma_1}), (\Gamma_2+\excl{\Gamma_3}{\Gamma_1}+\excl{\Gamma_4}{\Gamma_1})\right) \tag{3}
\end{align*}
\end{lemma}

\begin{lemma}[Composition of context shuffle]
\label{lemma:shufflecomposition}
For typing contexts $\Gamma_i$ and $\Gamma_i'$ for $i\in \mathbb{N}$, there exixts typing contexts $\Gamma$ and $\Gamma'$ such that:
\begin{align*}
\bigcup_{i}(\Gamma_i,\Gamma_i') = \Gamma, \Gamma'\ \land\ \bigcup_i(\Gamma_i+\Gamma_i') = \Gamma + \Gamma'
\end{align*}
\end{lemma}

\begin{lemma}[Distribution of version resouce over context addition]
\label{lemma:distributivelaw}
For a typing context $\Gamma$ and resources $r_i \in R$:
\begin{align*}
(r_1 \cdot \Gamma) + (r_2\cdot \Gamma) &= (r_1 \oplus r_2)\cdot\Gamma\\
\bigcup_i(r_i \cdot \Gamma) &= (\sum_i r_i)\cdot\Gamma
\end{align*}
\end{lemma}

\begin{lemma}[Disjoint context collapse]
\label{lemma:collapse}
Given typing contexts $\Gamma_1$, $\Delta$, and $\Gamma_2$ such that $\Gamma_1$ and $\Gamma_2$ are disjoint, then we can conclude the following.
\begin{align*}
(\Gamma_1+\Delta+\Gamma_2) = (\Gamma_1+\incl{\Delta}{\Gamma_1}),\excl{\Delta}{(\Gamma_1,\Gamma_2)},(\Gamma_2+\incl{\Delta}{\Gamma_2})
\end{align*}
\end{lemma}

%%%%%%%%%%%%%%%%%%%%%%%%%%%%%%%%%%%%%%%%%%%%%%%%%%%%%%%%%%%%%%%%%%%%%%%%%%%%%
%%%%%%%%%%%%%%%%%%%%%%%%%%%%%%%%%%%%%%%%%%%%%%%%%%%%%%%%%%%%%%%%%%%%%%%%%%%%%
%%%%%%%%%%%%%%%%%%%%%%%%%%%%%%%%%%%%%%%%%%%%%%%%%%%%%%%%%%%%%%%%%%%%%%%%%%%%%
%%%%%%%%%%%%%%%%%%%%%%%%%%%%%%%%%%%%%%%%%%%%%%%%%%%%%%%%%%%%%%%%%%%%%%%%%%%%%
%%%%%%%%%%%%%%%%%%%%%%%%%%%%%%%%%%%%%%%%%%%%%%%%%%%%%%%%%%%%%%%%%%%%%%%%%%%%%
%%%%%%%%%%%%%%%%%%%%%%%%%%%%%%%%%%%%%%%%%%%%%%%%%%%%%%%%%%%%%%%%%%%%%%%%%%%%%

\subsection{Substituions Properties}
\label{appendix:lemsubstitution}

\begin{lemma}[Well-typed linear substitution]
\label{lemma:substitution1}
\begin{align*}
    \left.
    \begin{aligned}
          \Delta \vdash t': A\\
          \Gamma,x:A,\Gamma' \vdash t:B
    \end{aligned}
    \right\}
    \hspace{1em}\Longrightarrow\hspace{1em}
    \Gamma + \Delta + \Gamma' \vdash [t'/x]t:B
\end{align*}
\end{lemma}

\begin{proof}
This proof is given by induction on the structure of $\Gamma,x:A,\Gamma' \vdash t:B$ (assumption 2).
Consider the cases for the last rule used in the typing derivation of assumption 2.

\begin{itemize}
\item Case (\textsc{int})
\begin{center}
    \begin{minipage}{.27\linewidth}
        \infrule[int]{
             \\% \\
        }{
            \emptyset \vdash n : \textsf{Int}
        }
    \end{minipage}
    % \hspace{1em}\& \hspace{1em}Case 
    % \begin{minipage}{.50\linewidth} % 0.32だった
    %     \infrule[C]{
    %         (x:\forall\{\overrightarrow{\alpha:\kappa}\}.A) \in D
    %         \andalso
    %         %\theta,\Sigma',\theta_{\kappa'} = \textsf{instantiate}(\overrightarrow{\alpha:\kappa},\theta_{\kappa})
    %         \theta,\Sigma' = \textsf{inst}(\overrightarrow{\alpha:\kappa})
    %     }{
    %         D;\Sigma,\Sigma';\emptyset \vdash x : \theta A
    %         %D;\Sigma,\Sigma';\emptyset \vdash C : (\theta_{\kappa}' \uplus \theta)A
    %     }
    % \end{minipage}
\end{center}
In this case, the above typing context is empty ($= \emptyset$), so this case holds trivially.\\

\item Case (\textsc{var})
\begin{center}
    \begin{minipage}{.50\linewidth}
        \infrule[var]{
            \vdash B
        }{
            y:B \vdash y:B
        }
    \end{minipage}
\end{center}
We are given
\begin{gather*}
\Gamma=\Gamma'=\emptyset,\quad
x=t=y,\quad
A = B
~.
\end{gather*}
Now the conclusion of the lemma is
\begin{align*}
\Delta \vdash [t'/y]y:B
~.
\end{align*}
Since $[t'/y]y=t'$ from the definition of substitution, the conclusion of the lemma is assumption 1 itself.\\

\item Case (\textsc{abs})
\begin{center}
    \begin{minipage}{.75\linewidth}
        \infrule[abs]{
            - \vdash p : B_1 \rhd \Delta'
            \andalso
            \Gamma, x:A, \Gamma', \Delta' \vdash t : B_2%\theta B
        }{
            \Gamma, x:A, \Gamma' \vdash \lam{p}{t} : \ftype{B_1}{B_2}
        }
    \end{minipage}
\end{center}
In this case, by applying the induction hypothesis to the second premise, we know the following:
\begin{align*}
    \Gamma + \Delta + (\Gamma', \Delta') \vdash [t'/x]t : B_2
\end{align*}
where $y:B$ is disjoint with $\Gamma$, $\Delta$, and $\Gamma'$.
Thus, $\Gamma + \Delta + (\Gamma', \Delta') = (\Gamma + \Delta + \Gamma'), \Delta'$ from Lemma \ref{lemma:shuffle} (2), the typing derivation above is equal to the following:
\begin{align*}
    (\Gamma + \Delta + \Gamma'), \Delta' \vdash [t'/x]t : B_2
\end{align*}
We then reapply (\textsc{abs}) to obtain the following:
\begin{center}
    \begin{minipage}{.75\linewidth}
        \infrule[abs]{
             - \vdash p : B_1 \rhd \Delta'
            \andalso
            (\Gamma + \Delta + \Gamma'), \Delta' \vdash [t'/x]t : B_2
        }{
            \Gamma + \Delta + \Gamma' \vdash \lam{p}{[t'/x]t} : \ftype{B_1}{B_2}
        }
    \end{minipage}
\end{center}
By the definition of substitution $\lam{p}{[t'/x]t}=[t'/x](\lam{p}{t})$, and we obtain the conclusion of the lemma.
\\

\item Case (\textsc{let})
\begin{center}
    \begin{minipage}{.75\linewidth}
        \infrule[let]{
             \Gamma_1 \vdash t_1 : \vertype{r}{A}
             \andalso
             \Gamma_2, x:\verctype{A}{r} \vdash t_2 : B
        }{
             \Gamma_1 + \Gamma_2 \vdash \clet{x}{t_1}{t_2} : B
        }
    \end{minipage}
\end{center}
This case is similar to the case (\textsc{app}).
\\

\item Case (\textsc{app})
\begin{center}
    \begin{minipage}{.75\linewidth}
        \infrule[app]{
             \Gamma_1 \vdash t_1 : \ftype{B_1}{B_2}
             \andalso
             \Gamma_2 \vdash t_2 : B_1
        }{
             \Gamma_1 + \Gamma_2 \vdash \app{t_1}{t_2} : B_2
        }
    \end{minipage}
\end{center}
We are given
\begin{gather*}
\Gamma, x:A, \Gamma' = \Gamma_1 + \Gamma_2,\quad
t = \app{t_1}{t_2},\quad
B = B_2
~.
\end{gather*}
By the definition of the context addition $+$, the linear assumption $x:A$ is contained in only one of $\Gamma_1$ or $\Gamma_2$.
\begin{itemize}
\item Suppose $(x:A) \in \Gamma_1$ and $(x:A) \notin \Gamma_2$.\\
Let $\Gamma_1'$ and $\Gamma_1''$ be typing contexts such that they satisfy $\Gamma_1 = (\Gamma_1', x:A, \Gamma_1'')$.
The last typing derivation of ($\textsc{app}$) is rewritten as follows.
\begin{center}
    \begin{minipage}{.60\linewidth}
        \infrule[app]{
             \Gamma_1', x:A, \Gamma_1'' \vdash t_1 : \ftype{B_1}{B_2}
            \\
             \Gamma_2 \vdash t_2 : B_1
        }{
             (\Gamma_1', x:A, \Gamma_1'' ) + \Gamma_2 \vdash \app{t_1}{t_2} : B_2
        }
    \end{minipage}
\end{center}
Now, we compare the typing contexts between the lemma and the above conclusion as follows:
\begin{align*}
(\Gamma, x:A, \Gamma')
    &= (\Gamma_1 + \Gamma_2)\\
    &= (\Gamma_1', x:A, \Gamma_1'' ) + \Gamma_2\tag{$\because$ $\Gamma_1 = (\Gamma_1', x:A, \Gamma_1'')$}\\
    &= (\Gamma_1' + \incl{\Gamma_2}{\Gamma_1'}),x:A,(\Gamma_1'' + \excl{\Gamma_2}{\Gamma_1'})\tag{$\because$ Lemma \ref{lemma:shuffle} (1)}
\end{align*}
By the commutativity of "$,$", we can take $\Gamma$ and $\Gamma'$ arbitrarily so that they satisfy the above equation. So here we know $\Gamma = (\Gamma_1' + \incl{\Gamma_2}{\Gamma_1'})$ and $\Gamma' = (\Gamma_1'' + \excl{\Gamma_2}{\Gamma_1'})$.\par
We then apply the induction hypothesis to each of the two premises and reapply (\textsc{app}) as follows:
\begin{center}
    \begin{minipage}{.75\linewidth}
        \infrule[app]{
             \Gamma_1' + \Delta + \Gamma_1'' \vdash [t'/x]t_1 : \ftype{B_1}{B_2}
            \\
             \Gamma_2 \vdash t_2 : B_1
        }{
             (\Gamma_1' + \Delta + \Gamma_1'' ) + \Gamma_2 \vdash \app{([t'/x]t_1)}{t_2} : B_2
        }
    \end{minipage}
\end{center}
Since $\app{([t'/x]t_1)}{t_2}=[t'/x](\app{t_1}{t_2})$ if $x\notin FV(t_2)$, the conclusion of the above derivation is equivalent to the conclusion of the lemma except for the typing contexts.\par
Finally, we must show that $(\Gamma + \Delta + \Gamma') = ((\Gamma_1 + \Delta + \Gamma_1'') + \Gamma_2)$.
This holds from the following reasoning:
\begin{align*}
(\Gamma + \Delta + \Gamma')
    &= (\Gamma_1' + \incl{\Gamma_2}{\Gamma_1'}) + \Delta + (\Gamma_1'' + \excl{\Gamma_2}{\Gamma_1'})\tag{$\because$ $\Gamma = (\Gamma_1' + \incl{\Gamma_2}{\Gamma_1'})$ and $\Gamma' = (\Gamma_1'' + \excl{\Gamma_2}{\Gamma_1'})$}\\
    &= \Gamma_1' + \incl{\Gamma_2}{\Gamma_1'} + \Delta + \Gamma_1'' + \excl{\Gamma_2}{\Gamma_1'}\tag{$\because$ $+$ associativity}\\
    &= \Gamma_1' + \Delta + \Gamma''_1 + \incl{\Gamma_2}{\Gamma_1'} + \excl{\Gamma_2}{\Gamma_1'}\tag{$\because$ $+$ commutativity}\\
    &= (\Gamma_1' + \Delta + \Gamma''_1) + (\incl{\Gamma_2}{\Gamma_1'} + \excl{\Gamma_2}{\Gamma_1'})\tag{$\because$ $+$ associativity}\\
    &= (\Gamma_1' + \Delta + \Gamma''_1) + \Gamma_2\tag{$\because$ Lemma \ref{lemma:restriction}}
\end{align*}
Thus, we obtain the conclusion of the lemma.

\item Suppose $(x:A) \notin \Gamma_1$ and $(x:A) \in \Gamma_2$\\
Let $\Gamma_2'$ and $\Gamma_2''$ be typing contexts such that they satisfy  $\Gamma_2 = (\Gamma_2', x:A, \Gamma_2'')$.
The last typing derivation of ($\textsc{app}$) is rewritten as follows.
\begin{center}
    \begin{minipage}{.75\linewidth}
        \infrule[app]{
             \Gamma_1 \vdash t_1 : \ftype{B_1}{B_2}
             \andalso
             \Gamma_2', x:A, \Gamma_2'' \vdash t_2 : B_1
        }{
             \Gamma_1 + (\Gamma_2', x:A, \Gamma_2'') \vdash \app{t_1}{t_2} : B_2
        }
    \end{minipage}
\end{center}
This case is similar to the case $(x:A)\in \Gamma_1$, but using \ref{lemma:shuffle} (2) instead of \ref{lemma:shuffle} (1).\\
\end{itemize}

\item Case (\textsc{weak})
\begin{center}
    \begin{minipage}{.55\linewidth}
        \infrule[weak]{
            \Gamma_1, x:A, \Gamma_2 \vdash t : B
            \andalso
            \vdash \Delta'
        }{
            (\Gamma_1, x:A, \Gamma_2) + \verctype{\Delta'}{0} \vdash t : B
        }
    \end{minipage}
\end{center}
In this case, the linear assumption $x:A$ is not contained in versioned context $\verctype{\Delta'}{0}$.
We then compare the typing contexts between the conclusion of the lemma and that of (\textsc{weak}) as follows:
\begin{align*}
    (\Gamma, x:A, \Gamma')
    &= (\Gamma_1, x:A, \Gamma_2) + \verctype{\Delta'}{0}\\
    &= (\Gamma_1 + \incl{(\verctype{\Delta'}{0})}{\Gamma_1}),x:A,(\Gamma_2 + \excl{(\verctype{\Delta'}{0})}{\Gamma_1})\tag{$\because$ Lemma \ref{lemma:shuffle} (1)}
\end{align*}
By the commutativity of "$,$", we can take $\Gamma$ and $\Gamma'$ arbitrarily so that they satisfy the above equation. So here we obtain $\Gamma=\Gamma_1 + \incl{(\verctype{\Delta'}{0})}{\Gamma_1}$ and $\Gamma'=\Gamma_2 + \excl{(\verctype{\Delta'}{0})}{\Gamma_1}$.
We then apply the induction hypothesis to each of the premise and reapply (\textsc{weak}) as follows:
\begin{center}
    \begin{minipage}{.65\linewidth}
        \infrule[weak]{
            \Gamma_1 + \Delta + \Gamma_2 \vdash [t'/x]t : B
            \andalso
            \vdash \Delta'
        }{
            (\Gamma_1 + \Delta + \Gamma_2) + \verctype{\Delta'}{0} \vdash [t'/x]t : B
        }
    \end{minipage}
\end{center}
Since $\app{([t'/x]t_1)}{([t'/x]t_2)}=[t'/x](\app{t_1}{t_2})$, the conclusion of the above derivation is equivalent to the conclusion of the lemma except for typing contexts.\par
Finally, we must show that $(\Gamma_1 + \Delta + \Gamma_2) + \verctype{\Delta'}{0} = \Gamma + \Delta + \Gamma'$.
This holds from the following reasoning:
\begin{align*}
    (\Gamma + \Delta + \Gamma')
    &= (\Gamma_1 + \incl{(\verctype{\Delta'}{0})}{\Gamma_1})+\Delta+(\Gamma_2 + \excl{(\verctype{\Delta'}{0})}{\Gamma_1})\tag{$\because$ $\Gamma=\Gamma_1 + \incl{(\verctype{\Delta'}{0})}{\Gamma_1}$ and $\Gamma'=\Gamma_2 + \excl{(\verctype{\Delta'}{0})}{\Gamma_1}$}\\
    &= \Gamma_1 + \incl{(\verctype{\Delta'}{0})}{\Gamma_1} + \Delta + \Gamma_2 + \excl{(\verctype{\Delta'}{0})}{\Gamma_1} \tag{$\because$ $+$ associativity}\\
    &= \Gamma_1 + \Delta + \Gamma_2 + \incl{(\verctype{\Delta'}{0})}{\Gamma_1} + \excl{(\verctype{\Delta'}{0})}{\Gamma_1} \tag{$\because$ $+$ commutativity}\\
    &= (\Gamma_1 + \Delta + \Gamma_2) + (\incl{(\verctype{\Delta'}{0})}{\Gamma_1} + \excl{(\verctype{\Delta'}{0})}{\Gamma_1}) \tag{$\because$ $+$ associativity}\\
    &= (\Gamma_1 + \Delta + \Gamma_2) + \verctype{\Delta'}{0}\tag{$\because$ Lemma \ref{lemma:restriction}}
\end{align*}
Thus, we obtain the conclusion of the lemma.
\\

\item Case (\textsc{der})
\begin{center}
    \begin{minipage}{.55\linewidth}
        \infrule[der]{
            \Gamma, x:A ,\Gamma'', y:B_1 \vdash t : B_2
        }{
            \Gamma, x:A, \Gamma'', y:\verctype{B_1}{1} \vdash t : B_2
        }
    \end{minipage}
\end{center}
In this case, a linear assumption $x:A$ cannot be a versioned assumption $y:\verctype{B_1}{1}$.
Applying the induction hypothesis to the premise, we obtain the following:
\begin{align*}
    \Gamma + \Delta + (\Gamma'', y:B_1) \vdash [t'/x]t : B_2
\end{align*}
Note that $\Gamma + \Delta + (\Gamma'', y:B_1)= (\Gamma + \Delta + \Gamma''), y:B_1$ holds because $y:B_1$ is a linear assumption and is disjoint with $\Gamma$, $\Delta$, and $\Gamma''$.
Thus, the above judgement is equivalent to the following:
\begin{align*}
    (\Gamma + \Delta + \Gamma''), y:B_1 \vdash [t'/x]t : B_2
\end{align*}
We then reapply (\textsc{der}) to obtain the following:
\begin{center}
    \begin{minipage}{.65\linewidth}
        \infrule[der]{
            (\Gamma + \Delta + \Gamma''), y:B_1 \vdash [t'/x]t : B_2
        }{
            (\Gamma + \Delta + \Gamma''), y:\verctype{B_1}{1} \vdash [t'/x]t : B_2
        }
    \end{minipage}
\end{center}
Finally, since $y:\verctype{B_1}{1}$ is disjoint with $\Gamma+\Delta+\Gamma''$, $((\Gamma + \Delta + \Gamma''), y:\verctype{B_1}{1}) = \Gamma + \Delta + (\Gamma'', y:\verctype{B_1}{1})$ holds. Thus, the conclusion of the above derivation is equivalent to the following:
\begin{align*}
    \Gamma + \Delta + (\Gamma'', y:\verctype{B_1}{1}) \vdash [t'/x]t : B_2
\end{align*}
Thus, we obtain the conclusion of the lemma.
\\

\item Case (\textsc{pr})
\begin{center}
    \begin{minipage}{.4\linewidth}
        \infrule[pr]{
            \verctype{\Gamma}{} \vdash t : B
            \andalso
            \vdash r
        }{
            r\cdot\verctype{\Gamma}{} \vdash [t] : \vertype{r}{B} 
        }
    \end{minipage}
\end{center}
This case holds trivially, because the typing context $[\Gamma]$ contains only versioned assumptions and does not contain any linear assumptions.
\\

\item Case (\textsc{ver})
\begin{center}
    \begin{minipage}{.65\linewidth}
        \infrule[ver]{
            \verctype{\Gamma_i}{} \vdash t_i : A
            \andalso
            \vdash \{\overline{l_i}\}
        }{
            \bigcup_i(\{l_i\}\cdot [\Gamma_i]) \vdash \nvval{\overline{l_i=t_i}} : \vertype{\{\overline{l_i}\}}{A}
        }
    \end{minipage}
\end{center}
This case holds trivially, because the typing context of the conclusion contains only versioned assumptions (by $[\Gamma_i]$ in the premise) and does not contain any linear assumptions.
\\

\item Case (\textsc{veri})
\begin{center}
    \begin{minipage}{.65\linewidth}
        \infrule[veri]{
            \verctype{\Gamma_i}{} \vdash t_i : A
            \andalso
            \vdash \{\overline{l_i}\}
            \andalso
            l_k \in \{\overline{l_i}\}
        }{
            \bigcup_i(\{l_i\}\cdot [\Gamma_i]) \vdash \ivval{\overline{l_i=t_i}}{l_k} : A
        }
    \end{minipage}
\end{center}
This case holds trivially, because the typing context of the conclusion contains only versioned assumptions (by $[\Gamma_i]$ in the premise) and does not contain any linear assumptions.
\\

\item Case (\textsc{extr})
\begin{center}
    \begin{minipage}{.45\linewidth}
        \infrule[extr]{
            \Gamma \vdash t : \vertype{r}{A}
            \andalso
            l \in r
        }{
            \Gamma \vdash t.l : A
        }
    \end{minipage}
\end{center}
In this case, we apply the induction hypothesis to the premise and then reapply (\textsc{extr}), we obtain the conclusion of the lemma.\\

\item Case (\textsc{sub})
\begin{center}
    \begin{minipage}{.50\linewidth}
        \infrule[\textsc{sub}]{
            \Gamma,y:\verctype{B'}{r}, \Gamma' \vdash t : B
            \andalso
            r \sqsubseteq s
            \andalso
            \vdash s
        }{
            \Gamma,y:\verctype{B'}{s}, \Gamma' \vdash t : B
        }
    \end{minipage}
\end{center}
In this case, a linear assumption $x:A$ cannot be a versioned assumption $y:\verctype{B_1}{s}$, and only one of $(x:A)\in\Gamma$ or $(x:A)\in\Gamma'$ holds.
In either case, applying the induction hypothesis to the premise and reappling (\textsc{sub}), we obtain the conclusion of the lemma.

\end{itemize}
\end{proof}

%%%%%%%%%%%%%%%%%%%%%%%%%%%%%%%%%%%%%%%%%%%%%%%%%%%%%%%%%%%%%%%%%%%%%%%%%%%%%
%%%%%%%%%%%%%%%%%%%%%%%%%%%%%%%%%%%%%%%%%%%%%%%%%%%%%%%%%%%%%%%%%%%%%%%%%%%%%
%%%%%%%%%%%%%%%%%%%%%%%%%%%%%%%%%%%%%%%%%%%%%%%%%%%%%%%%%%%%%%%%%%%%%%%%%%%%%
%%%%%%%%%%%%%%%%%%%%%%%%%%%%%%%%%%%%%%%%%%%%%%%%%%%%%%%%%%%%%%%%%%%%%%%%%%%%%
%%%%%%%%%%%%%%%%%%%%%%%%%%%%%%%%%%%%%%%%%%%%%%%%%%%%%%%%%%%%%%%%%%%%%%%%%%%%%
%%%%%%%%%%%%%%%%%%%%%%%%%%%%%%%%%%%%%%%%%%%%%%%%%%%%%%%%%%%%%%%%%%%%%%%%%%%%%

\begin{lemma}[Well-typed versioned substitution]
\label{lemma:substitution2}
\begin{align*}
    \left.
    \begin{aligned}\relax
          [\Delta] \vdash t':A\\
          \Gamma,x:\verctype{A}{r},\Gamma' \vdash t:B
    \end{aligned}
    \right\}
    \hspace{1em}\Longrightarrow\hspace{1em}
    \Gamma + r\cdot\Delta + \Gamma' \vdash [t'/x]t:B
\end{align*}
\end{lemma}
% \begin{lemma}[Well-typed versioned substitution]
% \label{lemma:substitution2}
% % Given $[\Delta] \vdash t':A$ (assumption 1) and $\Gamma,x:\verctype{A}{r},\Gamma' \vdash t:B$ (assumption 2), then $\Gamma + r\cdot\Delta + \Gamma' \vdash [t'/x]t:B$ holds.
% \begin{align*}
%     \left.
%     \begin{aligned}
%           [\Delta] \vdash t':A\\
%           \Gamma,x:\verctype{A}{r},\Gamma' \vdash t:B
%     \end{aligned}
%     \right\}
%     \hspace{1em}\Longrightarrow\hspace{1em}
%     \Gamma + r\cdot\Delta + \Gamma' \vdash [t'/x]t:B
% \end{align*}
% \end{lemma}

\begin{proof}
This proof is given by induction on structure of $\Gamma,x:\verctype{A}{r},\Gamma' \vdash t:B$ (assumption 2).
Consider the cases for the last rule used in the typing derivation of assumption 2.

\begin{itemize}
\item Case (\textsc{int})
\begin{center}
    \begin{minipage}{.22\linewidth}
        \infrule[int]{
             \\% \\
        }{
            \emptyset \vdash n : \textsf{Int}
        }
    \end{minipage}
    % \hspace{1em}\& \hspace{1em}Case 
    % \begin{minipage}{.39\linewidth} % 0.32だった
    %     \infrule[C]{
    %         (x:\forall\{\overrightarrow{\alpha:\kappa}\}.A) \in D
    %         \andalso
    %         %\theta,\Sigma',\theta_{\kappa'} = \textsf{instantiate}(\overrightarrow{\alpha:\kappa},\theta_{\kappa})
    %         \theta,\Sigma' = \textsf{inst}(\overrightarrow{\alpha:\kappa})
    %     }{
    %         D;\Sigma,\Sigma';\emptyset \vdash x : \theta A
    %         %D;\Sigma,\Sigma';\emptyset \vdash C : (\theta_{\kappa}' \uplus \theta)A
    %     }
    % \end{minipage}
\end{center}
This case holds trivially because the typing context of (\textsc{int}) is empty ($=\emptyset$).\\

\item Case (\textsc{var})
\begin{center}
    \begin{minipage}{.35\linewidth}
        \infrule[var]{
            \vdash B
        }{
            y:B \vdash y:B
        }
    \end{minipage}
\end{center}
In this case, $x:\verctype{A}{r}$ is a versioned assumption and $y:B$ is a linear assumption, so $x\neq y$ holds, and yet the typing context besides $y:B$ is empty.
Thus, there are no versioned variables to be substituted, so this case holds trivially.\\

\item Case (\textsc{abs})
\begin{center}
    \begin{minipage}{.8\linewidth}
        \infrule[abs]{
            - \vdash p : B_1 \rhd \Delta'
            \andalso
            \Gamma, x:\verctype{A}{r}, \Gamma', \Delta' \vdash t : B_2%\theta B
        }{
            \Gamma, x:\verctype{A}{r}, \Gamma' \vdash \lam{p}{t} : \ftype{B_1}{B_2}
        }
    \end{minipage}
\end{center}
In this case, we know the following by applying induction hypothesis to the partial derivation of (\textsc{abs}):
\begin{align*}
    \Gamma + r\cdot\Delta + (\Gamma', \Delta') \vdash [t'/x]t : B_2
\end{align*}
where $\Delta'$ ($\mathrm{dom}(\Delta') = \{y\}$) is disjoint with $\Gamma$, $\Delta$, and $\Gamma'$.
Thus, $\Gamma + r\cdot\Delta + (\Gamma', \Delta') = (\Gamma + r\cdot\Delta + \Gamma'), \Delta'$ from Lemma \ref{lemma:shuffle} (2), the typing derivation above is equal to the following:
\begin{align*}
    (\Gamma + r\cdot\Delta + \Gamma'), \Delta' \vdash [t'/x]t : B_2
\end{align*}
We then reapply (\textsc{abs}) to obtain the following:
\begin{center}
    \begin{minipage}{.80\linewidth}
        \infrule[abs]{
             - \vdash p : B_1 \rhd \Delta'
            \andalso
            (\Gamma + r\cdot\Delta + \Gamma'), \Delta' \vdash [t'/x]t : B_2
        }{
            \Gamma + r\cdot\Delta + \Gamma' \vdash \lam{p}{[t'/x]t} : \ftype{B_1}{B_2}
        }
    \end{minipage}
\end{center}
Since $\lam{p}{[t'/x]t}=[t'/x](\lam{p}{t})$ from the definition of substitution, we obtain the conclusion of the lemma.\\

\item Case (\textsc{app})
\begin{center}
    \begin{minipage}{.75\linewidth}
        \infrule[app]{
            \Gamma_1 \vdash t_1 : \ftype{B_1}{B_2}
            \andalso
            \Gamma_2 \vdash t_2 : B_1
        }{
             \Gamma_1 + \Gamma_2 \vdash \app{t_1}{t_2} : B_2
        }
    \end{minipage}
\end{center}

We are given
\begin{gather*}
\Gamma, x:\verctype{A}{r}, \Gamma' = \Gamma_1 + \Gamma_2,\quad
t = \app{t_1}{t_2},\quad
B = B_2
~.
\end{gather*}
By the definition of the context addition $+$, the linear assumption $x:A$ is contained in either or both of the typing context $\Gamma_1$ or $\Gamma_2$.
\begin{itemize}
\item Suppose $(x:\verctype{A}{r}) \in \Gamma_1$ and $x \notin \mathrm{dom}(\Gamma_2)$\\
Let $\Gamma_1'$ and $\Gamma_1''$ be typing contexts such that they satisfy $\Gamma_1 = (\Gamma_1', x:\verctype{A}{r}, \Gamma_1'')$.
The last derivation of (\textsc{app}) is rewritten as follows:
\begin{center}
    \begin{minipage}{.9\linewidth}
        \infrule[app]{
            \Gamma_1', x:\verctype{A}{r}, \Gamma_1'' \vdash t_1 \ftype{B_1}{B_2}
            \andalso
            \Gamma_2 \vdash t_2 : B_1
        }{
             (\Gamma_1', x:\verctype{A}{r}, \Gamma_1'' ) + \Gamma_2 \vdash \app{t_1}{t_2} : B_2
        }
    \end{minipage}
\end{center}
We compare the typing contexts between the conclusion of the lemma and that of the above derivation to obtain the following:
\begin{align*}
(\Gamma, x:\verctype{A}{r}, \Gamma')
    &= (\Gamma_1', x:\verctype{A}{r}, \Gamma_1'' ) + \Gamma_2\\
    &= (\Gamma_1' + \incl{\Gamma_2}{\Gamma_1'}),x:\verctype{A}{r},(\Gamma_1'' + \excl{\Gamma_2}{\Gamma_1'})\tag{$\because$ Lemma \ref{lemma:shuffle} (1)}
\end{align*}
By the commutativity of "$,$", we can take $\Gamma$ and $\Gamma'$ arbitrarily so that they satisfy the above equation. So here we know $\Gamma = (\Gamma_1' + \incl{\Gamma_2}{\Gamma_1'})$ and $\Gamma' = (\Gamma_1'' + \excl{\Gamma_2}{\Gamma_1'})$.\par
We then apply the induction hypothesis to each of the two premises of the last derivation and reapply (\textsc{app}) as follows:
\begin{center}
    \begin{minipage}{.85\linewidth}
        \infrule[app]{
             \Gamma_1' + r\cdot\Delta + \Gamma_1'' \vdash [t'/x]t_1 : \ftype{B_1}{B_2}
            \\
             \Gamma_2 \vdash [t'/x]t_2 : B_1
        }{
             (\Gamma_1' + r\cdot\Delta + \Gamma_1'' ) + \Gamma_2 \vdash \app{([t'/x]t_1)}{([t'/x]t_2)} : B_2
        }
    \end{minipage}
\end{center}
Since $\app{([t'/x]t_1)}{([t'/x]t_2)}=[t'/x](\app{t_1}{t_2})$, the conclusion of the above derivation
is equivalent to the conclusion of the lemma except for the typing contexts.
Finally, we must show that $(\Gamma + r\cdot\Delta + \Gamma') = ((\Gamma_1 + r\cdot\Delta + \Gamma_1'') + \Gamma_2)$.
This holds from the following reasoning:
\begin{align*}
(\Gamma + r\cdot\Delta + \Gamma')
    &= (\Gamma_1' + \incl{\Gamma_2}{\Gamma_1'}) + r\cdot\Delta + (\Gamma_1'' + \excl{\Gamma_2}{\Gamma_1'})\tag{$\because$ $\Gamma = (\Gamma_1' + \incl{\Gamma_2}{\Gamma_1'})$ \& $\Gamma' = (\Gamma_1'' + \excl{\Gamma_2}{\Gamma_1'})$}\\
    &= \Gamma_1' + \incl{\Gamma_2}{\Gamma_1'} + r\cdot\Delta + \Gamma_1'' + \excl{\Gamma_2}{\Gamma_1'}\tag{$\because$ $+$ associativity}\\
    &= \Gamma_1' + r\cdot\Delta + \Gamma''_1 + \incl{\Gamma_2}{\Gamma_1'} + \excl{\Gamma_2}{\Gamma_1'}\tag{$\because$ $+$ commutativity}\\
    &= (\Gamma_1' + r\cdot\Delta + \Gamma''_1) + (\incl{\Gamma_2}{\Gamma_1'} + \excl{\Gamma_2}{\Gamma_1'})\tag{$\because$ $+$ associativity}\\
    &= (\Gamma_1' + r\cdot\Delta + \Gamma''_1) + \Gamma_2\tag{$\because$ Lemma \ref{lemma:restriction}}
\end{align*}
Thus, we obtain the conclusion of the lemma.

\item Suppose $x \notin \mathrm{dom}(\Gamma_1)$ and $(x:\verctype{A}{r}) \in \Gamma_2$\\
Let $\Gamma_2'$ and $\Gamma_2''$ be typing contexts such that they satisfy $\Gamma_2 = (\Gamma_2', x:\verctype{A}{r}, \Gamma_2'')$.
The last typing derivation of ($\textsc{app}$) is rewritten as follows:
\begin{center}
    \begin{minipage}{.8\linewidth}
        \infrule[app]{
            \Gamma_1 \vdash t_1 : \ftype{B_1}{B_2}
            \andalso
            \Gamma_2', x:\verctype{A}{r}, \Gamma_2'' \vdash t_2 : B_1
        }{
             \Gamma_1 + (\Gamma_2', x:\verctype{A}{r}, \Gamma_2'') \vdash \app{t_1}{t_2} : B_2
        }
    \end{minipage}
\end{center}
This case is similar to the case $(x:\verctype{A}{r})\in \Gamma_1$ and $x \notin \mathrm{dom}(\Gamma_2)$, but using \ref{lemma:shuffle} (2) instead of \ref{lemma:shuffle} (1).\\

\item Suppose $(x:\verctype{A}{r_1}) \in \Gamma_1$ and $(x:\verctype{A}{r_2}) \in \Gamma_2$ where $r=r_1\oplus r_2$.\\
Let $\Gamma_1'$, $\Gamma_1''$, $\Gamma_2'$, and $\Gamma_2''$ be typing contexts such that they satisfy $\Gamma_1 = (\Gamma_1', x:\verctype{A}{r_1}, \Gamma_1'')$ and $\Gamma_2 = (\Gamma_2', x:\verctype{A}{r_2}, \Gamma_2'')$.
The last derivation of (\textsc{app}) is rewritten as follows:
\begin{center}
    \begin{minipage}{.80\linewidth}
        \infrule[app]{
             \Gamma_1', x:\verctype{A}{r_1}, \Gamma_1'' \vdash t_1 : \ftype{B_1}{B_2}
            \\
             \Gamma_2', x:\verctype{A}{r_2}, \Gamma_2'' \vdash t_2 : B_1
        }{
             (\Gamma_1', x:\verctype{A}{r_1}, \Gamma_1'') + (\Gamma_2', x:\verctype{A}{r_2}, \Gamma_2'') \vdash \app{t_1}{t_2} : B_2
        }
    \end{minipage}
\end{center}
Now, we compare the typing contexts between the lemma and the above conclusion as follows:
\begin{align*}
(\Gamma, x:\verctype{A}{r}, \Gamma')
    &= (\Gamma_1', x:\verctype{A}{r_1}, \Gamma_1'') + (\Gamma_2', x:\verctype{A}{r_2}, \Gamma_2'')\\
    &= (\Gamma_1', \Gamma_1'', x:\verctype{A}{r_1}) + (\Gamma_2', \Gamma_2'', x:\verctype{A}{r_2})\tag{$\because$ $,$ commutativity}\\
    &= ((\Gamma_1', \Gamma_1'') + (\Gamma_2', \Gamma_2'')), x:\verctype{A}{r_1\oplus r_2}\tag{$\because$ $+$ definiton}\\
    &= ((\Gamma_1'+\incl{\Gamma'_2}{\Gamma'_1}+\incl{\Gamma''_2}{\Gamma'_1}), (\Gamma''_1+\excl{\Gamma'_2}{\Gamma'_1}+\excl{\Gamma''_2}{\Gamma'_1})),x:\verctype{A}{r_1\oplus r_2}\tag{$\because$ Lemma \ref{lemma:shuffle} (3)}\\
    &= (\Gamma_1'+\incl{\Gamma'_2}{\Gamma'_1}+\incl{\Gamma''_2}{\Gamma'_1}), (\Gamma''_1+\excl{\Gamma'_2}{\Gamma'_1}+\excl{\Gamma''_2}{\Gamma'_1}), x:\verctype{A}{r_1\oplus r_2}\tag{$\because$ $,$ associativity}\\
    &= (\Gamma_1'+\incl{\Gamma'_2}{\Gamma'_1}+\incl{\Gamma''_2}{\Gamma'_1}),x:\verctype{A}{r_1\oplus r_2}, (\Gamma''_1+\excl{\Gamma'_2}{\Gamma'_1}+\excl{\Gamma''_2}{\Gamma'_1})\tag{$\because$ $,$ commutativity}
\end{align*}
By the commutativity of "$,$", we can take $\Gamma$ and $\Gamma'$ arbitrarily so that they satisfy the above equation. So here we know $\Gamma = (\Gamma_1'+\incl{\Gamma'_2}{\Gamma'_1}+\incl{\Gamma''_2}{\Gamma'_1})$ and $\Gamma' = (\Gamma''_1+\excl{\Gamma'_2}{\Gamma'_1}+\excl{\Gamma''_2}{\Gamma'_1})$.\par
We then apply the induction hypothesis to each of the two premises of the last derivation and reapply (\textsc{app}) as follows:
\begin{center}
    \begin{minipage}{\linewidth}
        \infrule[app]{
             \Gamma_1' + r_1\cdot\Delta + \Gamma_1'' \vdash [t'/x]t_1 : \ftype{B_1}{B_2}
            \\
             \Gamma_2' + r_2\cdot\Delta + \Gamma_2'' \vdash [t'/x]t_2 : B_1
        }{
             (\Gamma_1' + r_1\cdot\Delta + \Gamma_1'') + (\Gamma_2' + r_2\cdot\Delta + \Gamma_2'') \vdash \app{([t'/x]t_1)}{([t'/x]t_2)} : B_2
        }
    \end{minipage}
\end{center}
Since $\app{([t'/x]t_1)}{([t'/x]t_2)}=[t'/x](\app{t_1}{t_2})$, the conclusion of the above derivation is equivalent to the conclusion of the lemma except for the typing contexts.
Finally, we must show that $\Gamma + r\cdot\Delta + \Gamma' = (\Gamma_1' + r_1\cdot\Delta + \Gamma_1'') + (\Gamma_2' + r_2\cdot\Delta + \Gamma_2'')$.
\begin{align*}
(\Gamma + r\cdot\Delta + \Gamma')
    &= (\Gamma_1'+\incl{\Gamma'_2}{\Gamma'_1}+\incl{\Gamma''_2}{\Gamma'_1}) + (r_1\oplus r_2)\cdot\Delta +  (\Gamma''_1+\excl{\Gamma'_2}{\Gamma'_1}+\excl{\Gamma''_2}{\Gamma'_1})\tag{$\because$ $r=r_1\oplus r_2$ \& $\Gamma = (\Gamma_1'+\incl{\Gamma'_2}{\Gamma'_1}+\incl{\Gamma''_2}{\Gamma'_1})$ \& $\Gamma' = (\Gamma''_1+\excl{\Gamma'_2}{\Gamma'_1}+\excl{\Gamma''_2}{\Gamma'_1})$}\\
    &= \Gamma_1' + (r_1\oplus r_2)\cdot\Delta + \Gamma''_1 + (\incl{\Gamma'_2}{\Gamma'_1}+\excl{\Gamma'_2}{\Gamma'_1}) + (\incl{\Gamma''_2}{\Gamma'_1}+\excl{\Gamma''_2}{\Gamma'_1}) \tag{$\because$ $+$ associativity \& commutativity}\\
    &= \Gamma_1' + (r_1\oplus r_2)\cdot\Delta + \Gamma''_1 + \Gamma'_2 + \Gamma''_2\tag{$\because$ Lemma \ref{lemma:restriction}}\\
    &= \Gamma_1' + r_1\cdot\Delta + r_2\cdot\Delta + \Gamma''_1 + \Gamma'_2 + \Gamma''_2\tag{$\because$ Lemma \ref{lemma:distributivelaw}}\\
    &= (\Gamma_1' + r_1\cdot\Delta + \Gamma''_1) + (\Gamma'_2 + r_2\cdot\Delta + \Gamma''_2)\tag{$\because$ $+$ associativity and commutativity}
\end{align*}
Thus, we obtain the conclusion of the lemma.
\\
\end{itemize}

\item Case (\textsc{weak})
\begin{center}
    \begin{minipage}{.4\linewidth}
        \infrule[weak]{
            \Gamma'' \vdash t : B
            \andalso
            \vdash \Delta'
        }{
            \Gamma'' + \verctype{\Delta'}{0} \vdash t : B
        }
    \end{minipage}
\end{center}

In this case, we know $(\Gamma, x:\verctype{A}{r}, \Gamma') = \Gamma'' + \verctype{\Delta'}{0}$.
There are two cases where the versioned assumption $x:\verctype{A}{r}$ is contained in $\verctype{\Delta'}{0}$ and not included.

\begin{itemize}
\item Suppose $(x:\verctype{A}{r})\in\verctype{\Delta'}{0}$\\
We know $r=0$.
Let $\Delta_1$ and $\Delta_2$ be typing context such that $\Delta' = (\Delta_1, x:\verctype{A}{0}, \Delta_2)$.
% このとき$x:\verctype{A}{0} \notin \Gamma''$としてよい。
The last derivation is rewritten as follows:
\begin{center}
    \begin{minipage}{.55\linewidth}
        \infrule[weak]{
            \Gamma'' \vdash t : B
            \andalso
            \vdash \Delta_1+\Delta+\Delta_2
        }{
            \Gamma'' + [\Delta_1, x:\verctype{A}{0}, \Delta_2]_0 \vdash t : B
        }
    \end{minipage}
\end{center}

We compare the typing contexts between the conclusion of the lemma and that of the above derivation to obtain the following:
\begin{align*}
(\Gamma, x:\verctype{A}{0}, \Gamma')
    &= \Gamma'' + [\Delta_1, x:\verctype{A}{0}, \Delta_2]_0\tag{$\because$ $\Delta' = (\Delta_1, x:\verctype{A}{0}, \Delta_2)$}\\
    &= \Gamma'' + (\verctype{\Delta_1}{0}, x:\verctype{A}{0}, \verctype{\Delta_2}{0}) \tag{$\because$ definiton of $[\Gamma]_0$}\\
    &= (\incl{\Gamma''}{\verctype{\Delta_2}{0}} + \verctype{\Delta_1}{0}), x:\verctype{A}{0}, (\excl{\Gamma''}{\verctype{\Delta_2}{0}} + \verctype{\Delta_2}{0})\tag{$\because$ Lemma \ref{lemma:shuffle} (2)}
\end{align*}
By the commutativity of "$,$", we can take $\Gamma$ and $\Gamma'$ arbitrarily so that they satisfy the above equation. So here we know $\Gamma = (\incl{\Gamma''}{\verctype{\Delta_2}{0}} + \verctype{\Delta_1}{0})$ and $\Gamma' = (\excl{\Gamma''}{\verctype{\Delta_2}{0}} + \verctype{\Delta_2}{0})$.\par
We then apply the induction hypothesis to the premise of the last derivation and reapply (\textsc{weak}) as follows:
\begin{center}
    \begin{minipage}{.6\linewidth}
        \infrule[weak]{
            \Gamma'' \vdash [t'/x]t : B
            \andalso
            \vdash \Delta_1+\Delta+\Delta_2
        }{
            \Gamma'' + \verctype{\Delta_1+\Delta+\Delta_2}{0} \vdash [t'/x]t : B
        }
    \end{minipage}
\end{center}
where we choose $\Delta_1+\Delta+\Delta_2$ as the newly added typing context.
Since $x$ is unused by $t$, thus note that $[t'/x]t = t$, the conclusion of the above derivation is equivalent to the conclusion of the lemma except for typing contexts.\par
Finally, we must show that $(\Gamma + r\cdot\Delta + \Gamma') = \Gamma'' + \verctype{\Delta_1 + \Delta + \Delta_2}{0}$.
\begin{align*}
(\Gamma + r\cdot\Delta + \Gamma')
    &= (\incl{\Gamma''}{\verctype{\Delta_2}{0}} + \verctype{\Delta_1}{0}) + \verctype{\Delta}{0} + (\excl{\Gamma''}{\verctype{\Delta_2}{0}} + \verctype{\Delta_2}{0}) \tag{$\because$ $r=0$ \& $\Gamma = (\incl{\Gamma''}{\verctype{\Delta_2}{0}} + \verctype{\Delta_1}{0})$ \& $\Gamma' = (\excl{\Gamma''}{\verctype{\Delta_2}{0}} + \verctype{\Delta_2}{0})$}\\
    &= (\incl{\Gamma''}{\verctype{\Delta_2}{0}} + \excl{\Gamma''}{\verctype{\Delta_2}{0}}) + (\verctype{\Delta_1}{0} + \verctype{\Delta}{0} + \verctype{\Delta_2}{0}) \tag{$\because$ $+$ associativity and commutativity}\\
    &= \Gamma'' + (\verctype{\Delta_1}{0} + \verctype{\Delta}{0} + \verctype{\Delta_2}{0}) \tag{$\because$ Lemma \ref{lemma:restriction}}\\
    &= \Gamma'' + \verctype{\Delta_1 + \Delta + \Delta_2}{0} \tag{$\because$ definition of $[\Gamma]_0$}
\end{align*}
Thus, we obtain the conclusion of the lemma.

\item Suppose $(x:\verctype{A}{r})\notin\verctype{\Delta'}{0}$\\
Let $\Gamma_1$ and $\Gamma_2$ be typing context such that $\Gamma'' = (\Gamma_1, x:\verctype{A}{r}, \Gamma_2)$.
The last typing derivation of $(\textsc{weak})$ is rewritten as follows:
\begin{center}
    \begin{minipage}{.65\linewidth}
        \infrule[weak]{
            (\Gamma_1, x:\verctype{A}{r}, \Gamma_2) \vdash t : B
            \andalso
            \vdash \Delta'
        }{
            (\Gamma_1, x:\verctype{A}{r}, \Gamma_2) + \verctype{\Delta'}{0} \vdash t : B
        }
    \end{minipage}
\end{center}

We then compare the typing context between the conclusion of the lemma and that of the that of above derivation as follows:
\begin{align*}
(\Gamma, x:\verctype{A}{r}, \Gamma')
    &= (\Gamma_1, x:\verctype{A}{r}, \Gamma_2) + \verctype{\Delta'}{0}\\
    &= (\Gamma_1 + \incl{(\verctype{\Delta'}{0})}{\Gamma_1}),x:\verctype{A}{r},(\Gamma_2 + \excl{(\verctype{\Delta'}{0})}{\Gamma_1})\tag{$\because$ Lemma \ref{lemma:shuffle} (1)}
\end{align*}
By the commutativity of "$,$", we can take $\Gamma$ and $\Gamma'$ arbitrarily so that they satisfy the above equation. So here we know $\Gamma = (\Gamma_1 + \incl{(\verctype{\Delta'}{0})}{\Gamma_1})$ and $\Gamma' = (\Gamma_2 + \excl{(\verctype{\Delta'}{0})}{\Gamma_1})$.
We then apply the induction hypothesis to the premise of the last derivation and reapply (\textsc{weak}) as follows:
\begin{center}
    \begin{minipage}{.7\linewidth}
        \infrule[weak]{
            \Gamma_1 + r\cdot\Delta + \Gamma_2 \vdash [t'/x]t : B
            \andalso
            \vdash \Delta'
        }{
            (\Gamma_1 + r\cdot\Delta + \Gamma_2) + \verctype{\Delta'}{0} \vdash [t'/x]t : B
        }
    \end{minipage}
\end{center}
The conclusion of the above derivation is equivalent to the conclusion of the lemma except for the typing contexts.
Finally, we must show that $(\Gamma + r\cdot\Delta + \Gamma') = (\Gamma_1 + r\cdot\Delta + \Gamma_2) + \verctype{\Delta'}{0}$.
\begin{align*}
(\Gamma + r\cdot\Delta + \Gamma')
    &= (\Gamma_1 + \incl{(\verctype{\Delta'}{0})}{\Gamma_1}) + r\cdot\Delta + (\Gamma_2 + \excl{(\verctype{\Delta'}{0})}{\Gamma_1}) \tag{$\because$ $\Gamma = (\Gamma_1 + \incl{(\verctype{\Delta'}{0})}{\Gamma_1})$ \& $\Gamma' = (\Gamma_2 + \excl{(\verctype{\Delta'}{0})}{\Gamma_1})$}\\
    &= (\Gamma_1 + r\cdot\Delta + \Gamma_2) + (\incl{\verctype{\Delta'}{0}}{\Gamma_1} + \excl{\verctype{\Delta'}{0}}{\Gamma_1}) \tag{$\because$ $+$ associativity and commutativity}\\
    &= (\Gamma_1 + r\cdot\Delta + \Gamma_2) + \verctype{\Delta'}{0}\tag{$\because$ Lemma \ref{lemma:restriction}}
\end{align*}
Thus, we obtain the conclusion of the lemma.\\
\end{itemize}

\item Case (\textsc{der})
\begin{center}
    \begin{minipage}{.38\linewidth}
        \infrule[der]{
            \Gamma'', y:B_1 \vdash t : B_2
        }{
            \Gamma'', y:\verctype{B_1}{1} \vdash t : B_2
        }
    \end{minipage}
\end{center}
In this case, we know $(\Gamma'', y:\verctype{B_1}{1}) = (\Gamma, x:\verctype{A}{r}, \Gamma')$.
There are two cases in which the versioned assumption $x:\verctype{A}{r}$ is equivalent to $y:\verctype{B_1}{1}$ and not equivalent to.
\begin{itemize}
\item Suppose $x:\verctype{A}{r} = y:\verctype{B_1}{1}$\\
We know $x=y$, $A=B_1$, $r=1$, $\Gamma = \Gamma''$, and $\Gamma' = \emptyset$.
The last derivation is rewritten as follows:
\begin{center}
    \begin{minipage}{.38\linewidth}
        \infrule[der]{
            \Gamma'', x:A \vdash t : B_2
        }{
            \Gamma'', x:\verctype{A}{1} \vdash t : B_2
        }
    \end{minipage}
\end{center}
We then apply Lemma \ref{lemma:substitution1} to the premise to obtain the following:
\begin{align*}
    \Gamma'' + \Delta \vdash [t'/x]t : B_2
\end{align*}

Note that $\Delta$ is a versioned assumption by the assumption 1 and thus $\Gamma'' + \Delta = \Gamma'' + r\cdot\Delta$ where $r=1$, we obtain the conclusion of the lemma.

\item Suppose $x:\verctype{A}{r} \neq y:\verctype{B_1}{1}$\\
Let $\Gamma_1$ and $\Gamma_2$ be typing contexts such that $\Gamma'' = (\Gamma_1, x:\verctype{A}{r}, \Gamma_1')$.
The last derivation is rewritten as follows:
\begin{center}
    \begin{minipage}{.60\linewidth}
        \infrule[der]{
            (\Gamma_1, x:\verctype{A}{r}, \Gamma_1'), y:B_1 \vdash t : B_2
        }{
            (\Gamma_1, x:\verctype{A}{r}, \Gamma_1'), y:\verctype{B_1}{1} \vdash t : B_2
        }
    \end{minipage}
\end{center}
We then apply the induction hypothesis to the premise of the last derivation and reapply (\textsc{der}) to obtain the following:
\begin{center}
    \begin{minipage}{.70\linewidth}
        \infrule[der]{
            (\Gamma + r\cdot\Delta + \Gamma''), y:B_1 \vdash [t'/x]t : B_2
        }{
            (\Gamma + r\cdot\Delta + \Gamma''), y:\verctype{B_1}{1} \vdash [t'/x]t : B_2
        }
    \end{minipage}
\end{center}
Since $y:\verctype{B_1}{1}$ is desjoint with $\Gamma+r\cdot\Delta+\Gamma''$ and thus $((\Gamma + r\cdot\Delta + \Gamma''), y:\verctype{B_1}{1}) = \Gamma + r\cdot\Delta + (\Gamma'', y:\verctype{B_1}{1})$, we obtain the conclusion of the lemma.\\
\end{itemize}

\item Case (\textsc{pr})
\begin{center}
    \begin{minipage}{.40\linewidth}
        \infrule[pr]{
             [\Gamma_1] \vdash t : B
             \andalso
             \vdash r'
        }{
            r'\cdot[\Gamma_1] \vdash [t] : \vertype{r'}{B}
        }
    \end{minipage}
\end{center}
Let $r''$ be a version resouce and $\Gamma'_1$ and $\Gamma''_1$ be typing contexts such that $r'' \sqsubseteq r'$ and $[\Gamma_1] = [\Gamma'_1, x:\verctype{A}{r''}, \Gamma''_1]$.
The last derivation is rewritten as follows:
\begin{center}
    \begin{minipage}{.55\linewidth}
        \infrule[pr]{
             [\Gamma'_1, x:\verctype{A}{r''}, \Gamma''_1] \vdash t : B
        }{
            r'\cdot[\Gamma'_1, x:\verctype{A}{r''}, \Gamma''_1] \vdash [t] : \vertype{r'}{B} 
        }
    \end{minipage}
\end{center}
We then compare the conclusion of the lemma and the above conclusion.
\begin{align*}
(\Gamma, x:\verctype{A}{r}, \Gamma') &= r' \cdot [\Gamma_1]\\
&= r' \cdot [\Gamma'_1, x:\verctype{A}{r''}, \Gamma''_1]\tag{$\because$ $[\Gamma_1] = [\Gamma'_1, x:\verctype{A}{r''}, \Gamma''_1]$}\\
&= r'\cdot[\Gamma'_1], \ x:\verctype{A}{r''\otimes r'}, \ r'\cdot[\Gamma''_1]\tag{$\because$ $\cdot$ definition}\\
&= r'\cdot[\Gamma'_1], \ x:\verctype{A}{r'}, \ r'\cdot[\Gamma''_1] \tag{$\because$ $r''\sqsubseteq r'$}
\end{align*}
By the commutativity of "$,$", we can take $\Gamma$ and $\Gamma'$ arbitrarily so that they satisfy the above equation. So here we know $\Gamma = (r'\cdot[\Gamma'_1])$ and $\Gamma' = (r'\cdot[\Gamma''_1])$.\par
We then apply the induction hypothesis to the premise of the last derivation and reapply (\textsc{pr}) to obtain the following:
\begin{center}
    \begin{minipage}{.6\linewidth}
        \infrule[pr]{
             [\Gamma'_1 + r''\cdot\Delta + \Gamma''_1] \vdash [t'/x]t : B
             \andalso
             \vdash r'
        }{
            r'\cdot[\Gamma'_1 + r''\cdot\Delta + \Gamma''_1] \vdash [[t'/x]t] : \vertype{r'}{B} 
        }
    \end{minipage}
\end{center}
where we use $[\Gamma'_1, x:\verctype{A}{r''}, \Gamma''_1] = [\Gamma'_1], x:\verctype{A}{r''}, [\Gamma''_1]$ and $[\Gamma'_1 + r''\cdot\Delta + \Gamma''_1] = [\Gamma'_1] + r''\cdot\Delta + [\Gamma''_1]$ before applying (\textsc{pr}).\par
Since $[[t'/x]t] = [t'/x][t]$ by the definiton of substitution, the above conclusion is equivalent to the conclusion of the lemma except for the typing contexts.
Finally, we must show that $(\Gamma + r'\cdot\Delta + \Gamma') = r'\cdot[\Gamma'_1 + r''\cdot\Delta + \Gamma''_1]$ by the following reasoning:
\begin{align*}
(\Gamma + r'\cdot\Delta + \Gamma')
&= r'\cdot[\Gamma'_1] + r'\cdot\Delta + r'\cdot[\Gamma'_1]\tag{$\because$ $\Gamma = (r'\cdot[\Gamma'_1])$ \& $\Gamma' = (r'\cdot[\Gamma''_1])$}\\
&= r'\cdot[\Gamma'_1] + (r'\otimes r'')\cdot\Delta + r'\cdot[\Gamma'_1] \tag{$\because$ $r''\sqsubseteq r'$}\\
&= r'\cdot[\Gamma'_1] + r'\cdot ( r''\cdot\Delta) + r'\cdot[\Gamma'_1] \tag{$\because$ $\otimes$ associativity}\\
&= r'\cdot[\Gamma'_1 + r''\cdot\Delta + \Gamma'_1] \tag{$\because$ $\cdot$ distributive law over $+$}
\end{align*}
Thus, we obtain the conclusion of the lemma.
\\

\item Case (\textsc{ver})
\begin{center}
    \begin{minipage}{.60\linewidth}
        \infrule[ver]{
            [\Gamma_i] \vdash t_i : B
            \andalso
            \vdash \{\overline{l_i}\}
        }{
            \bigcup_i(\{l_i\}\cdot [\Gamma_i]) \vdash \nvval{\overline{l_i=t_i}} : \vertype{\{\overline{l_i}\}}{B}
        }
    \end{minipage}
\end{center}
% By the second premise, $\forall i,j.\mathrm{dom}(\Gamma_i) = \mathrm{dom}(\Gamma_i) = $.
% By using this partition, the last derivation is rewritten as follows:
% \begin{center}
%     \begin{minipage}{.60\linewidth}
%         \infrule[ver]{
%             [\Gamma'_i, x:\verctype{A}{\sigma_i}, \Gamma''_i] \vdash t_i : B
%         }{
%             \bigcup_i(\{l_i\}\cdot [\Gamma'_i, x:\verctype{A}{\sigma_i}, \Gamma''_i]) \vdash \nvval{\overline{l_i=t_i}} : \vertype{\overline{l_i}}{B}
%         }
%     \end{minipage}
% \end{center}
We compare the typing contexts between the lemma and the above conclusion as follows:
\begin{align*}
(\Gamma, x:\verctype{A}{r}, \Gamma') &= \bigcup_{i} \ \left(\{l_i\}\cdot[\Gamma_i]\right)\\
&= \bigcup_{i\in I_x} \ \left(\{l_i\}\cdot[\Gamma'_i, \ x:\verctype{A}{r_i}, \ \Gamma''_i]\right) + \bigcup_{i\in J_x} \ \left(\{l_i\}\cdot[\Gamma'_i, \ \Gamma''_i]\right)  \tag{$\because$ $I_x=\{i\,|\,x\in\mathrm{dom}(\Gamma_i)\}$ and $J_x=\{i\,|\,x\notin\mathrm{dom}(\Gamma_i)\}$}
\end{align*}
We then reorganise the typing context $\bigcup_{i\in I_x} \ \left(\{l_i\}\cdot[\Gamma'_i, \ x:\verctype{A}{r_i}, \ \Gamma''_i]\right)$ as follows:
\begin{align*}
&\bigcup_{i\in I_x} \ \left(\{l_i\}\cdot[\Gamma'_i, \ x:\verctype{A}{r_i}, \ \Gamma''_i]\right)\\
= &\bigcup_{i\in I_x} \ \left(\{l_i\}\cdot[x:\verctype{A}{r_i}, \ \Gamma'_i, \ \Gamma''_i]\right) \tag{$\because$ $,$ associativity}\\
= &\bigcup_{i\in I_x} \ \left(\{l_i\}\cdot(x:\verctype{A}{r_i}), \ \{l_i\}\cdot[\Gamma'_i], \ \{l_i\}\cdot[\Gamma''_i]\right) \tag{$\because$ $\cdot$ distributive law}\\
= &\bigcup_{i\in I_x} \ \left(\{l_i\}\cdot(x:\verctype{A}{r_i})\right), \ \bigcup_{i\in I_x} \ \left(\{l_i\}\cdot[\Gamma'_i], \{l_i\}\cdot[\Gamma''_i]\right) \tag{$\because$ Sum of each disjoint sub context}\\
= &\bigcup_{i\in I_x} \ \left(x:\verctype{A}{\{l_i\}\otimes r_i}\right), \ \bigcup_{i\in I_x} \ \left(\{l_i\}\cdot[\Gamma'_i], \{l_i\}\cdot[\Gamma''_i]\right) \tag{$\because$ $\cdot$ definition}\\
= &\ x:\verctype{A}{\sum_{i\in I_x}\{l_i\}\otimes r_i}, \ \ \bigcup_{i\in I_x} \ \left(\{l_i\}\cdot[\Gamma'_i], \{l_i\}\cdot[\Gamma''_i]\right) \tag{$\because$ $\bigcup$ and $+$ definition}
\end{align*}
Thus, we obtain the following:
\begin{align*}
&(\Gamma, x:\verctype{A}{r}, \Gamma')\\
= &\left(x:\verctype{A}{\sum_{i\in I_x}\{l_i\}\otimes r_i}, \ \ \bigcup_{i\in I_x} \ \left(\{l_i\}\cdot[\Gamma'_i], \{l_i\}\cdot[\Gamma''_i]\right)\right) + \bigcup_{i\in J_x} \ \left(\{l_i\}\cdot[\Gamma'_i, \ \Gamma''_i]\right)\\
= &\ x:\verctype{A}{\sum_{i\in I_x}\{l_i\}\otimes r_i}, \ \ \bigcup_{i} \ \left(\{l_i\}\cdot[\Gamma'_i], \{l_i\}\cdot[\Gamma''_i]\right) \tag{$\because$ $\bigcup_{i\in J_x} \ \left(\{l_i\}\cdot[\Gamma'_i, \ \Gamma''_i]\right)$ are disjoint with $x:\verctype{A}{\sum_{i\in I_x}\{l_i\}\otimes r_i}$}
\end{align*}
Therefore, By Lemma \ref{lemma:shufflecomposition}, there exists typing contexts $\Gamma'_{\overline{i}}$ and $\Gamma''_{\overline{i}}$ such that:
\begin{align*}
\Gamma'_{\overline{i}}, \Gamma''_{\overline{i}} &= \bigcup_{i} \ \left(\{l_i\}\cdot[\Gamma'_i], \{l_i\}\cdot[\Gamma''_i]\right)\\
\Gamma'_{\overline{i}} + \Gamma''_{\overline{i}} &= \bigcup_{i} \ \left(\{l_i\}\cdot[\Gamma'_i] + \{l_i\}\cdot[\Gamma''_i]\right)
\end{align*}
Thus, we obtain the following:
\begin{align*}
(\Gamma, x:\verctype{A}{r}, \Gamma')
&= x:\verctype{A}{\sum_{i\in I_x}\{l_i\}\otimes r_i}, \ \Gamma'_{\overline{i}}, \ \Gamma''_{\overline{i}}\\
&= \Gamma'_{\overline{i}}, \ x:\verctype{A}{\sum_{i\in I_x}\{l_i\}\otimes r_i}, \ \Gamma''_{\overline{i}} \tag{$\because$ $,$ commutativity}
\end{align*}
By the commutativity of "$,$", we can take $\Gamma$ and $\Gamma'$ arbitrarily so that they satisfy the above equation. So here we know $\Gamma = \Gamma'_{\overline{i}}$, $\Gamma' = \Gamma''_{\overline{i}}$, and $r=\sum_{i\in I_x}(\{l_i\}\otimes r_i)$.
We then apply the induction hypothesis to the premise whose typing context contains $x$.
Here, we define a typing context $\Delta_i$ as follows:
\begin{align*}
\Delta_i =
\left\{
\begin{aligned}
\Delta \hspace{1em} (i\in I_x)\\
\emptyset \hspace{1em} (i\in J_x)\\
\end{aligned}
\right.
\end{align*}
By using $\Delta_i$, we reapply (\textsc{ver}) as follows:
\begin{center}
    \begin{minipage}{.8\linewidth}
        \infrule[ver]{
            [\Gamma'_i + r_i\cdot\Delta_i + \Gamma''_i] \vdash [t'/x]t_i : B
            \andalso
            \vdash \{\overline{l_i}\}
        }{
            \bigcup_i(\{l_i\}\cdot [\Gamma'_i + r_i\cdot\Delta_i + \Gamma''_i]) \vdash \nvval{\overline{l_i=[t'/x]t_i}} : \vertype{\{\overline{l_i}\}}{B}
        }
    \end{minipage}
\end{center}
Since $\{\overline{l_i=[t'/x]t_i}\,|\,l_k\} = [t'/x]\{\overline{l_i=t_i}\,|\,l_k\}$ by the definition of substitution, the above conclusion is equivalent to the conclusion of the lemma except for typing contexts.
Finally, we must show that $(\Gamma + r\cdot\Delta + \Gamma') = \bigcup_i(\{l_i\}\cdot [\Gamma'_i + r_i\cdot\Delta_i + \Gamma''_i])$.
\begin{align*}
(\Gamma + r\cdot\Delta + \Gamma')
&= \Gamma'_{\overline{i}} + r\cdot\Delta + \Gamma''_{\overline{i}}\tag{$\because$ $\Gamma = \Gamma'_{\overline{i}}$ \& $\Gamma' = \Gamma''_{\overline{i}}$}\\
&= r\cdot\Delta + (\Gamma'_{\overline{i}} + \Gamma''_{\overline{i}}) \tag{$+$ associativity \& commutativity}\\
&= r\cdot\Delta + \bigcup_{i} \ \left(\{l_i\}\cdot[\Gamma'_i] + \{l_i\}\cdot[\Gamma''_i]\right) \tag{$\because$ $\Gamma'_{\overline{i}} + \Gamma''_{\overline{i}} = \bigcup_{i} \ \left(\{l_i\}\cdot[\Gamma'_i] + \{l_i\}\cdot[\Gamma''_i]\right)$}\\
&= (\sum_{i\in I_x}(\{l_i\}\otimes r_i))\cdot\Delta + \bigcup_{i} \ \left(\{l_i\}\cdot[\Gamma'_i] + \{l_i\}\cdot[\Gamma''_i]\right) \tag{$\because$ $r=\sum_{i\in I_x}(\{l_i\}\otimes r_i)$}\\
&= \bigcup_{i\in I_x}\left(\{l_i\}\cdot\left(r_i\cdot\Delta\right)\right) + \bigcup_{i} \ \left(\{l_i\}\cdot[\Gamma'_i] + \{l_i\}\cdot[\Gamma''_i]\right) \tag{$\because$ $\bigcup$ definition}\\
&= \bigcup_{i}\left(\{l_i\}\cdot\left(r_i\cdot\Delta_i\right)\right) + \bigcup_{i} \ \left(\{l_i\}\cdot[\Gamma'_i] + \{l_i\}\cdot[\Gamma''_i]\right) \tag{$\because$ $\Delta_i$ definition}\\
&= \bigcup_i\left(\{l_i\}\cdot\left(r_i\cdot\Delta_i\right) + \{l_i\}\cdot[\Gamma'_i] + \{l_i\}\cdot[\Gamma''_i]\right) \tag{$\because$ $+$ commutativity \& associativity}\\
&= \bigcup_{i} \left(\{l_i\}\cdot\left(\left(r_i\cdot\Delta_i\right) + [\Gamma'_i] + [\Gamma''_i]\right)\right) \tag{$\because$ districutive law}\\
&= \bigcup_{i} \left(\{l_i\}\cdot\left([\Gamma'_i] + \left(r_i\cdot\Delta_i\right) + [\Gamma''_i]\right)\right) \tag{$\because$ $+$ commutativity}\\
&= \bigcup_{i} \left(\{l_i\}\cdot[\Gamma'_i + r_i\cdot\Delta_i + \Gamma''_i]\right) \tag{$\because$ $[\cdot]$ definition}
\end{align*}
Thus, we obtain the conclusion of the lemma.
\\

\item Case (\textsc{veri})
\begin{center}
    \begin{minipage}{.55\linewidth}
        \infrule[veri]{
            \verctype{\Gamma_i}{} \vdash t_i : B
            \andalso
            \vdash \{\overline{l_i}\}
            \andalso
            l_k \in \{\overline{l_i}\}
        }{
            \bigcup_i(\{l_i\}\cdot [\Gamma_i]) \vdash \ivval{\overline{l_i=t_i}}{l_k} : B
        }
    \end{minipage}
\end{center}
This case is similar to the case of (\textsc{ver}).
\\

\item Case (\textsc{extr})
\begin{center}
    \begin{minipage}{.45\linewidth}
        \infrule[extr]{
            \Gamma \vdash t : \vertype{r}{A}
            \andalso
            l \in r
        }{
            \Gamma \vdash t.l : A
        }
    \end{minipage}
\end{center}
In this case, we apply the induction hypothesis to the premise and then reapply (\textsc{extr}), we obtain the conclusion of the lemma.
\\

\item Case (\textsc{sub})
\begin{center}
    \begin{minipage}{.55\linewidth}
        \infrule[\textsc{sub}]{
            \Gamma_1,y:\verctype{B'}{r_1}, \Gamma_2 \vdash t : B
            \andalso
            r_1 \sqsubseteq r_2
            \andalso
            \vdash r_2
        }{
            \Gamma_1,y:\verctype{B'}{r_2}, \Gamma_2 \vdash t : B
        }
    \end{minipage}
\end{center}
In this case, we know $(\Gamma, x:\verctype{A}{r}, \Gamma') = (\Gamma, y:\verctype{B'}{r_2}, \Gamma')$.
There are three cases where the versioned assumption $x:\verctype{A}{r}$ is included in $\Gamma_1$, included in $\Gamma_2$, or equal to $y:\verctype{B'}{r_2}$.
\begin{itemize}
\item Suppose $(x:\verctype{A}{r})\in\Gamma_1$.\\
Let $\Gamma'_1$ and $\Gamma''_1$ be typing contexts such that $\Gamma_1 = (\Gamma'_1, x:\verctype{A}{r}, \Gamma''_1)$.
The last derivation is rewritten as follows:
\begin{center}
    \begin{minipage}{.75\linewidth}
        \infrule[\textsc{sub}]{
            \Gamma'_1, x:\verctype{A}{r}, \Gamma''_1,y:\verctype{B'}{r_1}, \Gamma_2 \vdash t : B
            \andalso
            r_1 \sqsubseteq r_2
            \andalso
            \vdash r_2
        }{
            \Gamma'_1, x:\verctype{A}{r}, \Gamma''_1,y:\verctype{B'}{r_2}, \Gamma_2 \vdash t : B
        }
    \end{minipage}
\end{center}
We then apply the induction hypothesis to the premise of the last derivation to obtain the following:
\begin{align}
\label{eq:subset1}
\Gamma'_1 + r\cdot\Delta + (\Gamma''_1, y:\verctype{B'}{r_1}, \Gamma_2) \vdash [t'/x]t : B
\end{align}
The typing context of the above conclusion can be transformed as follows:
\begin{align*}
&\Gamma'_1 + r\cdot\Delta + (\Gamma''_1, y:\verctype{B'}{r_1}, \Gamma_2)\\
=\ &(\Gamma'_1+\incl{(r\cdot\Delta)}{\Gamma'_1}),\excl{(r\cdot\Delta)}{(\Gamma'_1,(\Gamma''_1, y:\verctype{B'}{r_1}, \Gamma_2))},\\
&\hspace{5em} \left((\Gamma''_1, y:\verctype{B'}{r_1}, \Gamma_2)+\incl{(r\cdot\Delta)}{(\Gamma''_1, y:\verctype{B'}{r_1}, \Gamma_2)}\right) \tag{$\because$ Lemma \ref{lemma:collapse}}\\
=\ &(\Gamma'_1+\incl{(r\cdot\Delta)}{\Gamma'_1}),\excl{(r\cdot\Delta)}{(\Gamma'_1,(\Gamma''_1, y:\verctype{B'}{r_1}, \Gamma_2))},\\
&\hspace{5em} \left((\Gamma''_1, y:\verctype{B'}{r_1}, \Gamma_2) + \left(\incl{(r\cdot\Delta)}{\Gamma''_1}, \incl{(r\cdot\Delta)}{(y:\verctype{B'}{r_1})}, \incl{(r\cdot\Delta)}{\Gamma_2})\right)\right) \tag{$\because$ \ref{def:restriction}}\\
=\ &(\Gamma'_1+\incl{(r\cdot\Delta)}{\Gamma'_1}),\excl{(r\cdot\Delta)}{(\Gamma'_1,(\Gamma''_1, y:\verctype{B'}{r_1}, \Gamma_2))},\\
&\hspace{5em} \left(\Gamma''_1+\incl{(r\cdot\Delta)}{\Gamma''_1}\right), \left(y:\verctype{B'}{r_1}+ \incl{(r\cdot\Delta)}{(y:\verctype{B'}{r_1})}\right), \left(\Gamma_2 + \incl{(r\cdot\Delta)}{\Gamma_2}\right) \tag{$\because$ Lemma \ref{lemma:shuffle}}\\
=\ &\Gamma_3, y:\verctype{B'}{r_1\oplus r_3}, \Gamma'_3
\end{align*}
The last equational transformation holds by the following equation \ref{eq:subset2}.\par
Let $\Gamma_3$ and $\Gamma'_3$ be typing contexts that satisfy the following:
\begin{align*}
\Gamma_3 &= (\Gamma'_1+\incl{(r\cdot\Delta)}{\Gamma'_1}),\excl{(r\cdot\Delta)}{(\Gamma'_1,(\Gamma''_1, y:\verctype{B'}{r_1}, \Gamma_2))}, \left(\Gamma''_1+\incl{(r\cdot\Delta)}{\Gamma''_1}\right)\\
\Gamma'_3 &= \left(\Gamma_2 + \incl{(r\cdot\Delta)}{\Gamma_2}\right)
\end{align*}

For $\incl{(r\cdot\Delta)}{(y:\verctype{B'}{r_1})}$,  Let $r_3$ and $r_3'$ be typing contexts such that $r_3=r\otimes r'_3$ and stisfy the following:
\begin{align*}
\incl{(r\cdot\Delta)}{(y:\verctype{B'}{r_1})} =
\left\{
\begin{aligned}
    & r\cdot(y:\verctype{B'}{r'_3}) = y:\verctype{B'}{r\otimes r'_3} = y:\verctype{B'}{r_3}  & (y \in \mathrm{dom}(\Delta))\\
    & \emptyset & (y \notin \mathrm{dom}(\Delta))
\end{aligned}
\right.
\end{align*}
Thus, we obtain the following equation.
\begin{align}
\label{eq:subset2}
y:\verctype{B'}{r_1} + \incl{(r\cdot\Delta)}{(y:\verctype{B'}{r_1})} = 
\left\{
\begin{aligned}
    & y:\verctype{B'}{r_1\oplus r_3} & (y \in \mathrm{dom}(\Delta))\\
    & y:\verctype{B'}{r_1\oplus r_3} = y:\verctype{B'}{r_1} & (y \notin \mathrm{dom}(\Delta))
\end{aligned}
\right.
\end{align}
Applying all of the above transformations and reapplying (\textsc{sub}) to the expression \ref{eq:subset1}, we obtain the following:
\begin{center}
    \begin{minipage}{\linewidth}
        \infrule[\textsc{sub}]{
            \Gamma_3, y:\verctype{B'}{r_1\oplus r_3}, \Gamma'_3 \vdash [t'/x]t : B
            \andalso
            (r_1 \oplus r_3) \sqsubseteq (r_2 \oplus r_3)
            \andalso
            \vdash r_2 \oplus r_3
        }{
            \Gamma_3, y:\verctype{B'}{r_2\oplus r_3}, \Gamma'_3 \vdash [t'/x]t : B
        }
    \end{minipage}
\end{center}
The conclusion of the above derivation is equivalent to the conclusion of the lemma except for the typing contexts. \par
Finally, we must show that $\Gamma'_1 + r\cdot\Delta + (\Gamma''_1, y:\verctype{B'}{r_2}, \Gamma_2) = (\Gamma_3, y:\verctype{B'}{r_2\oplus r_3}, \Gamma'_3)$.
\begin{align*}
&\Gamma'_1 + r\cdot\Delta + (\Gamma''_1, y:\verctype{B'}{r_2}, \Gamma_2)\\
=\ &(\Gamma'_1+\incl{(r\cdot\Delta)}{\Gamma'_1}),\excl{(r\cdot\Delta)}{(\Gamma'_1,(\Gamma''_1, y:\verctype{B'}{r_2}, \Gamma_2))},\\
&\hspace{5em} \left((\Gamma''_1, y:\verctype{B'}{r_2}, \Gamma_2)+\incl{(r\cdot\Delta)}{(\Gamma''_1, y:\verctype{B'}{r_2}, \Gamma_2)}\right) \tag{$\because$ Lemma \ref{lemma:collapse}}\\
=\ &(\Gamma'_1+\incl{(r\cdot\Delta)}{\Gamma'_1}),\excl{(r\cdot\Delta)}{(\Gamma'_1,(\Gamma''_1, y:\verctype{B'}{r_2}, \Gamma_2))},\\
&\hspace{5em} \left((\Gamma''_1, y:\verctype{B'}{r_2}, \Gamma_2) + \left(\incl{(r\cdot\Delta)}{\Gamma''_1}, \incl{(r\cdot\Delta)}{(y:\verctype{B'}{r_2})}, \incl{(r\cdot\Delta)}{\Gamma_2})\right)\right) \tag{$\because$ \ref{def:restriction}}\\
=\ &(\Gamma'_1+\incl{(r\cdot\Delta)}{\Gamma'_1}),\excl{(r\cdot\Delta)}{(\Gamma'_1,(\Gamma''_1, y:\verctype{B'}{r_2}, \Gamma_2))},\\
&\hspace{5em} \left(\Gamma''_1+\incl{(r\cdot\Delta)}{\Gamma''_1}\right), \left(y:\verctype{B'}{r_2}+ \incl{(r\cdot\Delta)}{(y:\verctype{B'}{r_2})}\right), \left(\Gamma_2 + \incl{(r\cdot\Delta)}{\Gamma_2}\right) \tag{$\because$ Lemma \ref{lemma:shuffle}}\\
=\ &\Gamma_3, y:\verctype{B'}{r_2\oplus r_3}, \Gamma'_3
\end{align*}
The last transformation is based on the following equation that can be derived from the definition \ref{def:restriction}.
\begin{align*}
\excl{(r\cdot\Delta)}{(\Gamma'_1,(\Gamma''_1, y:\verctype{B'}{r_1}, \Gamma_2))} &= \excl{(r\cdot\Delta)}{(\Gamma'_1,(\Gamma''_1, y:\verctype{B'}{r_2}, \Gamma_2))}\\
\incl{(r\cdot\Delta)}{(y:\verctype{B'}{r_1})} &= \incl{(r\cdot\Delta)}{(y:\verctype{B'}{r_2})}
\end{align*}
Thus, we obtain the conclusion of the lemma.

\item Suppose $(x:\verctype{A}{r})\in\Gamma_2$.\\
This case is similar to the case of $(x:\verctype{A}{r})\in\Gamma_1$.

\item Suppose $(x:\verctype{A}{r}) = y:\verctype{B'}{r_2}$.\\
The last derivation is rewritten as follows:
\begin{center}
    \begin{minipage}{.55\linewidth}
        \infrule[\textsc{sub}]{
            \Gamma_1,x:\verctype{A}{r'}, \Gamma_2 \vdash t : B
            \andalso
            r' \sqsubseteq r
            \andalso
            \vdash r
        }{
            \Gamma_1,x:\verctype{A}{r}, \Gamma_2 \vdash t : B
        }
    \end{minipage}
\end{center}
We apply the induction hypothesis to the premise and then reapply (\textsc{sub}), we obtain the conclusion of the lemma.
\end{itemize}

\end{itemize}
\end{proof}

%%%%%%%%%%%%%%%%%%%%%%%%%%%%%%%%%%%%%%%%%%%%%%%%%%%%%%%%%%%%%%%%%%%%%%%%%%%%%
%%%%%%%%%%%%%%%%%%%%%%%%%%%%%%%%%%%%%%%%%%%%%%%%%%%%%%%%%%%%%%%%%%%%%%%%%%%%%
%%%%%%%%%%%%%%%%%%%%%%%%%%%%%%%%%%%%%%%%%%%%%%%%%%%%%%%%%%%%%%%%%%%%%%%%%%%%%
%%%%%%%%%%%%%%%%%%%%%%%%%%%%%%%%%%%%%%%%%%%%%%%%%%%%%%%%%%%%%%%%%%%%%%%%%%%%%
%%%%%%%%%%%%%%%%%%%%%%%%%%%%%%%%%%%%%%%%%%%%%%%%%%%%%%%%%%%%%%%%%%%%%%%%%%%%%
%%%%%%%%%%%%%%%%%%%%%%%%%%%%%%%%%%%%%%%%%%%%%%%%%%%%%%%%%%%%%%%%%%%%%%%%%%%%%

\subsection{Type Safety}
\label{appendix:typesafety}
\begin{lemma}[Inversion lemma]
\label{lemma:typedvalue}
Let $v$ be a value such that $\Gamma \vdash v:A$. The followings hold for a type $A$.
\begin{itemize}
  \item $A=\inttype \Longrightarrow $ $v = n$ for some integer constant $n$.
  \item $A=\vertype{r}{B} \Longrightarrow $ $v = \pr{t'}$ for some term $t'$, or $v=\nvval{\overline{l_i=t_i}}$ for some terms $t_i$ and some labels $l_i\in r$.
  \item $A=\apptype{B}{B'} \Longrightarrow  $ $v=\lam{p}{t}$ for some pattern $p$ and term $t$.
  % \item \emph{otherwise} $\Longrightarrow$ $v=x$ for some variable $x\in \mathrm{dom}(\Gamma)$. (変数は値ではない。)
\end{itemize}
\end{lemma}

\begin{lemma}[Type safety for default version overwriting]
\label{lemma:overwriting}

For any version label $l$:
\begin{align*}
    \Gamma \vdash t : A
    \hspace{1em}\Longrightarrow\hspace{1em}
    \Gamma \vdash t@l : A
\end{align*}
\end{lemma}
\begin{proof}
The proof is given by induction on the typing derivation of $\Gamma \vdash t : A$.
Consider the cases for the last rule used in the typing derivation of assumption.
\\

\begin{itemize}
\item Case (\textsc{int})
\begin{center}
    \begin{minipage}{.25\linewidth}
        \infrule[int]{
            \\
        }{
            \emptyset \vdash n:\textsf{Int}
        }
    \end{minipage}
\end{center}
This case holds trivially because $n@l \equiv n$ for any labels $l$.
\\

\item Case (\textsc{var})
\begin{center}
    \begin{minipage}{.35\linewidth}
        \infrule[var]{
            \vdash A
        }{
            x:A \vdash x:A
        }
    \end{minipage}
\end{center}
This case holds trivially because $x@l = x$ for any labels $l$.
\\

\item Case (\textsc{abs})
\begin{center}
    \begin{minipage}{.55\linewidth}
        \infrule[abs]{
            - \vdash p : B_1 \rhd \Delta'
            \andalso
            \Gamma, \Delta' \,\vdash\, t_1 : A_2%\theta B
        }{
            \Gamma \,\vdash\, \lam{p}{t_1} : \ftype{A_1}{A_2}
        }
    \end{minipage}
\end{center}
By induction hypothesis, there exists a term $t_1@l$ such that:
\begin{align*}
\Gamma, \Delta' \,\vdash\, t_1@l : A_2
\end{align*}
We then reapply (\textsc{abs}) to obtain the following:
\begin{center}
    \begin{minipage}{.65\linewidth}
        \infrule[abs]{
            - \vdash p : B_1 \rhd \Delta'
            \andalso
            \Gamma, \Delta' \,\vdash\, t_1@l : A_2%\theta B
        }{
            \Gamma \,\vdash\, \lam{p}{(t_1@l)} : \ftype{A_1}{A_2}
        }
    \end{minipage}
\end{center}
Thus, note that $(\lam{p}{t_1})@l \equiv \lam{p}{(t_1@l)}$, we obtain the conclusion of the lemma.
\\

\item Case (\textsc{app})
\begin{center}
    \begin{minipage}{.65\linewidth}
        \infrule[app]{
            \Gamma_1 \vdash t_1 : \ftype{B}{A}
            \andalso
            \Gamma_2 \vdash t_2 : B
        }{
            \Gamma_1 + \Gamma_2 \vdash \app{t_1}{t_2} : A
        }
    \end{minipage}
\end{center}
By induction hypothesis, there exists terms $t_1@l$ and $t_2@l$ such that:
\begin{align*}
\Gamma_1 &\vdash t_1@l : \ftype{B}{A}\\
\Gamma_2 &\vdash t_2@l : B
\end{align*}
We then reapply (\textsc{app}) to obtain the following:
\begin{center}
    \begin{minipage}{.7\linewidth}
        \infrule[app]{
            \Gamma_1 \vdash t_1@l : \ftype{B}{A}
            \andalso
            \Gamma_2 \vdash t_2@l : B
        }{
            \Gamma_1 + \Gamma_2 \vdash \app{(t_1@l)}{(t_2@l)} : A
        }
    \end{minipage}
\end{center}
Thus, note that $(\app{t_1}{t_2})@l \equiv \app{(t_1@l)}{(t_2@l)}$, we obtain the conclusion of the lemma.
\\

\item Case (\textsc{let})
\begin{center}
    \begin{minipage}{.70\linewidth}
        \infrule[let]{
            \Gamma_1 \,\vdash\, t_1 : \vertype{r}{A}
            \andalso
            \Gamma_2, x:\verctype{A}{r} \,\vdash\, t_2 : B
        }{
            \Gamma_1 + \Gamma_2 \,\vdash\, \clet{x}{t_1}{t_2} : B
        }
    \end{minipage}
\end{center}
By induction hypothesis, there exists terms $t_1@l$ and $2@l$ such that:
\begin{align*}
\Gamma_1 \,&\vdash\, t_1@l : \vertype{r}{A}\\
\Gamma_2, x:\verctype{A}{r} \,&\vdash\, t_2@l : B
\end{align*}
We then reapply (\textsc{let}) to obtain the following:
\begin{center}
    \begin{minipage}{.75\linewidth}
        \infrule[let]{
            \Gamma_1 \,\vdash\, t_1@l : \vertype{r}{A}
            \andalso
            \Gamma_2, x:\verctype{A}{r} \,\vdash\, t_2@l : B
        }{
            \Gamma_1 + \Gamma_2 \,\vdash\, \clet{x}{(t_1@l)}{(t_2@l)} : B
        }
    \end{minipage}
\end{center}
Thus, note that $(\clet{x}{t_1}{t_2})@l \equiv \clet{x}{(t_1@l)}{(t_2@l)}$, we obtain the conclusion of the lemma.
\\

\item Case (\textsc{weak})
\begin{center}
    \begin{minipage}{.45\linewidth}
        \infrule[weak]{
            \Gamma_1 \vdash t : A
            \andalso
            \vdash \Delta'
        }{
            \Gamma_1 + \verctype{\Delta'}{0} \vdash t : A
        }
    \end{minipage}
\end{center}
By induction hypothesis, we know the following:
\begin{align*}
\Gamma_1 \vdash t@l : A
\end{align*}
We then reapply (\textsc{weak}) to obtain the following:
\begin{center}
    \begin{minipage}{.45\linewidth}
        \infrule[weak]{
            \Gamma_1 \vdash t@l : A
        }{
            \Gamma_1 + \verctype{\Delta'}{0} \vdash t@l : A
        }
    \end{minipage}
\end{center}
Thus, we obtain the conclusion of the lemma.
\\

\item Case (\textsc{der})
\begin{center}
    \begin{minipage}{.45\linewidth}
        \infrule[der]{
            \Gamma_1, x:B \vdash t : A
        }{
            \Gamma_1, x:\verctype{B}{1} \vdash t : A
        }
    \end{minipage}
\end{center}
By induction hypothesis, there exists terms $t@l$ such that:
\begin{align*}
\Gamma_1, x:B \vdash t@l : A
\end{align*}
We then reapply (\textsc{der}) to obtain the following:
\begin{center}
    \begin{minipage}{.45\linewidth}
        \infrule[der]{
            \Gamma_1, x:B \vdash t@l : A
        }{
            \Gamma_1, x:\verctype{B}{1} \vdash t@l : A
        }
    \end{minipage}
\end{center}
Thus, we obtain the conclusion of the lemma.
\\

\item Case (\textsc{pr})
\begin{center}
    \begin{minipage}{.45\linewidth}
        \infrule[pr]{
            [\Gamma] \vdash t : B
            \andalso
            \vdash r
        }{
            r\cdot\verctype{\Gamma}{} \vdash \pr{t} : \vertype{r}{B} 
        }
    \end{minipage}
\end{center}
This case holds trivially  because $\pr{t}@l \equiv \pr{t}$ for any labels $l$.
\\

\item Case (\textsc{ver})
\begin{center}
    \begin{minipage}{.65\linewidth}
        \infrule[ver]{
            \verctype{\Gamma_i}{} \vdash t_i : A
            \andalso
            \vdash \{\overline{l_i}\}
        }{
            \bigcup_i(\{l_i\}\cdot [\Gamma_i]) \vdash \nvval{\overline{l_i=t_i}} : \vertype{\{\overline{l_i}\}}{A}
        }
    \end{minipage}
\end{center}
This case holds trivially because $\nvval{\overline{l_i=t_i}}@l\equiv \nvval{\overline{l_i=t_i}}$ for any labels $l$.
\\

\item Case (\textsc{veri})
\begin{center}
    \begin{minipage}{.55\linewidth}
        \infrule[veri]{
            \verctype{\Gamma_i}{} \vdash t_i : A
            \andalso
            \vdash \{\overline{l_i}\}
            \andalso
            l_k \in \{\overline{l_i}\}
        }{
            \bigcup_i(\{l_i\}\cdot [\Gamma_i]) \vdash \ivval{\overline{l_i=t_i}}{l_k} : A
        }
    \end{minipage}
\end{center}
In this case, there are two possibilities for the one step evaluation of $t$.
\begin{itemize}
\item Suppose $l \in \{\overline{l_i}\}$.\\
We can apply the default version overwriting as follows:
\begin{center}
    \begin{minipage}{.40\linewidth}
        \infrule[]{
            l \in \{\overline{l_i}\}
        }{
            \ivval{\overline{l_i=t_i}}{l_k}@l \,\equiv\, \ivval{\overline{l_i=t_i}}{l}
        }
    \end{minipage}
\end{center}
In this case, we can derive the type of $\ivval{\overline{l_i=t_i}}{l}$ as follows:
\begin{center}
    \begin{minipage}{.65\linewidth}
        \infrule[veri]{
            \verctype{\Gamma_i}{} \vdash t_i : A
        }{
            \bigcup_i(\{l_i\}\cdot [\Gamma_i]) \vdash \ivval{\overline{l_i=t_i}}{l} : A
        }
    \end{minipage}
\end{center}
Thus, we obtain the conclusion of the lemma.
\item Suppose $l \notin \{\overline{l_i}\}$.\\
We can apply the default version overwriting as follows:
\begin{center}
    \begin{minipage}{.50\linewidth}
        \infrule[]{
            l \notin \{\overline{l_i}\}
        }{
            \ivval{\overline{l_i=t_i}}{l_k}@l \,\equiv\, \ivval{\overline{l_i=t_i}}{l_k}
        }
    \end{minipage}
\end{center}
This case holds trivially because $\ivval{\overline{l_i=t_i}}{l_k}@l = \ivval{\overline{l_i=t_i}}{l_k}$.\\
\end{itemize}

\item Case (\textsc{extr})
\begin{center}
    \begin{minipage}{.50\linewidth}
        \infrule[extr]{
            \Gamma \vdash t_1 : \vertype{r}{A}
            \andalso
            l_k \in r
        }{
            \Gamma \vdash t_1.l_k : A
        }
    \end{minipage}
\end{center}
By induction hypothesis, there exists a term $t_1@l$ such that:
\begin{align*}
\Gamma \vdash t_1@l : \vertype{r}{A}
\end{align*}
We then reapply (\textsc{extr}) to obtain the following:
\begin{center}
    \begin{minipage}{.55\linewidth}
        \infrule[extr]{
            \Gamma \vdash t_1@l : \vertype{r}{A}
            \andalso
            l_k \in r
        }{
            \Gamma \vdash (t_1@l).l_k : A
        }
    \end{minipage}
\end{center}
Thus, note that $(t_1.l_k)@l \equiv (t_1@l).l_k$, we obtain the conclusion of the lemma.
\\

\item Case (\textsc{sub})
\begin{center}
    \begin{minipage}{.55\linewidth}
        \infrule[\textsc{sub}]{
            \Gamma_1, x:\verctype{B}{r}, \Gamma_2 \vdash t : A
            \andalso
            r \sqsubseteq s
            \andalso
            \vdash s
        }{
            \Gamma_1, x:\verctype{B}{s}, \Gamma_2 \vdash t : A
        }
    \end{minipage}
\end{center}
By induction hypothesis, there exists a term $t@l$ such that:
\begin{align*}
\Gamma_1, x:\verctype{B}{r}, \Gamma_2 \vdash t@l : A
\end{align*}
We then reapply (\textsc{sub}) to obtain the following:
\begin{center}
    \begin{minipage}{.55\linewidth}
        \infrule[\textsc{sub}]{
            \Gamma_1, x:\verctype{B}{r}, \Gamma_2 \vdash t@l : A
            \andalso
            r \sqsubseteq s
            \andalso
            \vdash s
        }{
            \Gamma_1, x:\verctype{B}{s}, \Gamma_2 \vdash t@l : A
        }
    \end{minipage}
\end{center}
Thus, we obtain the conclusion of the lemma.

\end{itemize}
\end{proof}

\begin{lemma}[Type-safe extraction for versioned values]
\label{lemma:extraction}
\begin{align*}
    [\Gamma] \vdash u : \vertype{r}{A}
    \hspace{1em}\Longrightarrow\hspace{1em}
    \forall l_k\in r.\ \exists t'.
    \left\{
    \begin{aligned}
        &u.l_k   \longrightarrow     t' & & \ (progress)\\
        &[\Gamma] \vdash t' : A &  & \ (preservation)
    \end{aligned}
    \right.
\end{align*}
\end{lemma}

\begin{proof}
By inversion lemma (\ref{lemma:typedvalue}), $u$ has either a form $\pr{t''}$ or $\nvval{\overline{l_i=t_i}}$.
% weak, der, subの有限回の適用の後に必ずprがある。(|Γ|が有限とすると、リソースが0 or 1になってる変数も有限個。subもリソース集合が有界なので、真に小さなリソースに対してのみ適用すれば必ず有限回の適用ですべてのリソースが0になる。)
% weak, der, subの有限回の適用の後に必ずverがある。

\begin{itemize}
\item Suppose $u=\pr{t''}$.\\
We can apply (\textsc{E-ex1}) as follows:
\begin{center}
    \begin{minipage}{.40\linewidth}
        \infrule[E-ex1]{
            \\
        }{
            \pr{t''}.l_k \leadsto t''@l_k
        }
    \end{minipage}
\end{center}
Also, we get the following derivation for $v$.
\begin{center}
\begin{prooftree}
\AxiomC{$ [\Gamma'] \vdash t'' : A$}
\AxiomC{$ \vdash r$}
\RightLabel{(\textsc{pr})}
\BinaryInfC{$ r\cdot [\Gamma'] \vdash \pr{t''} : \vertype{r}{A}$}
\RightLabel{(\textsc{weak}) or (\textsc{sub})}
\UnaryInfC{$\vdots$}
\RightLabel{(\textsc{weak}) or (\textsc{sub})}
\UnaryInfC{$ [\Gamma] \vdash \pr{t''} : \vertype{r}{A}$}
\end{prooftree}
\end{center}
By Lemma \ref{lemma:overwriting}, we know the following:
\begin{align*}
[\Gamma'] \vdash t''@l_k : A
\end{align*}
Finally, we can rearrange the typing context as follows:
\begin{center}
\begin{prooftree}
\AxiomC{$ [\Gamma'] \vdash t''@l_k : A$}
\RightLabel{(\textsc{weak}) or (\textsc{sub})}
\UnaryInfC{$\vdots$}
\RightLabel{(\textsc{weak}) or (\textsc{sub})}
\UnaryInfC{$ [\Gamma] \vdash t''@l_k : A$}
\end{prooftree}
\end{center}
Here, we follow the same manner as for the derivation of $\pr{t''}$ (which may use (\textsc{weak}) and (\textsc{sub})) to get $[\Gamma]$ from $r\cdot [\Gamma']$.

Thus, we obtain the conclusion of the lemma.

\item Suppose $u=\nvval{\overline{l_i=t_i}}$.\\
We can apply (\textsc{E-ex2}) as follows:
\begin{center}
    \begin{minipage}{.50\linewidth}
        \infrule[E-ex2]{
            \\
        }{
            \nvval{\overline{l_i=t_i}}.l_k \leadsto t_k@l_k
        }
    \end{minipage}
\end{center}
Also, we get the following derivation for $v$.
\begin{center}
\begin{prooftree}
\AxiomC{$ [\Gamma'_i] \vdash t_i : A$}
\AxiomC{$\vdash \{\overline{l_i}\}$}
\RightLabel{(\textsc{ver})}
\BinaryInfC{$ \bigcup_i(\{l_i\}\cdot[\Gamma'_i]) \vdash \nvval{\overline{l_i=t_i}} : \vertype{\{\overline{l_i}\}}{A}$}
\RightLabel{(\textsc{weak}) or (\textsc{sub})}
\UnaryInfC{$\vdots$}
\RightLabel{(\textsc{weak}) or (\textsc{sub})}
\UnaryInfC{$ [\Gamma] \vdash \nvval{\overline{l_i=t_i}} : \vertype{\{\overline{l_i}\}}{A}$}
\end{prooftree}
\end{center}
By Lemma \ref{lemma:overwriting}, we know the following:
\begin{align*}
[\Gamma'_k] \vdash t_k@l_k : A
\end{align*}
Finally, we can rearrange the typing context as follows:
\begin{center}
\begin{prooftree}
\AxiomC{$ [\Gamma'_k] \vdash t_k@l_k : A$}
\AxiomC{$ r_{kj} \sqsubseteq r_{kj}\otimes \{l_k\}$}
\RightLabel{(\textsc{sub}) $*\,|\Gamma_k'|$}
\BinaryInfC{$ \underbrace{\{l_k\}\cdot[\Gamma'_k] \vdash t_k@l_k : A}_{P}$}
\end{prooftree}
\end{center}

\begin{center}
\begin{prooftree}
\AxiomC{$P$}
\AxiomC{$ r_{kj}\otimes \{l_k\} \sqsubseteq \textstyle{\sum_{i}}(r_{ij}\otimes \{l_i\})$}
\RightLabel{(\textsc{sub}) $*\,|\Gamma_k'|$}
\BinaryInfC{$ \textstyle{\bigcup_i}(\{l_i\}\cdot[\Gamma'_i]) \vdash t_k@l_k : A$}
\RightLabel{(\textsc{weak}) or (\textsc{sub})}
\UnaryInfC{$\vdots$}
\RightLabel{(\textsc{weak}) or (\textsc{sub})}
\UnaryInfC{$ [\Gamma] \vdash t_k@l_k : A$}
\end{prooftree}
\end{center}
Here in the multiple application of (\textsc{sub}), the second premise compares the resources of j-th versioned assumption between the first premise and conclusion.
Also, we follow the same manner as for the derivation of $\nvval{\overline{l_i=t_i}}$ (which may use (\textsc{weak}) and (\textsc{sub})) to get $[\Gamma]$ from $\textstyle{\bigcup_i}(\{l_i\}\cdot[\Gamma'_i])$.

Thus, we obtain the conclusion of the lemma.

\end{itemize}
\end{proof}

%%%%%%%%%%%%%%%%%%%%%%%%%%%%%%%%%%%%%%%%%%%%
%%%%%%%%%%%%%%%%%%%%%%%%%%%%%%%%%%%%%%%%%%%%
%%%%%%%%%%%%%%%%%%%%%%%%%%%%%%%%%%%%%%%%%%%%
%%%%%%%%%%%%%%%%%%%%%%%%%%%%%%%%%%%%%%%%%%%%
%%%%%%%%%%%%%%%%%%%%%%%%%%%%%%%%%%%%%%%%%%%%
%%%%%%%%%%%%%%%%%%%%%%%%%%%%%%%%%%%%%%%%%%%%

\begin{theorem}[Type preservation for reductions]
\label{lemma:preservationreduction}
\begin{align*}
    \left.
    \begin{aligned}
        &\Gamma \vdash t : A\\
        &t \leadsto t'
    \end{aligned}
    \right\}
    \hspace{1em}\Longrightarrow\hspace{1em}
    \Gamma \vdash t' : A
\end{align*}
\end{theorem}

\begin{proof}
The proof is given by induction on the typing derivation of $t$.
Consider the cases for the last rule used in the typing derivation of the first assumption.
\\

\begin{itemize}
\item Case (\textsc{int})
\begin{center}
    \begin{minipage}{.25\linewidth}
        \infrule[int]{
            \\
        }{
            \emptyset \vdash n:\textsf{Int}
        }
    \end{minipage}
\end{center}
This case holds trivially because there are no reduction rules that can be applied to $n$.
\\

\item Case (\textsc{var})
\begin{center}
    \begin{minipage}{.35\linewidth}
        \infrule[var]{
            \vdash A
        }{
            x:A \vdash x:A
        }
    \end{minipage}
\end{center}
This case holds trivially because there are no reduction rules that can be applied to $x$.
\\

\item Case (\textsc{abs})
\begin{center}
    \begin{minipage}{.65\linewidth}
        \infrule[abs]{
            - \vdash p : B_1 \rhd \Delta'
            \andalso
            \Gamma, \Delta' \,\vdash\, t_1 : A_2%\theta B
        }{
            \Gamma \,\vdash\, \lam{p}{t_1} : \ftype{A_1}{A_2}
        }
    \end{minipage}
\end{center}
This case holds trivially because there are no reduction rules that can be applied to $\lam{p}{t_1}$.
\\

\item Case (\textsc{app})
\begin{center}
    \begin{minipage}{.65\linewidth}
        \infrule[app]{
            \Gamma_1 \vdash t_1 : \ftype{B}{A}
            \andalso
            \Gamma_2 \vdash t_2 : B
        }{
            \Gamma_1 + \Gamma_2 \vdash \app{t_1}{t_2} : A
        }
    \end{minipage}
\end{center}

We perform case analysis for the ruduction rule applied last.
\begin{itemize}
\item Case (\textsc{E-abs1})
\begin{center}
        \begin{minipage}{.5\linewidth}
            \infrule[E-abs1]{
                \\
            }{
                \underbrace{\app{(\lam{x}{t'_1})}{t_2}}_{t} \leadsto \app{(t_2 \rhd x)}{t'_1}
            }
        \end{minipage}
\end{center}
where $t_1=\lam{x}{t'_1}$ for a term $t_1'$.
Then we can apply ($\rhd_{\textrm{var}}$) to obtain the following:
\begin{center}
    \begin{minipage}{.45\linewidth}
        \infrule[$\rhd_{\textrm{var}}$]{
            \\
        }{
            \app{(t_2 \rhd x)}{t'_1} = [t_2 / x]t'_1
        }
    \end{minipage}
\end{center}
In this case, we know the typing derivation of $t$ has the following form:
\begin{center}
\begin{prooftree}
\AxiomC{$ \Gamma'_1, x:B \vdash t'_1 : A$}
\RightLabel{(\textsc{abs})}
\UnaryInfC{$ \Gamma'_1 \vdash \lam{x}{t'_1} : \ftype{B}{A} $}
\RightLabel{(\textsc{weak}), (\textsc{der}), or (\textsc{sub})}
\UnaryInfC{$ \vdots $}
\RightLabel{(\textsc{weak}), (\textsc{der}), or (\textsc{sub})}
\UnaryInfC{$\Gamma_1 \vdash \lam{x}{t'_1} : \ftype{B}{A}$}
\AxiomC{$ \Gamma_2 \vdash t_2 : B$}
\RightLabel{(\textsc{app})}
\BinaryInfC{$ \Gamma_1 + \Gamma_2 \vdash \app{(\lam{x}{t'_1})}{t_2} : A$}
\end{prooftree}
\end{center}
By Lemma \ref{lemma:substitution1}, we know the following:
\begin{align*}
    \left.
    \begin{aligned}
          \Gamma_2 & \vdash t_2:B \\
          \Gamma'_1, x:B & \vdash t'_1:A
    \end{aligned}
    \right\}
    \hspace{1em}\Longrightarrow\hspace{1em}
    \Gamma'_1 + \Gamma_2 \vdash [t_2/x]t'_1:A
\end{align*}
Finally, we can rearrange the typing context as follows:
\begin{center}
\begin{prooftree}
\AxiomC{$ \Gamma'_1 + \Gamma_2 \vdash [t_2/x]t'_1:A $}
\RightLabel{(\textsc{weak}), (\textsc{der}), or (\textsc{sub})}
\UnaryInfC{$ \vdots $}
\RightLabel{(\textsc{weak}), (\textsc{der}), or (\textsc{sub})}
\UnaryInfC{$\Gamma_1 + \Gamma_2 \vdash [t_2/x]t'_1:A$}
\end{prooftree}
\end{center}
Here, there exists a derive tree to get $\Gamma_1+\Gamma_2$ from $\Gamma'_1+\Gamma_2$ as for the derivation of $\lam{x}{t'_1}$ which may use (\textsc{weak}), (\textsc{der}) and (\textsc{sub}).

By choosing $t'=[t_2/x]t'_1$, we obtain the conclusion of the theorem.\\

\item Case (\textsc{E-abs2})
\begin{center}
        \begin{minipage}{.75\linewidth}
            \infrule[E-abs2]{
                \\
            }{
                \underbrace{\app{(\lam{\pr{x}}{t'_1})}{t_2}}_{t} \leadsto \underbrace{\clet{x}{t_2}{t'_1}}_{t'}
            }
        \end{minipage}
\end{center}
In this case, we know the typing derivation of $t$ has the following form:
\begin{center}
\begin{prooftree}
\AxiomC{$ $}
\RightLabel{(\mbox{[}\textsc{pVar}\mbox{]})}
\UnaryInfC{$ r \vdash x : B \rhd x:\verctype{B}{r}$}
\RightLabel{(\textsc{p}$_\square$)}
\UnaryInfC{$ - \vdash \pr{x} : \vertype{r}{B} \rhd x:\verctype{B}{r}$}
\AxiomC{$ \Gamma_1, x:\verctype{B}{r} \vdash t_1' : A$}
\RightLabel{(\textsc{abs})}
\BinaryInfC{$ \underbrace{ \Gamma_1 \vdash  \lam{\pr{x}}{t'_1} : \ftype{\vertype{r}{B}}{A} }_{P}$}
\end{prooftree}
\begin{prooftree}
\AxiomC{$P$}
\AxiomC{$ \Gamma_2 \vdash t_2 : \vertype{r}{B}$}
\RightLabel{(\textsc{app})}
\BinaryInfC{$ \Gamma_1+\Gamma_2 \vdash \app{(\lam{\pr{x}}{t'_1})}{t_2} : A$}
\end{prooftree}
\end{center}
Therefore, we can construct the derivation tree for $t'$ as follows.
\begin{center}
\begin{prooftree}
\AxiomC{$ \Gamma_2 \vdash t_2 : \vertype{r}{B}$}
\AxiomC{$ \Gamma_1, x:\verctype{B}{r} \vdash t_1' : A$}
\RightLabel{(\textsc{app})}
\BinaryInfC{$ \Gamma_1+\Gamma_2 \vdash \clet{x}{t_2}{t'_1} : A$}
\end{prooftree}
\end{center}
Hence, we have the conclusion of the theorem.\\
\end{itemize}

\item Case (\textsc{let})
\begin{center}
    \begin{minipage}{.70\linewidth}
        \infrule[let]{
            \Gamma_1 \,\vdash\, t_1 : \vertype{r}{A}
            \andalso
            \Gamma_2, x:\verctype{A}{r} \,\vdash\, t_2 : B
        }{
            \Gamma_1 + \Gamma_2 \,\vdash\, \clet{x}{t_1}{t_2} : B
        }
    \end{minipage}
\end{center}
The only reduction rule we can apply is (\textsc{E-clet}) with two substitution rules, depending on whether $t_1$ has the form $[t'_1]$ or $\nvval{\overline{l_i=t''_i}}$.
\begin{itemize}
\item Suppose $t_1=[t_1']$.\\
We can apply (\textsc{E-clet}) to obtain the following.
\begin{center}
    \begin{minipage}{.65\linewidth}
        \infrule[E-clet]{
            \\
        }{
            \underbrace{\clet{x}{[t_1']}{t_2}}_{t} \leadsto ([t_1'] \rhd \pr{x})t_2
        }
    \end{minipage}
\end{center}
Thus, we can apply (\textsc{$\rhd_\square$}) and (\textsc{$\rhd_{\textnormal{var}}$}) to obtain the following.
\begin{center}
\begin{prooftree}
\AxiomC{$ $}
\RightLabel{($\rhd_{\textnormal{var}}$)}
\UnaryInfC{$ (t_1' \rhd x)t_2 = [t_1' / x] t_2$}
\RightLabel{($\rhd_{\square}$)}
\UnaryInfC{$ ([t_1'] \rhd \pr{x})t_2 = [t_1' / x] t_2$}
\end{prooftree}
\end{center}
In this case, we know the typing derivation of $t$ has the following form:
\begin{center}
\begin{prooftree}
\AxiomC{$ [\Gamma_1'] \vdash t'_1 : A$}
\AxiomC{$ \vdash r$}
\RightLabel{(\textsc{pr})}
\BinaryInfC{$ r\cdot [\Gamma_1'] \,\vdash\, [t_1'] : \vertype{r}{A} $}
\RightLabel{(\textsc{weak}) or (\textsc{sub})}
\UnaryInfC{$ \vdots $}
\RightLabel{(\textsc{weak}) or (\textsc{sub})}
\UnaryInfC{$ \Gamma_1 \,\vdash\, [t_1'] : \vertype{r}{A}$}
\AxiomC{$ \Gamma_2, x:\verctype{A}{r} \,\vdash\, t_2 : B $}
\RightLabel{(\textsc{let})}
\BinaryInfC{$ \Gamma_1+\Gamma_2 \,\vdash\, \clet{x}{[t_1']}{t_2} : B $}
\end{prooftree}
\end{center}
By Lemma \ref{lemma:substitution2}, we know the following:
\begin{align*}
    \left.
    \begin{aligned}\relax
          [\Gamma_1'] &\,\vdash\, t'_1 : A\\
          \Gamma_2, x:\verctype{A}{r} &\,\vdash\, t_2 : B
    \end{aligned}
    \right\}
    \hspace{1em}\Longrightarrow\hspace{1em}
    \Gamma_2 + r\cdot[\Gamma_1'] \,\vdash\, [t'_1/x]t_2 : B
\end{align*}
Finally, we can rearrange the typing context as follows:
\begin{center}
\begin{prooftree}
\AxiomC{$ \Gamma_2 + r\cdot[\Gamma_1'] \,\vdash\, [t'_1/x]t_2 : B $}
\RightLabel{(\textsc{weak}) or (\textsc{sub})}
\UnaryInfC{$ \vdots $}
\RightLabel{(\textsc{weak}) or (\textsc{sub})}
\UnaryInfC{$ \Gamma_2+\Gamma_1 \,\vdash\, [t'_1/x]t_2 : B$}
\end{prooftree}
\end{center}
Here, there exists a derive tree to get $\Gamma_2+\Gamma_1$ from $\Gamma_2+r\cdot [\Gamma_1']$ as for the derivation of $\pr{t'_1}$ which may use (\textsc{weak}) and (\textsc{sub}).

Thus, by choosing $t' = [t'_1/x]t_2$, we obtain the conclusion of the theorem.

\item Suppose $t_1 = \nvval{\overline{l_i=t''_i}}$.\\
We can apply (\textsc{E-clet}) to obtain the following:
\begin{center}
        \begin{minipage}{.95\linewidth}
            \infrule[E-clet]{
                \\
            }{
                \underbrace{\clet{x}{\nvval{\overline{l_i=t''_i}}}{t_2}}_{t} \leadsto (\nvval{\overline{l_i=t''_i}} \rhd \pr{x})t_2
            }
        \end{minipage}
\end{center}
Thus, we can apply (\textsc{$\rhd_\textnormal{ver}$}) and (\textsc{$\rhd_{\textnormal{var}}$}) to obtain the following.
\begin{center}
\begin{prooftree}
\AxiomC{$ $}
\RightLabel{($\rhd_{\textnormal{var}}$)}
\UnaryInfC{$ (\ivval{\overline{l_i=t''_i}}{l_k} \rhd x)t_2 = [\ivval{\overline{l_i=t''_i}}{l_k} / x] t_2$}
\RightLabel{($\rhd_\textnormal{ver}$)}
\UnaryInfC{$ (\nvval{\overline{l_i=t''_i}} \rhd \pr{x})t_2 = [\ivval{\overline{l_i=t''_i}}{l_k} / x] t_2$}
\end{prooftree}
\end{center}
In this case, we know the typing derivation of $t$ has the following form:
\begin{center}
\begin{prooftree}
\AxiomC{$ [\Gamma_i'] \vdash t''_i : A$}
\AxiomC{$ \vdash \{\overline{l_i}\}$}
\RightLabel{(\textsc{ver})}
\BinaryInfC{$ \bigcup_i(\{l_i\}\cdot [\Gamma'_i]) \,\vdash\, \nvval{\overline{l_i=t''_i}} : \vertype{\{\overline{l_i}\}}{A}$}
\RightLabel{(\textsc{weak}) or (\textsc{sub})}
\UnaryInfC{$ \vdots $}
\RightLabel{(\textsc{weak}) or (\textsc{sub})}
\UnaryInfC{$ \underbrace{\Gamma_1 \,\vdash\, \nvval{\overline{l_i=t''_i}} : \vertype{\{\overline{l_i}\}}{A}}_{P}$}
\end{prooftree}
\end{center}
\begin{center}
\begin{prooftree}
\AxiomC{$ P$}
\AxiomC{$ \Gamma_2, x:A \,\vdash\, t_2 : B $}
\RightLabel{(\textsc{der})}
\UnaryInfC{$ \Gamma_2, x:\verctype{A}{1} \,\vdash\, t_2 : B $}
\RightLabel{(\textsc{sub})$*|\{\overline{l_i}\}|$}
\UnaryInfC{$ \Gamma_2, x:\verctype{A}{\{\overline{l_i}\}} \,\vdash\, t_2 : B $}
\RightLabel{(\textsc{let})}
\BinaryInfC{$ \Gamma_1+\Gamma_2 \,\vdash\, \clet{x}{\nvval{\overline{l_i=t''_i}}}{t_2} : B $}
\end{prooftree}
\end{center}
Then we can derive the type of $\ivval{\overline{l_i=t''_i}}{l_k}$ as follows:
\begin{center}
    \begin{minipage}{.6\linewidth}
        \infrule[veri]{
            [\Gamma'_i] \,\vdash\, t''_i : A
        }{
            \bigcup_i(\{l_i\}\cdot [\Gamma'_i]) \,\vdash\, \ivval{\overline{l_i=t''_i}}{l_k} : A
        }
    \end{minipage}
\end{center}
By Lemma \ref{lemma:substitution1}, we know the following:
\begin{align*}
    \left.
    \begin{aligned}
          \textstyle{\bigcup_i(\{l_i\}\cdot [\Gamma'_i])} &\vdash \ivval{\overline{l_i=t''_i}}{l_k} : A\\
          \Gamma_2, x:A &\vdash t_2 : B
    \end{aligned}
    \right\}
    \hspace{.1em}\Longrightarrow\hspace{.1em}
    \left.
    \begin{aligned}
        \Gamma_2 + &\textstyle{\bigcup_i(\{l_i\}\cdot [\Gamma'_i])}\\
        &\,\vdash\, [\ivval{\overline{l_i=t''_i}}{l_k}/x] t_2 : B
    \end{aligned}
    \right.
\end{align*}
Finally, we can rearrange the typing context as follows:
\begin{center}
\begin{prooftree}
\AxiomC{$ \Gamma_2 + \textstyle{\bigcup_i(\{l_i\}\cdot [\Gamma'_i])} \,\vdash\, [\ivval{\overline{l_i=t''_i}}{l_k}/x] t_2 : B $}
\RightLabel{(\textsc{weak}) or (\textsc{sub})}
\UnaryInfC{$ \vdots $}
\RightLabel{(\textsc{weak}) or (\textsc{sub})}
\UnaryInfC{$ \Gamma_2 + \Gamma_1 \,\vdash\, [\ivval{\overline{l_i=t''_i}}{l_k}/x] t_2 : B $}
\end{prooftree}
\end{center}
Here, there exists a derive tree to get $\Gamma_2+\Gamma_1$ from $\Gamma_2 + \textstyle{\bigcup_i(\{l_i\}\cdot [\Gamma'_i])}$ as for the derivation of $\nvval{\overline{l_i=t''_i}}$ which may use (\textsc{weak}) and (\textsc{sub}).

Thus, by choosing $t' = [\ivval{\overline{l_i=t''_i}}{l_k}/x]t_2$, we obtain the conclusion of the theorem.\\
\end{itemize}

\item Case (\textsc{weak})
\begin{center}
    \begin{minipage}{.45\linewidth}
        \infrule[weak]{
            \Gamma_1 \vdash t : A
            \andalso
            \vdash \Delta'
        }{
            \Gamma_1 + \verctype{\Delta'}{0} \vdash t : A
        }
    \end{minipage}
\end{center}
In this case, $t$ does not change between before and after the last derivation.
The induction hypothesis implies that there exists a term $t''$ such that:
\begin{align*}
        t\leadsto t''
        \ \land\ \Gamma_1 \vdash t'' : A \tag{ih}
\end{align*}
We then reapply (\textsc{weak}) to obtain the following:
\begin{center}
    \begin{minipage}{.38\linewidth}
        \infrule[weak]{
            \Gamma_1 \vdash t'' : A
            \andalso
            \vdash \Delta'
        }{
            \Gamma_1 + \verctype{\Delta'}{0} \vdash t'' : A
        }
    \end{minipage}
\end{center}
Thus, by choosing $t'=t''$, we obtain the conclusion of the theorem.
\\

\item Case (\textsc{der})
\begin{center}
    \begin{minipage}{.45\linewidth}
        \infrule[der]{
            \Gamma_1, x:B \vdash t : A
        }{
            \Gamma_1, x:\verctype{B}{1} \vdash t : A
        }
    \end{minipage}
\end{center}
In this case, $t$ does not change between before and after the last derivation.
The induction hypothesis implies that there exists a term $t''$ such that:
\begin{align*}
        t\leadsto t''
        \ \land\ \Gamma_1, x:B \vdash t'' : A \tag{ih}
\end{align*}
We then reapply (\textsc{der}) to obtain the following:
\begin{center}
    \begin{minipage}{.38\linewidth}
        \infrule[der]{
            \Gamma_1, x:B \vdash t'' : A
        }{
            \Gamma_1, x:\verctype{B}{1} \vdash t'' : A
        }
    \end{minipage}
\end{center}
Thus, by choosing $t'=t''$, we obtain the conclusion of the theorem.
\\

\item Case (\textsc{pr})
\begin{center}
    \begin{minipage}{.35\linewidth}
        \infrule[pr]{
            \verctype{\Gamma}{} \vdash t'' : B
            \andalso
            \vdash r
        }{
            r\cdot\verctype{\Gamma}{} \vdash \pr{t''} : \vertype{r}{B} 
        }
    \end{minipage}
\end{center}
This case holds trivially because there are no reduction rules that can be applied to $\pr{t''}$.
\\

\item Case (\textsc{ver})
\begin{center}
    \begin{minipage}{.65\linewidth}
        \infrule[ver]{
            \verctype{\Gamma_i}{} \vdash t_i : A
            \andalso
            \vdash \{\overline{l_i}\}
        }{
            \bigcup_i(\{l_i\}\cdot [\Gamma_i]) \vdash \nvval{\overline{l_i=t_i}} : \vertype{\{\overline{l_i}\}}{A}
        }
    \end{minipage}
\end{center}
This case holds trivially because there are no reduction rules that can be applied to $\nvval{\overline{l_i=t_i}}$.
\\

\item Case (\textsc{veri})
\begin{center}
    \begin{minipage}{.55\linewidth}
        \infrule[veri]{
            \verctype{\Gamma_i}{} \vdash t_i : A
            \andalso
            \vdash \{\overline{l_i}\}
            \andalso
            l_k \in \{\overline{l_i}\}
        }{
            \bigcup_i(\{l_i\}\cdot [\Gamma_i]) \vdash \ivval{\overline{l_i=t_i}}{l_k} : A
        }
    \end{minipage}
\end{center}
In this case, the only reduction rule we can apply is (\textsc{E-veri}).
\begin{center}
        \begin{minipage}{.50\linewidth}
            \infrule[E-veri]{
                \\
            }{
                \underbrace{\ivval{\overline{l_i=t_i}}{l_k}}_{t} \leadsto t_k@l_k
            }
        \end{minipage}
\end{center}
By Lemma \ref{lemma:overwriting}, we obtain the following:
\begin{align*}
    [\Gamma_k] \vdash t_k : A
    \hspace{1em}\Longrightarrow\hspace{1em}
    [\Gamma_k] \vdash t_k@l_k : A
\end{align*}
Finally, we can rearrange the typing context as follows:
\begin{center}
\begin{prooftree}
\AxiomC{$ [\Gamma_k] \vdash t_k@l_k : A $}
\RightLabel{(\textsc{weak}), (\textsc{der}) or (\textsc{sub})}
\UnaryInfC{$ \vdots $}
\RightLabel{(\textsc{weak}), (\textsc{der}) or (\textsc{sub})}
\UnaryInfC{$ \bigcup_i(\{l_i\}\cdot [\Gamma_i]) \vdash t_k@l_k : A $}
\end{prooftree}
\end{center}
Thus, by choosing $t' = t_k@l_k$, we obtain the conclusion of the theorem.
\\

\item Case (\textsc{extr})
\begin{center}
    \begin{minipage}{.50\linewidth}
        \infrule[extr]{
            \Gamma \vdash t_1 : \vertype{r}{A}
            \andalso
            l_k \in r
        }{
            \Gamma \vdash t_1.l_k : A
        }
    \end{minipage}
\end{center}
In this case, there are two reduction rules that we can apply to $t$, dependenig on whether $t_1$ has the form $[t'_1]$ or $\nvval{\overline{l_i=t''_i}}$.
\begin{itemize}
\item Suppose $t_1 = \pr{t'_1}$.\\
We know the typing derivation of $t$ has the following form:
\begin{prooftree}
    \AxiomC{$ [\Gamma'] \vdash t'_1 : A$}
    \AxiomC{$ \vdash r$}
    \RightLabel{(\textsc{pr})}
    \BinaryInfC{$ r\cdot[\Gamma'] \vdash \pr{t'_1} : \vertype{r}{A}$}
    \RightLabel{(\textsc{weak}) or (\textsc{sub})}
    \UnaryInfC{$ \vdots $}
    \RightLabel{(\textsc{weak}) or (\textsc{sub})}
    \UnaryInfC{$ \Gamma \vdash \pr{t'_1} : \vertype{r}{A}$}
    \AxiomC{$l_k \in r$}
    \RightLabel{(\textsc{extr})}
    \BinaryInfC{$ \Gamma \vdash \pr{t'_1}.l_k : A$}
\end{prooftree}
By Lemma \ref{lemma:extraction}, we know the following:
\begin{align*}
    r\cdot[\Gamma'] \vdash \pr{t'_1} : \vertype{r}{A}
    \hspace{1em}\Longrightarrow\hspace{1em}
    \exists t'.
    \left\{
    \begin{aligned}
        &[t'_1].l_k   \longrightarrow      t' \\
        &r\cdot[\Gamma'] \vdash t' : A
    \end{aligned}
    \right.
\end{align*}
Finally, we can rearrange the typing context as follows:
\begin{center}
\begin{prooftree}
\AxiomC{$ r\cdot[\Gamma'] \vdash t' : A$}
\RightLabel{(\textsc{weak}) or (\textsc{sub})}
\UnaryInfC{$ \vdots $}
\RightLabel{(\textsc{weak}) or (\textsc{sub})}
\UnaryInfC{$ \Gamma \vdash t' : A$}
\end{prooftree}
\end{center}
Here, we follow the same manner as for the derivation of $\pr{t'_1}$ (which may use (\textsc{weak}) and (\textsc{sub})) to get $\Gamma$ from $r\cdot[\Gamma']$.

Thus, we obtain the conclusion of the theorem.

\item Suppose $t_1 = \nvval{\overline{l_i=t_i}}$.\\
The last derivation is rewritten as follows:
\begin{prooftree}
    \AxiomC{$[\Gamma'_i] \vdash t_i : A$}
    \AxiomC{$\vdash \{\overline{l_i}\}$}
    \RightLabel{(\textsc{ver})}
    \BinaryInfC{$ \bigcup_i\{l_i\}\cdot [\Gamma'_i] \vdash \nvval{\overline{l_i=t_i}} : \vertype{\{\overline{l_i}\}}{A}$}
    \RightLabel{(\textsc{weak}) or (\textsc{sub})}
    \UnaryInfC{$ \vdots $}
    \RightLabel{(\textsc{weak}) or (\textsc{sub})}
    \UnaryInfC{$ \Gamma \vdash \nvval{\overline{l_i=t_i}} : \vertype{\{\overline{l_i}\}}{A}$}
    \AxiomC{$l_k \in \{\overline{l_i}\}$}
    \RightLabel{(\textsc{extr})}
    \BinaryInfC{$ \Gamma \vdash \nvval{\overline{l_i=t_i}}.l_k : A$}
\end{prooftree}
By Lemma \ref{lemma:extraction}, we know the following:
\begin{align*}
    \textstyle{\bigcup_i}\{l_i\}\cdot [\Gamma'_i] \vdash \nvval{\overline{l_i=t_i}} : \vertype{\{\overline{l_i}\}}{A}
    \hspace{1em}\Longrightarrow\hspace{1em}
    \exists t'.
    \left\{
    \begin{aligned}
        &\nvval{\overline{l_i=t_i}}.l_k   \longrightarrow  t' \\
        &\textstyle{\bigcup_i}\{l_i\}\cdot [\Gamma'_i] \vdash t' : A
    \end{aligned}
    \right.
\end{align*}
Finally, we can rearrange the typing context as follows:
\begin{center}
\begin{prooftree}
\AxiomC{$ \textstyle{\bigcup_i}\{l_i\}\cdot [\Gamma'_i] \vdash t' : A$}
\RightLabel{(\textsc{weak}) or (\textsc{sub})}
\UnaryInfC{$ \vdots $}
\RightLabel{(\textsc{weak}) or (\textsc{sub})}
\UnaryInfC{$ \Gamma \vdash t' : A$}
\end{prooftree}
\end{center}
Here, we follow the same manner as for the derivation of $\nvval{\overline{l_i=t_i}}$ (which may use (\textsc{weak}) and (\textsc{sub})) to get $\Gamma$ from $\bigcup_i\{l_i\}\cdot [\Gamma'_i]$.

Thus, we obtain the conclusion of the theorem.\\
\end{itemize}

\item Case (\textsc{sub})
\begin{center}
    \begin{minipage}{.65\linewidth}
        \infrule[\textsc{sub}]{
            \Gamma_1, x:\verctype{B}{r}, \Gamma_2 \vdash t : A
            \andalso
            r \sqsubseteq s
            \andalso
            \vdash s
        }{
            \Gamma_1, x:\verctype{B}{s}, \Gamma_2 \vdash t : A
        }
    \end{minipage}
\end{center}
In this case, $t$ does not change between before and after the last derivation.
The induction hypothesis implies that there exists a term $t''$ such that:
\begin{align*}
    t \leadsto t''
    \ \land\ 
    \Gamma_1, x:\verctype{B}{r}, \Gamma_2 \vdash t'' : A \tag{ih}
\end{align*}
We then reapply (\textsc{sub}) to obtain the following:
\begin{prooftree}
    \AxiomC{$\Gamma_1, x:\verctype{B}{r}, \Gamma_2 \vdash t'' : A$}
    \AxiomC{$ r \sqsubseteq s$}
    \AxiomC{$ \vdash s$}
    \RightLabel{(\textsc{sub})}
    \TrinaryInfC{$ \Gamma_1, x:\verctype{B}{s}, \Gamma_2 \vdash t'' : A$}
\end{prooftree}
Thus, by choosing $t' = t''$, we obtain the conclusion of the theorem.

\end{itemize}
\end{proof}

%%%%%%%%%%%%%%%%%%%%%%%%%%%%%%%%%%%%%%%%%%%%
%%%%%%%%%%%%%%%%%%%%%%%%%%%%%%%%%%%%%%%%%%%%
%%%%%%%%%%%%%%%%%%%%%%%%%%%%%%%%%%%%%%%%%%%%
%%%%%%%%%%%%%%%%%%%%%%%%%%%%%%%%%%%%%%%%%%%%
%%%%%%%%%%%%%%%%%%%%%%%%%%%%%%%%%%%%%%%%%%%%
%%%%%%%%%%%%%%%%%%%%%%%%%%%%%%%%%%%%%%%%%%%%

\begin{theorem}[Type preservation for evaluations]
\label{lemma:preservationevaluation}
\begin{align*}
    \left.
    \begin{aligned}
        &\Gamma \vdash t : A\\
        &t \longrightarrow t'
    \end{aligned}
    \right\}
    \hspace{1em}\Longrightarrow\hspace{1em}
    \Gamma \vdash t' : A
\end{align*}
\end{theorem}

\begin{proof}
The proof is given by induction on the typing derivation of $t$.
Consider the cases for the last rule used in the typing derivation of the first assumption.

\begin{itemize}
\item Case (\textsc{int})
\begin{center}
    \begin{minipage}{.25\linewidth}
        \infrule[int]{
            \\
        }{
            \emptyset \vdash n:\textsf{Int}
        }
    \end{minipage}
\end{center}
This case holds trivially because there are no evaluation rules that can be applied to $n$.
\\

\item Case (\textsc{var})
\begin{center}
    \begin{minipage}{.35\linewidth}
        \infrule[var]{
            \vdash A
        }{
            x:A \vdash x:A
        }
    \end{minipage}
\end{center}
This case holds trivially because there are no evaluation rules that can be applied to $x$.
\\

\item Case (\textsc{abs})
\begin{center}
    \begin{minipage}{.6\linewidth}
        \infrule[abs]{
            - \vdash p : B_1 \rhd \Delta'
            \andalso
            \Gamma, \Delta' \,\vdash\, t_1 : A_2%\theta B
        }{
            \Gamma \,\vdash\, \lam{p}{t_1} : \ftype{A_1}{A_2}
        }
    \end{minipage}
\end{center}
This case holds trivially because there are no evaluation rules that can be applied to $\lam{p}{t_1}$.
\\

\item Case (\textsc{app})
\begin{center}
    \begin{minipage}{.65\linewidth}
        \infrule[app]{
            \Gamma_1 \vdash t_1 : \ftype{B}{A}
            \andalso
            \Gamma_2 \vdash t_2 : B
        }{
            \Gamma_1 + \Gamma_2 \vdash \app{t_1}{t_2} : A
        }
    \end{minipage}
\end{center}
In this case, there are two evaluation rules that can be applied to $t$.

\begin{itemize}
\item Suppose the evaluation rule matches to $[\cdot].$\\
We perform the case analysis for the last ruduction rule.

\begin{itemize}
\item Case (\textsc{E-abs1})
We know the evaluation of the assumption has the following form:
\begin{center}
\begin{prooftree}
    \AxiomC{$ $}
    \RightLabel{$\textsc{E-abs1}$}
    \UnaryInfC{ $\app{(\lam{x}{t'_1})}{t_2} \leadsto \app{(t_2 \rhd x)}{t'_1}$}
    \UnaryInfC{$ \underbrace{\app{(\lam{x}{t'_1})}{t_2}}_{t} \longrightarrow \underbrace{\app{(t_2 \rhd x)}{t'_1}}_{t'}$}
\end{prooftree}
\end{center}
By Lemma \ref{lemma:preservationreduction}, we know the following:
\begin{align*}
    \left.
    \begin{aligned}
        &\Gamma_1+\Gamma_2 \vdash \app{(\lam{x}{t'_1})}{t_2} : A\\
        &\app{(\lam{x}{t'_1})}{t_2} \leadsto \app{(t_2 \rhd x)}{t'_1}
    \end{aligned}
    \right\}
    \hspace{1em}\Longrightarrow\hspace{1em}
    \Gamma_1+\Gamma_2 \vdash \app{(t_2 \rhd x)}{t'_1} : A
\end{align*}
Thus, we obtain the conclusion of the theorem.

\item Case (\textsc{E-abs2})
\begin{center}
\begin{prooftree}
    \AxiomC{$ $}
    \RightLabel{$\textsc{E-abs2}$}
    \UnaryInfC{ $\app{(\lam{\pr{x}}{t'_1})}{t_2} \leadsto \clet{x}{t_2}{t_1'}$}
    \UnaryInfC{$ \underbrace{\app{(\lam{\pr{x}}{t'_1})}{t_2}}_{t} \longrightarrow \underbrace{\clet{x}{t_2}{t_1'}}_{t'}$}
\end{prooftree}
\end{center}
In this case, we know the typing derivation of $t$ has the following form:
\begin{center}
\begin{prooftree}
\AxiomC{$ $}
\RightLabel{(\mbox{[}\textsc{pVar}\mbox{]})}
\UnaryInfC{$ r \vdash x : B \rhd x:\verctype{B}{r}$}
\RightLabel{(\textsc{p}$_\square$)}
\UnaryInfC{$ - \vdash \pr{x} : \vertype{r}{B} \rhd x:\verctype{B}{r}$}
\AxiomC{$ \Gamma_1, x:\verctype{B}{r} \vdash t_1' : A$}
\RightLabel{(\textsc{abs})}
\BinaryInfC{$ \underbrace{ \Gamma_1 \vdash  \lam{\pr{x}}{t'_1} : \ftype{\vertype{r}{B}}{A} }_{P}$}
\end{prooftree}
\begin{prooftree}
\AxiomC{$P$}
\AxiomC{$ \Gamma_2 \vdash t_2 : \vertype{r}{B}$}
\RightLabel{(\textsc{app})}
\BinaryInfC{$ \Gamma_1+\Gamma_2 \vdash \app{(\lam{\pr{x}}{t'_1})}{t_2} : A$}
\end{prooftree}
\end{center}
Therefore, we can construct the derivation tree for $t'$ as follows.
\begin{center}
\begin{prooftree}
\AxiomC{$ \Gamma_2 \vdash t_2 : \vertype{r}{B}$}
\AxiomC{$ \Gamma_1, x:\verctype{B}{r} \vdash t_1' : A$}
\RightLabel{(\textsc{app})}
\BinaryInfC{$ \Gamma_1+\Gamma_2 \vdash \clet{x}{t_2}{t'_1} : A$}
\end{prooftree}
\end{center}
Hence, we have the conclusion of the theorem.\\
\end{itemize}

\item Suppose the evaluation rule matches to $\app{E}{t}$.\\
We know the evaluation of the assumption has the following form:
\begin{center}
\begin{prooftree}
    \AxiomC{$ $}
    \UnaryInfC{ $t'_1 \leadsto t''_1$}
    \UnaryInfC{$ \underbrace{\app{E[t'_1]}{t_2}}_{t} \longrightarrow \app{E[t''_1]}{t_2}$}
\end{prooftree}
\end{center}
where $t_1=E[t'_1]$.

By induction hypothesis, we know the following:
\begin{align*}
    \left.
    \begin{aligned}
        &\Gamma_1 \vdash E[t'_1] : \ftype{B}{A}\\
        &E[t'_1] \longrightarrow E[t''_1]
    \end{aligned}
    \right\}
    \hspace{1em}\Longrightarrow\hspace{1em}
    \Gamma_1 \vdash E[t''_1] : \ftype{B}{A}
    \tag{ih}
\end{align*}
We then reapply (\textsc{app}) to obtain the following:
\begin{center}
    \begin{minipage}{.70\linewidth}
        \infrule[app]{
            \Gamma_1 \vdash E[t''_1] : \ftype{B}{A}
            \andalso
            \Gamma_2 \vdash t_2  : B
        }{
            \Gamma_1+\Gamma_2 \vdash \app{E[t''_1]}{t_2} : A
        }
    \end{minipage}
\end{center}
Thus, we obtain the conclusion of the theorem.\\
\end{itemize}

\item Case (\textsc{let})
\begin{center}
    \begin{minipage}{.70\linewidth}
        \infrule[let]{
            \Gamma_1 \,\vdash\, t_1 : \vertype{r}{A}
            \andalso
            \Gamma_2, x:\verctype{A}{r} \,\vdash\, t_2 : B
        }{
            \Gamma_1 + \Gamma_2 \,\vdash\, \clet{x}{t_1}{t_2} : B
        }
    \end{minipage}
\end{center}
In this case, there are two evaluation rules that we can apply to $t$.
\begin{itemize}
\item Suppose the evaluation rule matches to $[\cdot]$.\\
We know the evaluation of the assumption has the following form:
\begin{center}
\begin{prooftree}
    \AxiomC{$ $}
    \RightLabel{(\textsc{E-clet})}
    \UnaryInfC{ $\clet{x}{[t_1']}{t_2} \leadsto ([t_1'] \rhd \pr{x})t_2$}
    \UnaryInfC{$ \underbrace{\clet{x}{[t_1']}{t_2}}_{t} \longrightarrow ([t_1'] \rhd \pr{x})t_2$}
\end{prooftree}
\end{center}
By Lemma \ref{lemma:preservationreduction}, we know the following:
\begin{align*}
    \left.
    \begin{aligned}
        &\Gamma_1+\Gamma_2 \vdash \clet{x}{[t_1']}{t_2} : B\\
        &\clet{x}{[t_1']}{t_2} \leadsto ([t_1'] \rhd \pr{x})t_2
    \end{aligned}
    \right\}
    \hspace{1em}\Longrightarrow\hspace{1em}
    \Gamma_1+\Gamma_2 \vdash ([t_1'] \rhd \pr{x})t_2 : B
\end{align*}
Thus, we obtain the conclusion of the theorem.

\item Suppose the evaluation rule matches to $\clet{x}{E}{t}$.\\
We know the evaluation of the assumption has the following form:
\begin{center}
\begin{prooftree}
    \AxiomC{$ $}
    \UnaryInfC{ $t'_1 \leadsto t''_1$}
    % \UnaryInfC{ $t_1 \longrightarrow t'_1$}
    \UnaryInfC{$ \underbrace{\clet{x}{E[t'_1]}{t_2}}_{t} \longrightarrow \clet{x}{E[t''_1]}{t_2}$}
\end{prooftree}
\end{center}
where $t_1=E[t'_1]$.

By induction hypothesis, we know the following:
\begin{align*}
    \left.
    \begin{aligned}
        &\Gamma_1 \vdash E[t'_1] : \vertype{r}{A}\\
        &E[t'_1] \longrightarrow E[t''_1]
    \end{aligned}
    \right\}
    \hspace{1em}\Longrightarrow\hspace{1em}
    \Gamma_1 \vdash E[t''_1] : \vertype{r}{A}
    \tag{ih}
\end{align*}
We then reapply (\textsc{let}) to obtain the following:
\begin{center}
    \begin{minipage}{.70\linewidth}
        \infrule[let]{
            \Gamma_1 \,\vdash\, E[t''_1] : \vertype{r}{A}
            \andalso
            \Gamma_2, x:\verctype{A}{r} \,\vdash\, t_2 : B
        }{
            \Gamma_1 + \Gamma_2 \,\vdash\, \clet{x}{E[t''_1]}{t_2} : B
        }
    \end{minipage}
\end{center}
Thus, we obtain the conclusion of the theorem.\\
\end{itemize}

\item Case (\textsc{weak})
\begin{center}
    \begin{minipage}{.45\linewidth}
        \infrule[weak]{
            \Gamma_1 \vdash t : A
            \andalso
            \vdash \Delta'
        }{
            \Gamma_1 + \verctype{\Delta'}{0} \vdash t : A
        }
    \end{minipage}
\end{center}
In this case, $t$ does not change between before and after the last derivation.
The induction hypothesis implies that there exists a term $t'$ such that:
\begin{align*}
        t \longrightarrow t'
        \ \land\ \Gamma_1 \vdash t' : A \tag{ih}
\end{align*}
We then reapply (\textsc{weak}) to obtain the following:
\begin{center}
    \begin{minipage}{.38\linewidth}
        \infrule[weak]{
            \Gamma_1 \vdash t' : A
            \andalso
            \vdash \Delta'
        }{
            \Gamma_1 + \verctype{\Delta'}{0} \vdash t' : A
        }
    \end{minipage}
\end{center}
Thus, we obtain the conclusion of the theorem.
\\

\item Case (\textsc{der})
\begin{center}
    \begin{minipage}{.45\linewidth}
        \infrule[der]{
            \Gamma_1, x:B \vdash t : A
        }{
            \Gamma_1, x:\verctype{B}{1} \vdash t : A
        }
    \end{minipage}
\end{center}
In this case, $t$ does not change between before and after the last derivation.
The induction hypothesis implies that there exists a term $t'$ such that:
\begin{align*}
        t \longrightarrow t'
        \ \land\ \Gamma_1, x:B \vdash t' : A \tag{ih}
\end{align*}
We then reapply (\textsc{der}) to obtain the following:
\begin{center}
    \begin{minipage}{.38\linewidth}
        \infrule[der]{
            \Gamma_1, x:B \vdash t' : A
        }{
            \Gamma_1, x:\verctype{B}{1} \vdash t' : A
        }
    \end{minipage}
\end{center}
Thus, we obtain the conclusion of the theorem.
\\

\item Case (\textsc{pr})
\begin{center}
    \begin{minipage}{.40\linewidth}
        \infrule[pr]{
            \verctype{\Gamma}{} \vdash t'' : B
            \andalso
            \vdash r
        }{
            r\cdot\verctype{\Gamma}{} \vdash \pr{t''} : \vertype{r}{B} 
        }
    \end{minipage}
\end{center}
This case holds trivially because there are no evaluation rules that can be applied to $\pr{t''}$.
\\

\item Case (\textsc{ver})
\begin{center}
    \begin{minipage}{.65\linewidth}
        \infrule[ver]{
            \verctype{\Gamma_i}{} \vdash t_i : A
            \andalso
            \vdash \{\overline{l_i}\}
        }{
            \bigcup_i(\{l_i\}\cdot [\Gamma_i]) \vdash \nvval{\overline{l_i=t_i}} : \vertype{\{\overline{l_i}\}}{A}
        }
    \end{minipage}
\end{center}
This case holds trivially because there are no evaluation rules that can be applied to $\nvval{\overline{l_i=t_i}}$.
\\

\item Case (\textsc{veri})
\begin{center}
    \begin{minipage}{.55\linewidth}
        \infrule[veri]{
            \verctype{\Gamma_i}{} \vdash t_i : A
            \andalso
            \vdash \{\overline{l_i}\}
            \andalso
            l_k \in \{\overline{l_i}\}
        }{
            \bigcup_i(\{l_i\}\cdot [\Gamma_i]) \vdash \ivval{\overline{l_i=t_i}}{l_k} : A
        }
    \end{minipage}
\end{center}
In this case, the only evaluation rule we can apply is evaluation for $[\cdot]$.
We know the evaluation of the assumption has the following form:
\begin{center}
\begin{prooftree}
    \AxiomC{$ $}
    \RightLabel{\textsc{E-veri}}
    \UnaryInfC{ $\ivval{\overline{l_i=t_i}}{l_k} \leadsto t_k@l_k$}
    \UnaryInfC{ $\underbrace{\ivval{\overline{l_i=t_i}}{l_k}}_{t} \longrightarrow t_k@l_k$}
\end{prooftree}
\end{center}
By Lemma \ref{lemma:preservationreduction}, we know the following:
\begin{align*}
    \left.
    \begin{aligned}
        &\textstyle{\bigcup_i}(\{l_i\}\cdot [\Gamma_i]) \vdash \ivval{\overline{l_i=t_i}}{l_k} : A\\
        &\ivval{\overline{l_i=t_i}}{l_k} \leadsto t_k@l_k
    \end{aligned}
    \right\}
    \hspace{1em}\Longrightarrow\hspace{1em}
    \textstyle{\bigcup_i}(\{l_i\}\cdot [\Gamma_i]) \vdash t_k@l_k : A
\end{align*}
Thus, we obtain the conclusion of the theorem.
\\

\item Case (\textsc{extr})
\begin{center}
    \begin{minipage}{.50\linewidth}
        \infrule[extr]{
            \Gamma \vdash t_1 : \vertype{r}{A}
            \andalso
            l_k \in r
        }{
            \Gamma \vdash t_1.l_k : A
        }
    \end{minipage}
\end{center}
In this case, there are two evaluation rules that we can apply to $t$.
\begin{itemize}
\item Suppose the evaluation rule matches to $[\cdot]$.\\
We know the evaluation of the assumption has the following form:
\begin{center}
\begin{prooftree}
    \AxiomC{$ $}
    \RightLabel{\textsc{E-ex1} or \textsc{E-ex2}}
    \UnaryInfC{$ t_1.l_k \leadsto t'_1$}
    \UnaryInfC{$ \underbrace{t_1.l_k}_{t} \longrightarrow t'_1$}
\end{prooftree}
\end{center}
By Lemma \ref{lemma:preservationreduction}, we know the following:
\begin{align*}
    \left.
    \begin{aligned}
        &\Gamma \vdash t_1.l_k : A\\
        &t_1.l_k \leadsto t'_1
    \end{aligned}
    \right\}
    \hspace{1em}\Longrightarrow\hspace{1em}
    \Gamma \vdash t'_1 : A
\end{align*}
Thus, we obtain the conclusion of the theorem.

\item Suppose the evaluation rule matches to $E.l$.\\
We know the evaluation of the assumption has the following form:
\begin{center}
\begin{prooftree}
    \AxiomC{$ $}
    \UnaryInfC{$ t'_1 \leadsto t''_1$}
    \UnaryInfC{$ \underbrace{E[t'_1].l_k}_{t} \longrightarrow E[t''_1].l_k$}
\end{prooftree}
\end{center}
where $t_1=E[t'_1]$.

By induction hypothesis, we know the following:
\begin{align*}
    \left.
    \begin{aligned}
        &\Gamma \vdash E[t'_1] : \vertype{r}{A}\\
        &E[t'_1] \longrightarrow E[t''_1]
    \end{aligned}
    \right\}
    \hspace{1em}\Longrightarrow\hspace{1em}
    \Gamma \vdash E[t''_1] : \vertype{r}{A}
    \tag{ih}
\end{align*}
We the reapply (\textsc{extr}) to obtain the following:
\begin{center}
    \begin{minipage}{.55\linewidth}
        \infrule[extr]{
            \Gamma \vdash E[t''_1] : \vertype{r}{A}
            \andalso
            l_k \in r
        }{
            \Gamma \vdash E[t''_1].l_k : A
        }
    \end{minipage}
\end{center}
Thus, we obtain the conclusion of the theorem.\\
\end{itemize}

\item Case (\textsc{sub})
\begin{center}
    \begin{minipage}{.65\linewidth}
        \infrule[\textsc{sub}]{
            \Gamma_1, x:\verctype{B}{r}, \Gamma_2 \vdash t : A
            \andalso
            r \sqsubseteq s
            \andalso
            \vdash s
        }{
            \Gamma_1, x:\verctype{B}{s}, \Gamma_2 \vdash t : A
        }
    \end{minipage}
\end{center}
In this case, $t$ does not change between before and after the last derivation.
The induction hypothesis implies that there exists a term $t'$ such that:
\begin{align*}
    t \longrightarrow t'
    \ \land\ 
    \Gamma_1, x:\verctype{B}{r}, \Gamma_2 \vdash t' : A \tag{ih}
\end{align*}
We then reapply (\textsc{sub}) to obtain the following:
\begin{prooftree}
    \AxiomC{$\Gamma_1, x:\verctype{B}{r}, \Gamma_2 \vdash t' : A$}
    \AxiomC{$ r \sqsubseteq s$}
    \AxiomC{$ \vdash s$}
    \RightLabel{(\textsc{sub})}
    \TrinaryInfC{$ \Gamma_1, x:\verctype{B}{s}, \Gamma_2 \vdash t' : A$}
\end{prooftree}
Thus, we obtain the conclusion of the theorem.

\end{itemize}
\end{proof}

%%%%%%%%%%%%%%%%%%%%%%%%%%%%%%%%%%%%%%%%%%%%
%%%%%%%%%%%%%%%%%%%%%%%%%%%%%%%%%%%%%%%%%%%%
%%%%%%%%%%%%%%%%%%%%%%%%%%%%%%%%%%%%%%%%%%%%
%%%%%%%%%%%%%%%%%%%%%%%%%%%%%%%%%%%%%%%%%%%%
%%%%%%%%%%%%%%%%%%%%%%%%%%%%%%%%%%%%%%%%%%%%
%%%%%%%%%%%%%%%%%%%%%%%%%%%%%%%%%%%%%%%%%%%%

\begin{theorem}[\corelang{} progress]
\label{lemma:progress}
\begin{align*}
    \emptyset \vdash t:A \Longrightarrow
    (\textnormal{\textsf{value}}\ t) \lor (\exists t'. t \longrightarrow t')
\end{align*}
\end{theorem}

\begin{proof}
The proof is given by induction on the typing derivation of $t$.
Consider the cases for the last rule used in the typing derivation of the assumption.
\\

\begin{itemize}
\item Case (\textsc{int})
\begin{center}
    \begin{minipage}{.25\linewidth}
        \infrule[int]{
            \\
        }{
            \emptyset \vdash n:\textsf{Int}
        }
    \end{minipage}
\end{center}
This case holds trivially because \textsf{value} $n$.
\\

\item Case (\textsc{var})
This case holds trivially because $x:A$ cannot be $\emptyset$.
\\

\item Case (\textsc{abs})
\begin{center}
    \begin{minipage}{.55\linewidth}
        \infrule[abs]{
            - \vdash p : A_1 \rhd \Delta'
            \andalso
            \Delta' \,\vdash\, t : A_2%\theta B
        }{
            \emptyset \,\vdash\, \lam{p}{t} : \ftype{A_1}{A_2}
        }
    \end{minipage}
\end{center}
This case holds trivially because \textsf{value} $\lam{p}{t}$.
\\

\item Case (\textsc{app})
\begin{center}
    \begin{minipage}{.65\linewidth}
        \infrule[app]{
            \emptyset \vdash t_1 : \ftype{B}{A}
            \andalso
            \emptyset \vdash t_2 : B
        }{
            \emptyset \vdash \app{t_1}{t_2} : A
        }
    \end{minipage}
\end{center}
There are two cases whether $t_1$ is a value or not.
\begin{itemize}
\item Suppose $t_1$ is a value.\\
By the inversion lemma (\ref{lemma:typedvalue}), we know that there exists a term $t'_1$ and $t_1=\lam{p}{t'_1}$.
Thus, we can apply two rules to $t$ as follows.

\begin{itemize}
\item Case (\textsc{E-abs1})
\begin{center}
\begin{prooftree}
\AxiomC{$ $}
\RightLabel{(\textsc{E-abs1})}
\UnaryInfC{$ \app{(\lam{x}{t'_1})}{t_2} \leadsto \app{(t_2 \rhd x)}{t'_1}$}
\RightLabel{}
\UnaryInfC{$ \underbrace{\app{(\lam{x}{t'_1})}{t_2}}_{t} \longrightarrow \app{(t_2 \rhd x)}{t'_1}$}
\end{prooftree}
\end{center}
Furthermore, we know the following:
\begin{center}
        \begin{minipage}{.50\linewidth}
            \infrule[\ensuremath{\rhd_{\mathrm{var}}}]{
                \\
            }{
                \app{(t_2 \rhd x)}{t'_1} = [t_2/x]t'_1
            }
        \end{minipage}
\end{center}
By choosing $t'=[t_2/x]t'_1$, we obtain the conclusion of the theorem.

\item Case (\textsc{E-abs2})
\begin{center}
\begin{prooftree}
    \AxiomC{$ $}
    \RightLabel{$\textsc{E-abs2}$}
    \UnaryInfC{ $\app{(\lam{\pr{x}}{t'_1})}{t_2} \leadsto \clet{x}{t_2}{t_1'}$}
    \UnaryInfC{$ \underbrace{\app{(\lam{\pr{x}}{t'_1})}{t_2}}_{t} \longrightarrow \underbrace{\clet{x}{t_2}{t_1'}}_{t'}$}
\end{prooftree}
\end{center}
By choosing $t'=\clet{x}{t_2}{t_1'}$, we obtain the conclusion of the theorem.\\
\end{itemize}

\item Suppose $t_1$ is not a value.\\
There exists a term $t'_1$ such that:
\begin{align*}
    \begin{minipage}{.13\linewidth}
        \infrule[]{
            t_1 \leadsto t'_1
        }{
            t_1 \longrightarrow t'_1
        }
    \end{minipage}
\end{align*}
Also, we can apply evaluation for application to $t$.
\begin{center}
    \begin{minipage}{.27\linewidth}
        \infrule[]{
            t_1 \leadsto t'_1
        }{
            \underbrace{\app{t_1}{t_2}}_{t} \longrightarrow \app{t'_1}{t_2}
        }
    \end{minipage}
\end{center}
Thus, by choosing $t'=\app{t'_1}{t_2}$, we obtain the conclusion of the theorem.\\
\end{itemize}

\item Case (\textsc{let})
\begin{center}
    \begin{minipage}{.65\linewidth}
        \infrule[let]{
            \emptyset \,\vdash\, t_1 : \vertype{r}{A}
            \andalso
            x:\verctype{A}{r} \,\vdash\, t_2 : B
        }{
            \emptyset \,\vdash\, \clet{x}{t_1}{t_2} : B
        }
    \end{minipage}
\end{center}
There are two cases whether $t_1$ is a value or not.
\begin{itemize}
\item Suppose $t_1$ is a value.\\
By the inversion lemma (\ref{lemma:typedvalue}), we know that $t_1$ has either a form of $[t'_1]$ or and $\nvval{\overline{l_i=t''_i}}$.
\begin{itemize}
\item Case $t_1=[t_1']$.\\
In this case, we can apply (\textsc{E-clet}) to obtain the following.
\begin{center}
    \begin{minipage}{.8\linewidth}
        \infrule[E-clet]{
            \\
        }{
            \clet{x}{[t_1']}{t_2} \leadsto ([t_1'] \rhd \pr{x})t_2
        }
    \end{minipage}
\end{center}
Thus, we can apply (\textsc{$\rhd_\square$}) and (\textsc{$\rhd_{\textnormal{var}}$}) to obtain the following.
\begin{center}
\begin{prooftree}
\AxiomC{$ $}
\RightLabel{($\rhd_{\textnormal{var}}$)}
\UnaryInfC{$ (t_1' \rhd x)t_2 = [t_1' / x] t_2$}
\RightLabel{($\rhd_{\square}$)}
\UnaryInfC{$ ([t_1'] \rhd \pr{x})t_2 = [t_1' / x] t_2$}
\end{prooftree}
\end{center}
Thus, by choosing $t' = [t'_1/x]t_2$, we obtain the conclusion of the theorem.

\item Case $t_1 = \nvval{\overline{l_i=t''_i}}$.\\
In this case, we can apply (\textsc{E-clet}) to obtain the following:
\begin{center}
\begin{prooftree}
\AxiomC{$ $}
\RightLabel{(\textsc{E-clet})}
\UnaryInfC{$ \clet{x}{\nvval{\overline{l_i=t''_i}}}{t_2} \leadsto (\ivval{\overline{l_i=t''_i}}{l_k} \rhd \pr{x})t_2$}
\RightLabel{}
\UnaryInfC{$ \underbrace{\clet{x}{\nvval{\overline{l_i=t''_i}}}{t_2}}_{t} \longrightarrow (\ivval{\overline{l_i=t''_i}}{l_k} \rhd \pr{x})t_2 $}
\end{prooftree}
\end{center}
Also, we can apply (\textsc{$\rhd_\textnormal{ver}$}) and (\textsc{$\rhd_{\textnormal{var}}$}) to obtain the following.
\begin{center}
\begin{prooftree}
\AxiomC{$ $}
\RightLabel{($\rhd_{\textnormal{var}}$)}
\UnaryInfC{$ (\ivval{\overline{l_i=t''_i}}{l_k} \rhd x)t_2 = [\ivval{\overline{l_i=t''_i}}{l_k} / x] t_2$}
\RightLabel{($\rhd_\textnormal{ver}$)}
\UnaryInfC{$ (\nvval{\overline{l_i=t''_i}} \rhd \pr{x})t_2 = [\ivval{\overline{l_i=t''_i}}{l_k} / x] t_2$}
\end{prooftree}
\end{center}
Thus, by choosing $t' = [\ivval{\overline{l_i=t''_i}}{l_k}/x]t_2$, we obtain the conclusion of the theorem.
\end{itemize}

\item Suppose $t_1$ is not a value.\\
There exists terms $t'_1$ such that:
\begin{align*}
         \begin{minipage}{.13\linewidth}
            \infrule[]{
                 t_1\leadsto t'_1 \\
            }{
                 t_1\longrightarrow t'_1 
            }
        \end{minipage}
\end{align*}
Also, we can apply evaluation for contextual let bindings to $t$.
\begin{center}
    \begin{minipage}{.80\linewidth}
        \infrule[]{
            t_1 \leadsto t'_1
        }{
            \underbrace{\clet{x}{t_1}{t_2}}_{t} \longrightarrow \clet{x}{t'_1}{t_2}
        }
    \end{minipage}
\end{center}
Thus, by choosing $t'=(\clet{x}{t'_1}{t_2})$, we obtain the conclusion of the theorem.\\
\end{itemize}

\item Case (\textsc{weak})
\begin{center}
    \begin{minipage}{.34\linewidth}
        \infrule[weak]{
            \emptyset \vdash t : A
            \andalso
            \vdash \emptyset
        }{
            \emptyset \vdash t : A
        }
    \end{minipage}
\end{center}
In this case, $t$ does not change between before and after the last derivation.
Thus, we can obtain the conclusion of the theorem by induction hypothesis.
\\

\item Case (\textsc{der})\\
This case hold trivially because $\Gamma_1, x:\verctype{B}{1}$ cannot be $\emptyset$.
\\

\item Case (\textsc{pr})
\begin{center}
    \begin{minipage}{.29\linewidth}
        \infrule[pr]{
            \emptyset \vdash t : B
            \andalso
            \vdash r
        }{
            \emptyset \vdash \pr{t} : \vertype{r}{B} 
        }
    \end{minipage}
\end{center}
This case holds trivially because \pr{t} is a value.
\\

\item Case (\textsc{ver})
\begin{center}
    \begin{minipage}{.52\linewidth}
        \infrule[ver]{
            \emptyset \vdash t_i : A
            \andalso
            \vdash \{\overline{l_i}\}
        }{
            \emptyset \vdash \nvval{\overline{l_i=t_i}} : \vertype{\{\overline{l_i}\}}{A}
        }
    \end{minipage}
\end{center}
This case holds trivially because $\nvval{\overline{l_i=t_i}}$ is a value.
\\

\item Case (\textsc{veri})
\begin{center}
    \begin{minipage}{.45\linewidth}
        \infrule[veri]{
            \emptyset \vdash t_i : A
            \andalso
            \vdash \{\overline{l_i}\}
            \andalso
            l_k \in \{\overline{l_i}\}
        }{
            \emptyset \vdash \ivval{\overline{l_i=t_i}}{l_k} : A
        }
    \end{minipage}
\end{center}
In this case, we can apply (\textsc{E-veri}).
\begin{center}
\begin{prooftree}
    \AxiomC{$ $}
    \RightLabel{(\textsc{E-veri})}
    \UnaryInfC{$\ivval{\overline{l_i=t_i}}{l_k} \leadsto t_k@l_k$}
    \RightLabel{}
    \UnaryInfC{$\ivval{\overline{l_i=t_i}}{l_k} \longrightarrow t_k@l_k$}
\end{prooftree}
\end{center}
Thus, by choosing $t'=t_k@l_k$, we obtain the conclusion of the theorem.
\\

\item Case (\textsc{extr})
\begin{center}
    \begin{minipage}{.45\linewidth}
        \infrule[extr]{
            \emptyset \vdash t_1 : \vertype{r}{A}
            \andalso
            l_k \in r
        }{
            \emptyset \vdash t_1.l_k : A
        }
    \end{minipage}
\end{center}
In this case, we have two cases whether $t_1$ is a value or not.
\begin{itemize}
\item Suppose $t_1$ is a value. ($t_1=v_1$)\\
By Lemma \ref{lemma:extraction}, we know the following:
\begin{align*}
    \emptyset \,\vdash\, v_1 : \vertype{r}{A}
    \hspace{1em}\Longrightarrow\hspace{1em}
    \exists t'.
    \left\{
    \begin{aligned}
        &v_1.l_k   \longrightarrow  t' \\
        &\emptyset \vdash t' : A
    \end{aligned}
    \right.
\end{align*}
Thus, we obtain the conclusion of the theorem.

\item Suppose $t_1$ is not a value.\\
There exists a term $t_1$ such that:
\begin{align*}
    \begin{minipage}{.13\linewidth}
        \infrule[]{
            t_1 \leadsto t'_1
        }{
            t_1 \longrightarrow t'_1
        }
    \end{minipage}
\end{align*}
Also, we can apply an exaluation rule for extraction to $t$.
\begin{center}
    \begin{minipage}{.15\linewidth}
        \infrule[]{
            t_1 \leadsto t'_1
        }{
            \underbrace{t_1.l_k}_{t} \longrightarrow t'_1.l_k
        }
    \end{minipage}
\end{center}
Thus, by choosing $t'=t_1'.l_k$, we obtain the conclusion of the theorem.\\
\end{itemize}

\item Case (\textsc{sub})
\begin{center}
    \begin{minipage}{.65\linewidth}
        \infrule[\textsc{sub}]{
            \Gamma_1, x:\verctype{B}{r}, \Gamma_2 \vdash t : A
            \andalso
            r \sqsubseteq s
            \andalso
            \vdash s
        }{
            \Gamma_1, x:\verctype{B}{s}, \Gamma_2 \vdash t : A
        }
    \end{minipage}
\end{center}
In this case, $t$ does not change between before and after the last derivation.
Thus, by induction hypothesis, we obtain the conclusion of the theorem.

\end{itemize}
\end{proof}

\clearpage
\section{\vlmini{} Proofs}
\label{appendix:vlmini_safety}

\begin{definition}[Solution of Algorithmic Type Synthesis]
Suppose that $\Sigma;\Gamma \vdash t \Rightarrow A;\Sigma'; \Delta; \Theta; \mathcal{C}$. A solution for this judgement is a pair $(\theta, \eta, B)$ such that $\theta$ satisfies $\Theta$, $\eta$ satisfies $\mathcal{C}$, and $\theta \eta A$ = $B$.
\end{definition}

\begin{definition}[Solution of Pattern Type Synthesis]
Suppose that $\Sigma; R\vdash p : A \rhd \Gamma; \Sigma'; \Theta; \mathcal{C}$. A solution for this judgement is a pair $(\eta, \theta, B)$ such that $\theta$ satisfies $\Theta$, $\eta$ satisfies $\mathcal{C}$, and $\theta \eta A$ = $B$.
\end{definition}

\begin{lemma}[Relation of Resource Well-formedness]
\label{lemma:rel_res_wf}
\begin{align*}
% \begin{aligned}
\vdash \Sigma\;\land\;
\Sigma \vdash \eta\;\land\;%\\
FTV(r) \subseteq \mathrm{dom}(\eta)\;\land\;
\Sigma \vdash r:\labelskind
% \end{aligned}
\quad\Longrightarrow\quad
\vdash \eta r
\end{align*}
\end{lemma}
\begin{proof}
Straightforward by induction on the derivation of $\Sigma \vdash r:\labelskind$.
\end{proof}

\begin{lemma}[Relation of Type Well-formedness]
\label{lemma:rel_ty_wf}
\begin{align*}
\begin{aligned}
\vdash \Sigma\;\land\;
\Sigma \vdash \theta\;\land\;
\Sigma \vdash \eta\;\land\;\\
FTV(A) \subseteq \mathrm{dom}(\theta)\cup\mathrm{dom}(\eta)\;\land\;
\Sigma \vdash A:\typekind
\end{aligned}
\quad\Longrightarrow\quad
\vdash \eta \theta A
\end{align*}
\end{lemma}
\begin{proof}
Straightforward by induction on the derivation of $\Sigma \vdash A:\typekind$. The proof uses Lemma \ref{lemma:rel_res_wf}.
\end{proof}

\begin{lemma}[Relation of Type Environment Well-formedness]
\label{lemma:rel_te_wf}
\begin{align*}
\begin{aligned}
\vdash \Sigma\;\land\;
\Sigma \vdash \theta\;\land\;
\Sigma \vdash \eta\;\land\;\\
FTV(\Gamma) \subseteq \mathrm{dom}(\theta)\cup\mathrm{dom}(\eta)\;\land\;
\Sigma \vdash \Gamma
\end{aligned}
\quad\Longrightarrow\quad
\vdash \eta\theta\Gamma
\end{align*}
\end{lemma}
\begin{proof}
Straightforward by induction on the derivation of $\Sigma \vdash \Gamma$. The proof uses Lemmas \ref{lemma:rel_res_wf} and \ref{lemma:rel_ty_wf}.
\end{proof}

\begin{lemma}[Relation of Resource Environment Well-formedness]
\label{lemma:rel_re_wf}
\begin{align*}
% \begin{aligned}
\vdash \Sigma\;\land\;
\Sigma \vdash \eta\;\land\;
FTV(R) \subseteq \mathrm{dom}(\eta)\;\land\;
\Sigma \vdash R
% \end{aligned}
\quad\Longrightarrow\quad
\vdash \eta R
\end{align*}
\end{lemma}
\begin{proof}
Straightforward by induction on the derivation of $\Sigma \vdash R$. The proof uses Lemma \ref{lemma:rel_res_wf}.
\end{proof}

\begin{lemma}[Soundness of Pattern Type Synthesis]
\label{lemma:sound_p_inf}
\begin{align*}
\left.
\begin{aligned}
    \Sigma; R\vdash p : A \rhd \Gamma; \Sigma'; \Theta; \mathcal{C}\\
    (\theta, \eta, \eta \theta A) \textnormal{ is its solution}
\end{aligned}
\right\}
\quad\Longrightarrow\quad
\eta R\vdash p : \eta\theta A \rhd \eta\theta\Gamma
\end{align*}
\end{lemma}
\begin{proof}
By induction on the derivation of $\Sigma; R\vdash p : A \rhd \Gamma; \Sigma'; \Theta; \mathcal{C}$.
We perform case analysis on the rule applied last to derive $\Sigma; R\vdash p : A \rhd \Gamma; \Sigma'; \Theta; \mathcal{C}$.

\begin{itemize}
\item Case (\textsc{pInt}):
\begin{align*}
    \begin{minipage}{.55\linewidth}
      \infrule[pInt]{
        \vdash \Sigma
        \andalso
        \Sigma \vdash R
        \andalso
        \Sigma \vdash A:\typekind
      }{
        \Sigma; R \vdash n : A \rhd \emptyset; \Sigma; \{A\sim\inttype\}; \top
      }
    \end{minipage}
\end{align*}
We are given
\begin{gather*}
p = n,\quad
\Gamma = \emptyset,\quad
\Theta = \{A \sim \inttype\},\quad
\mathcal{C} = \top
~.
\end{gather*}
Hence, we have $\theta A = \inttype$ and $\eta = \emptyset$, therefore,
\begin{gather*}
\eta \theta A = \inttype, \quad
\eta \theta \Gamma = \emptyset
~.
\end{gather*}
By Lemma \ref{lemma:rel_re_wf},
\begin{gather*}
\vdash \eta R
~.
\end{gather*}
Therefore, by \textsc{pInt},
\begin{align*}
    \begin{minipage}{.55\linewidth}
      \infrule[pInt]{
        \vdash \eta R
      }{
        \eta R \vdash n : \inttype \rhd \emptyset
      }
    \end{minipage}
\end{align*}

\item Case (\textsc{pVar}), and (\textsc{[pVar]})%, (\textsc{p\_})
:\\
Similarly to the case (\textsc{pInt}).
We use 
% lemmas \ref{lemma:rel_ty_wf} and \ref{lemma:rel_re_wf} for the case (\textsc{p\_}), 
lemma \ref{lemma:rel_ty_wf} for the case (\textsc{pVar}), and lemmas \ref{lemma:rel_res_wf} and \ref{lemma:rel_ty_wf} for the case (\textsc{[pVar]}).

\item Case (\textsc{p$\square$}):
\begin{align*}
    \begin{minipage}{.85\linewidth}
      \vspace{0.75\baselineskip}
      \infrule[p$\square$]{
        \Sigma, \alpha:\textsf{Labels}, \beta:\textsf{Type}; \alpha \vdash p' : \beta \rhd \Gamma; \Sigma'; \Theta'; \mathcal{C}
      }{
        \Sigma; - \vdash [p'] : A \rhd \Gamma; \Sigma'; \Theta' \land \{A \sim \vertype{\alpha}{\beta}\}; \mathcal{C}
      }
    \end{minipage}
\end{align*}
We are given
\begin{gather*}
R = -,\quad
p = [p'],\quad
\Theta = \Theta' \land \{A \sim \vertype{\alpha}{\beta}\}
~.
\end{gather*}
Hence, we have $\theta$ unify $\Theta'$, $\theta A = \theta (\vertype{\alpha}{\beta})$, and $\eta \theta A = \eta \theta (\vertype{\alpha}{\beta})$, therefore,
\begin{gather*}
\eta \theta A = \eta \theta (\vertype{\alpha}{\beta}) = \eta \theta A, \quad
\eta R = -
~.
\end{gather*}
Hence, by the induction hypothesis,
\begin{align*}
\eta \alpha \vdash p' : \eta \theta \beta \rhd \eta \theta \Gamma
\end{align*}
Therefore, by \textsc{p$\square$},
\begin{align*}
    \begin{minipage}{.85\linewidth}
      \vspace{0.75\baselineskip}
      \infrule[p$\square$]{
        \eta \alpha \vdash p' : \eta \theta \beta \rhd \eta \theta \Gamma
      }{
        - \vdash [p'] : \vertype{\eta \alpha}{\eta \theta \beta} \rhd \eta \theta \Gamma
      }
    \end{minipage}
\end{align*}
Since $\theta$ does not include type substituions for resource variables, $\eta \alpha = \eta \theta \alpha$, hence $\vertype{\eta \alpha}{\eta \theta \beta} = \eta \theta (\vertype{\alpha}{\beta}) = \eta \theta A$. Therefore, we get the conclusion for this case.

% \item Case (\textsc{[p$\square$]}):\\
% Similarly to the case ([p$\square$]).
\end{itemize}

\end{proof}

\begin{theorem}[Soundness of Algorithmic Type Synthesis]
\label{lemma:sound_ty_inf}
\begin{align*}
\left.
\begin{aligned}
    \Sigma;\Gamma \vdash t \Rightarrow A;\Sigma'; \Delta; \Theta; \mathcal{C}\\
    (\theta, \eta, \eta \theta A) \textnormal{ is its solution}
\end{aligned}
\right\}
\quad\Longrightarrow\quad
\eta\theta\Delta \vdash t : \eta\theta A
\end{align*}
\end{theorem}

\begin{proof}
By induction on the derivation of $\Sigma;\Gamma \vdash t \Rightarrow A;\Sigma'; \Delta; \Theta; \mathcal{C}$.
We perform case analysis on the rule applied last to derive $\Sigma;\Gamma \vdash t \Rightarrow A;\Sigma'; \Delta; \Theta; \mathcal{C}$.

\begin{itemize}
\item Case (\textsc{$\Rightarrow_{\textsc{int}}$}):
\begin{align*}
    \begin{minipage}{.45\linewidth}
      \infrule[$\Rightarrow_{\textsc{int}}$]{
        \vdash \Sigma
        \andalso
        \Sigma \vdash \Gamma
      }{
        \Sigma; \Gamma \vdash n \Rightarrow \inttype; \Sigma; \emptyset; \top; \top
      }
    \end{minipage}
\end{align*}
We are given
\begin{gather*}
t = n,\quad
A = \inttype,\quad
\Delta = \emptyset,\quad
\Theta = \top,\quad
\mathcal{C} = \top
~.
\end{gather*}
Hence, we have $\theta = \emptyset$ and $\eta = \emptyset$, therefore,
\begin{gather*}
\eta \theta \Delta = \emptyset, \quad
\eta \theta A = \inttype
~.
\end{gather*}
% Furthermore, by Lemmas 
Therefore, by (\textsc{int}),
\begin{align*}
    \begin{minipage}{.30\linewidth}
      \infrule[int]{
        \\
      }{
        \emptyset \vdash n : \inttype
      }
    \end{minipage}
\end{align*}

\item Case (\textsc{$\Rightarrow_{\textsc{lin}}$}):
\begin{align*}
    \begin{minipage}{.5\linewidth}
      \infrule[$\Rightarrow_{\textsc{lin}}$]{
        \vdash \Sigma
        \andalso
        \Sigma \vdash \Gamma
        \andalso
        x:A\in\Gamma
      }{
        \Sigma; \Gamma \vdash x \Rightarrow A; \Sigma; x:A; \top; \top
      }
    \end{minipage}
\end{align*}
We are given
\begin{gather*}
t = x,\quad
% A = \inttype,\quad
\Delta = x:A,\quad
\Theta = \top,\quad
\mathcal{C} = \top
~.
\end{gather*}
Hence, we have $\theta = \emptyset$ and $\eta = \emptyset$, therefore,
\begin{gather*}
\eta \theta \Delta = x:A, \quad
\eta \theta A = A
~.
\end{gather*}
Furthermore, by Lemma \ref{lemma:rel_ty_wf}, we have
\begin{align*}
\vdash \eta \theta A~(= A)
~.
\end{align*}
Therefore, by (\textsc{var}),
\begin{align*}
    \begin{minipage}{.30\linewidth}
      \infrule[var]{
        \vdash A\\
      }{
        x:A \vdash x:A
      }
    \end{minipage}
\end{align*}

\item Case (\textsc{$\Rightarrow_{\textsc{gr}}$}):
\begin{align*}
    \begin{minipage}{.55\linewidth}
      \infrule[$\Rightarrow_{\textsc{gr}}$]{
        \vdash \Sigma
        \andalso
        \Sigma \vdash \Gamma
        \andalso
        x:\verctype{A}{r}\in\Gamma
      }{
        \Sigma; \Gamma \vdash x \Rightarrow A; \Sigma; x:\verctype{A}{1}; \top; \top
      }
    \end{minipage}
\end{align*}
We are given
\begin{gather*}
t = x,\quad
% A = \inttype,\quad
\Delta = x:\verctype{A}{1},\quad
\Theta = \top,\quad
\mathcal{C} = \top
~.
\end{gather*}
Hence, we have $\theta = \emptyset$ and $\eta = \emptyset$, therefore,
\begin{gather*}
\eta \theta \Delta = x:\verctype{A}{1}, \quad
\eta \theta A = A
~.
\end{gather*}
Furthermore, by Lemma \ref{lemma:rel_ty_wf}, we have
\begin{align*}
\vdash \eta \theta A~(= A)
~.
\end{align*}
Therefore, we conclude the case by the following derivation.
\begin{center}
\begin{prooftree}
\AxiomC{$ $}
\RightLabel{(\textsc{var})}
\UnaryInfC{$ x:A \vdash x : A$}
\RightLabel{(\textsc{der})}
\UnaryInfC{$ x:\verctype{A}{1} \vdash x:A$}
\end{prooftree}
\end{center}

\item Case (\textsc{$\Rightarrow_{\textsc{abs}}$}):
\begin{align*}
    \begin{minipage}{.95\linewidth}
      \infrule[$\Rightarrow_{\textsc{abs}}$]{
        \Sigma_1, \alpha:\textsf{Type};- \vdash p:\alpha \rhd \Gamma'; \Sigma_2; \Theta_1
        \andalso\\
        \Sigma_2;\Gamma,\Gamma' \vdash t \Rightarrow B;\Sigma_3;\Delta'; \Theta_2; \mathcal{C}
      }{
        \Sigma_1;\Gamma \vdash \lam{p}{t} \Rightarrow \ftype{\alpha}{B};\Sigma_3;\Delta'\backslash\Gamma' ; \Theta_1\land\Theta_2; \mathcal{C}
      }
    \end{minipage}
\end{align*}
We are given
\begin{gather*}
t = \ftype{\alpha}{B},\quad
% A = \inttype,\quad
\Delta = \Delta'\backslash\Gamma',\quad
\Theta = \Theta_1\land\Theta_2,
~.
\end{gather*}
Hence, we have $\theta$ unifies $\Theta_1$ and $\Theta_2$, $\eta$ unifies $\mathcal{C}$. Therefore, we have
\begin{gather*}
\eta \theta \Delta = \eta \theta (\Delta'\backslash\Gamma'), \quad
\eta \theta A = \eta \theta (\ftype{\alpha}{B})
~.
\end{gather*}
Furthermore, by Lemma \ref{lemma:sound_p_inf} and the induction hypothesis, we have
\begin{align*}
- \vdash p:\eta\theta\alpha \rhd \eta\theta\Gamma',\quad
\eta\theta \Delta' \vdash t : \eta \theta B
~.
\end{align*}
Therefore, by (\textsc{abs}),
\begin{align*}
    \begin{minipage}{.60\linewidth}
      \infrule[abs]{
        - \vdash p:\eta\theta\alpha \rhd \eta\theta\Gamma'
        \andalso
        \eta\theta \Delta' \vdash t : \eta \theta B
      }{
        \eta\theta\Delta' \backslash \eta\theta \Gamma' \vdash \lam{p}{t}:\ftype{\eta\theta\alpha}{\eta \theta B}
      }
    \end{minipage}
\end{align*}
Since $\eta\theta\Delta' \backslash \eta\theta \Gamma' = \eta\theta(\Delta' \backslash \Gamma')$ and $\ftype{\eta\theta\alpha}{\eta \theta B} = \eta\theta(\ftype{\alpha}{B})$, we have the conclusion of the case.

\item Case (\textsc{$\Rightarrow_{\textsc{app}}$}):
\begin{align*}
    \begin{minipage}{.95\linewidth}
      \vspace{0.5\baselineskip}
      \infrule[$\Rightarrow_{\textsc{app}}$]{
        \Sigma_1; \Gamma \vdash t_1 \Rightarrow A_1 ; \Sigma_2; \Delta_1; \Theta_1; \mathcal{C}_1
        \andalso\\
        \Sigma_2; \Gamma \vdash t_2 \Rightarrow A_2; \Sigma_3; \Delta_2; \Theta_2; \mathcal{C}_2
      }{
        \Sigma_1;\Gamma \vdash \app{t_1}{t_2} \Rightarrow \beta; \Sigma_3, \beta:\typekind; \Delta_1+\Delta_2; \\\hspace{11em}\Theta_1\land\Theta_2\land\{A_1\sim \ftype{A_2}{\beta}\}; \mathcal{C}_1 \land \mathcal{C}_2
      }
    \end{minipage}
\end{align*}
We are given
\begin{gather*}
t = \app{t_1}{t_2},\quad
A = \beta,\quad
\Delta = \Delta_2+\Delta_2,\\
\Theta = \Theta_1\land\Theta_2\land\{A_1\sim \ftype{A_2}{\beta}\},\quad
\mathcal{C} = \mathcal{C}_1 \land\mathcal{C}_2
~.
\end{gather*}
Hence, we have $\theta$ unifies $\Theta_1$ and $\Theta_2$, and $\theta A_1 = \theta (\ftype{A_2}{\beta})$.
Also, $\eta$ unifies $\mathcal{C}_1$ and $\mathcal{C}_2$, and then we have $\eta \theta A_1 = \eta \theta (\ftype{A_2}{\beta}) = \ftype{\eta \theta A_2}{\eta \theta\beta}$.
Therefore, we have
\begin{gather*}
\eta \theta \Delta = \eta\theta(\Delta_2+\Delta_2),\quad
\eta \theta A = \eta \theta \beta
~.
\end{gather*}
Furthermore, by the induction hypothesises,
\begin{gather*}
\eta \theta \Gamma \vdash t_1 : \eta \theta A_1~(=\ftype{\eta \theta A_2}{\eta \theta\beta}), \quad
\eta \theta \Gamma \vdash t_2 : \eta \theta A_2
~.
\end{gather*}
Therefore, by (\textsc{app}),
\begin{align*}
    \begin{minipage}{.70\linewidth}
      \infrule[app]{
        \eta \theta \Gamma \vdash t_1 : \ftype{\eta \theta A_2}{\eta \theta\beta}
        \andalso
        \eta \theta \Gamma \vdash t_2 : \eta \theta A_2
      }{
        \eta \theta \Gamma \vdash \app{t_1}{t_2} : \eta \theta \beta
      }
    \end{minipage}
\end{align*}

\item Case (\textsc{$\Rightarrow_{\textsc{pr}}$}):
\begin{align*}
    \begin{minipage}{.95\linewidth}
      \vspace{0.5\baselineskip}
      \infrule[$\Rightarrow_{\textsc{pr}}$]{
        \Sigma_1 \vdash [\Gamma\cap\textsf{FV}(t)]_{\textsf{Labels}} \rhd \Gamma'
        \andalso
        \Sigma_1; \Gamma' \vdash t \Rightarrow A'; \Sigma_2; \Delta'; \Theta; \mathcal{C}_1
        \andalso\\
        \Sigma_3 = \Sigma_2, \alpha:\textsf{Labels}
        \andalso
        \Sigma_3 \vdash \alpha \sqsubseteq_c \Gamma' \rhd \mathcal{C}_2
      }{
        \Sigma_1;\Gamma \vdash [t] \Rightarrow \vertype{\alpha}{A'}; \Sigma_3; \alpha \cdot \Delta' ;\Theta; \mathcal{C}_1 \land \mathcal{C}_2
      }
    \end{minipage}
\end{align*}
We are given
\begin{gather*}
t = \pr{t},\quad
A = \vertype{\alpha}{A'},\quad
\Delta = \alpha \cdot \Delta',\quad
\mathcal{C} = \mathcal{C}_1\land\mathcal{C}_2,
~.
\end{gather*}
Hence, $\theta$ unifies $\Theta$ and $\eta$ unifies $\mathcal{C}_1$ and $\mathcal{C}_2$. Therefore, we have
\begin{gather*}
\eta \theta \Delta = \eta \theta (\alpha \cdot \Delta') = (\eta \theta\alpha) \cdot (\eta \theta\Delta'), \quad
\eta\theta A = \eta \theta (\vertype{\alpha}{A'}) = \vertype{\eta \theta\alpha}{\eta \theta A'}
~.
\end{gather*}
Furthermore, by %Lemma \ref{lemma:sound_p_inf} and 
the induction hypothesis, we have
\begin{align*}
\eta \theta \Delta' \vdash t : \eta \theta A'
~.
\end{align*}
Therefore, by (\textsc{pr}),
\begin{align*}
    \begin{minipage}{.55\linewidth}
        \infrule[pr]{
            \eta \theta \Delta' \vdash t : \eta \theta A'
            \andalso
            \vdash \eta \theta \alpha
        }{
            (\eta \theta \alpha)\cdot(\eta \theta \Delta') \vdash \pr{t} : \vertype{\eta \theta \alpha}{\eta \theta A'} 
        }
    \end{minipage}
\end{align*}
\end{itemize}
\end{proof}
\clearpage

\end{document}